\documentclass[a4paper, USenglish, cleveref, autoref]{lipics-v2019}
\usepackage[draft,backgroundcolor=white!50!orange]{todonotes}

\usepackage{dsfont}
\usepackage{marvosym}
\usepackage{tcolorbox}
\usepackage{mathtools}
\usepackage{color, colortbl}
\usepackage{ wasysym }
\usepackage{algorithm}
\usepackage{algorithmicx }
\usepackage{algpseudocode}

\definecolor{myYellow}{RGB}{251, 199, 19}
\colorlet{myDGray}{black!70}
\colorlet{myLGray}{black!50}
\colorlet{myTGray}{black!75}

\usepackage{tikz}
\usetikzlibrary{arrows, petri, backgrounds, positioning, shapes, patterns, calc,decorations.pathmorphing}

\tikzset{every path/.style={->,thick}}
\tikzstyle{envplace}=[circle,thick,draw=black!75,fill=white,minimum size=15pt,transform shape]
\tikzstyle{sysplace}=[circle,thick,draw=black!75,fill=black!20,minimum size=15pt,transform shape]
\tikzstyle{transition}=[rectangle,thick,draw=black!75,fill=white,minimum size=10pt,transform shape]
\tikzstyle{token}=[circle,thick,fill=black!75, inner sep=1.5pt,transform shape]
\tikzstyle{specialSys}=[double=black!20,minimum size=13.5pt]
\tikzstyle{specialEnv}=[double,minimum size=13.5pt]

\tikzstyle{aastate}=[regular polygon,regular polygon sides=8,thick,fill=white,draw=black!75,minimum size=13pt,transform shape]
\tikzstyle{specialAA}=[regular polygon,regular polygon sides=8,thick,double,fill=white,draw=black!75,minimum size=0pt,transform shape]

\tikzstyle{arrow}=[->,thick,black!75]
\tikzstyle{darrow}=[<->,thick,black!75]

\newlength{\bls}
\newcommand{\refLemma}[1]{Lemma~\ref{lem:#1}}
\newcommand{\refTheo}[1]{Theorem~\ref{theo:#1}}

\newcommand{\refFig}[1]{Fig.~\ref{fig:#1}}
\newcommand{\refSection}[1]{Sec.~\ref{sec:#1}}

\newcommand{\refProp}[1]{Proposition~\ref{prop:#1}}

\newcommand{\abs}[1]{\ensuremath | #1 |}

\newcommand{\pNet}{\ensuremath\mathcal{N}}
\newcommand{\pl}{\ensuremath\mathcal{P}}
\newcommand{\places}{\pl}
\newcommand{\tr}{\ensuremath\mathcal{T}}
\newcommand{\transitions}{\tr}
\newcommand{\fl}{\ensuremath\mathcal{F}}
\newcommand{\flow}{\fl}
\newcommand{\init}{\ensuremath\mathit{In}}
\newcommand{\win}{\ensuremath\mathcal{W}}

\newcommand{\pre}[2]{\mathit{pre}^{#1}(#2)}
\newcommand{\post}[2]{\mathit{post}^{#1}(#2)}
\newcommand{\petriNet}{\ensuremath\pNet=(\pl,\tr,\fl,\init)}
\newcommand{\conflict}[2]{{#1}\,\sharp\,{#2}}

\newcommand{\fireSeq}[2]{\ensuremath #1[\bigtriangledown\, #2 \,]}
\newcommand{\fireTranTo}[3]{\ensuremath #1  \;\boldsymbol{[} \, {\scriptstyle #2}\, \boldsymbol{\rangle}\; #3}

\newcommand{\markpro}[2]{\ensuremath {#1}^{\langle {#2} \rangle} }
\newcommand{\pGame}{\ensuremath\mathcal{G}}
\newcommand{\plS}{\ensuremath\pl_\psys}
\newcommand{\plE}{\ensuremath\pl_\penv}
\newcommand{\bad}{\ensuremath\mathcal{B}}
\newcommand{\reach}{\ensuremath\mathcal{R}}
\newcommand{\pSpecial}{\ensuremath\mathcal{S}}

\newcommand{\slices}{\ensuremath \boldsymbol{\mathscr{S}}}
\newcommand{\processes}{\ensuremath \boldsymbol{\mathscr{P}}}

\newcommand{\slice}{\ensuremath \varsigma}
\newcommand{\ind}{\ensuremath \mathds{I}}

\newcommand{\penv}{\ensuremath \mathcal{E}}
\newcommand{\psys}{\ensuremath \mathcal{S}}

\newcommand{\pts}[1]{\ensuremath \slices(#1)}
\newcommand{\stp}[1]{\ensuremath \processes(#1)}

\newcommand{\projtr}[1]{\ensuremath \langle #1 \rangle_{\downarrow}^\transitions}

\newcommand{\DL}{\ensuremath \mathit{DL}}
\newcommand{\badw}{\ensuremath \bot}
\newcommand{\env}{\ensuremath \mathit{env}}
\newcommand{\sys}{\ensuremath \mathit{sys}}
\newcommand{\dom}{\ensuremath \mathit{dom}}
\newcommand{\Seq}{\ensuremath \mathit{seq}}
\newcommand{\past}{\ensuremath \mathit{past}}
\newcommand{\domain}{\ensuremath \mathit{domain}}
\newcommand{\Plays}{\ensuremath \mathit{Plays}}
\newcommand{\PlaysInf}{\ensuremath \mathit{Plays}^F}
\newcommand{\view}{\ensuremath \mathit{view}}

\newcommand{\CG}{\ensuremath \cGame}
\newcommand{\cGame}{\ensuremath\mathcal{C}}

\renewcommand{\AA}{\ensuremath\mathcal{A}}

\newcommand{\enabled}[1]{\ensuremath act(#1)}
\newcommand{\state}[1]{\ensuremath \mathit{state}(#1)}
\newcommand{\statep}[2]{\ensuremath \mathit{state}_{#1}(#2)}
\newcommand{\tuplein}{\ensuremath  \in }

\newcommand{\myItem}[1]{\textbf{\color{black}\sffamily #1}}
\newcommand{\inv}[1]{#1\raisebox{1.15ex}{$\scriptscriptstyle-\!1$}}

\newcommand{\relation}{\ensuremath \approx_\mathfrak{B}}
\newcommand{\pastt}[1]{\ensuremath \past_{\transitions}(#1)}
\newcommand{\rec}{\ensuremath \mathit{rec}}
\newcommand{\recv}{\ensuremath \mathit{rec}_{\varrho}}

\newcommand{\indx}{\ensuremath \mathfrak{I}}
\newcommand{\parcomp}{\ensuremath {\langle \parallel_{\slices} \rangle}}

\tikzset{every loop/.style={min distance=5mm}}

\newcommand{\sysstrat}{\ensuremath{\sigma}}
\graphicspath{{./graphics/}}%helpful if your graphic files are in another directory

\title{Translating Asynchronous Games for Distributed Synthesis (Full Version)}

\titlerunning{Translating Asynchronous Games for Distributed Synthesis (Full Version)}%optional, please use if title is longer than one line

\author{Raven Beutner}{Saarland University, Germany}{}{}{}% mandatory, please use full name; only 1 author per \author macro; first two parameters are mandatory, other parameters can be empty. Please provide at least the name of the affiliation and the country. The full address is optional

\author{Bernd Finkbeiner}{Saarland University, Germany}{}{}{}

\author{Jesko Hecking-Harbusch}{Saarland University, Germany}{}{}{}

\authorrunning{R.\ Beutner, B.\ Finkbeiner, and J.\ Hecking-Harbusch}% mandatory. First: Use abbreviated first/middle names. Second (only in severe cases): Use first author plus 'et al.'

\Copyright{Raven Beutner, Bernd Finkbeiner, and Jesko Hecking-Harbusch}% mandatory, please use full first names. LIPIcs license is "CC-BY";  http://creativecommons.org/licenses/by/3.0/

%mandatory: Please choose ACM 2012 classifications from https://dl.acm.org/ccs/ccs_flat.cfm

%mandatory; please add comma-separated list of keywords

\category{}%optional, e.g. invited paper

\relatedversion{Conference version: \url{https://doi.org/10.4230/LIPIcs.CONCUR.2019.26}~\cite{DBLP:conf/concur/BeutnerFH19}.}

\supplement{}%optional, e.g. related research data, source code, ... hosted on a repository like zenodo, figshare, GitHub, ...

\funding{Supported by the German Research Foundation (DFG) Grant Petri Games (392735815) and the Collaborative Research Center “Foundations of Perspicuous Software Systems” (TRR 248, 389792660), and by the European Research Council (ERC) Grant OSARES (683300).}%optional, to capture a funding statement, which applies to all authors. Please enter author specific funding statements as fifth argument of the \author macro.

\acknowledgements{}%optional

\nolinenumbers %uncomment to disable line numbering

\hideLIPIcs  %uncomment to remove references to LIPIcs series (logo, DOI, ...), e.g. when preparing a pre-final version to be uploaded to arXiv or another public repository

\EventEditors{Wan Fokkink and Rob van Glabbeek}
\EventNoEds{2}
\EventLongTitle{30th International Conference on Concurrency Theory (CONCUR 2019)}
\EventShortTitle{CONCUR 2019}
\EventAcronym{CONCUR}
\EventYear{2019}
\EventDate{August 27--30, 2019}
\EventLocation{Amsterdam, the Netherlands}
\EventLogo{}
\SeriesVolume{140}
\ArticleNo{22}
\theoremstyle{plain}

\begin{document}

\maketitle

\begin{abstract}
  In distributed synthesis, a set of process
  implementations is generated, which together, accomplish an objective against all
  possible behaviors of the environment. A lot of recent work has
  focussed on systems with causal memory,
  i.e., sets of asynchronous processes that exchange their causal histories upon
  synchronization. Decidability results for this problem have been stated either in
  terms of control games, which extend Zielonka's asynchronous
  automata by partitioning the actions into controllable and
  uncontrollable, or in terms of Petri games, which extend Petri nets
  by partitioning the tokens into system and environment players.  The
  precise connection between these two models was so far, however, an open
  question.

  In this paper, we provide the first formal connection between
  control games and Petri games.  We establish the equivalence of the
  two game types based on weak bisimulations between their strategies. For
  both directions, we show that a game of one type can be translated
  into an equivalent game of the other type.  We provide exponential
  upper and lower bounds for the translations. Our translations allow
  to transfer and combine decidability results between the
  two types of games. Exemplarily, we translate decidability in
  acyclic communication architectures, originally obtained for control
  games, to Petri games, and decidability in single-process systems, originally obtained for
  Petri games, to control games.
\end{abstract}

\section{Introduction}
\label{sec:introduction}
Synthesis is the task of automatically generating an implementation fulfilling a given objective or proving that no such implementation can exist.
Synthesis can be viewed as a game between the \emph{system} and the \emph{environment} with winning strategies for the system being correct implementations~\cite{buchi1967solving}.
We call a class of games \emph{decidable} if we can determine the existence of a winning strategy.
A distributed system consists of local processes, that possess incomplete information about the global system state.
\emph{Distributed synthesis} searches for distributed strategies that govern the local processes such that the system as a whole satisfies an objective, independently of the inputs that are received from the environment.

After some early results on \emph{synchronous} distributed systems~\cite{DBLP:conf/focs/PnueliR90}, most work has focussed on the synthesis of \emph{asynchronous} distributed systems with \emph{causal memory}~\cite{DBLP:conf/fsttcs/GastinLZ04,DBLP:conf/icalp/GenestGMW13,DBLP:conf/fsttcs/MuschollW14,DBLP:conf/fsttcs/Gimbert17,DBLP:journals/iandc/FinkbeinerO17,DBLP:conf/fsttcs/FinkbeinerG17}.
Causal memory means that two processes share no information while they run independently; during every synchronization, however, they exchange their complete local histories.
The study of the synthesis problem with causal memory has, so far, been carried out, independently of each other, in two different models: control games and Petri games.

\noindent\textbf{\textsf{Control Games and Petri Games~~}}
\emph{Control games}~\cite{DBLP:conf/icalp/GenestGMW13} are based on Zielonka's asynchronous automata~\cite{DBLP:journals/ita/Zielonka87}, which are compositions of local processes.
The actions of the asynchronous automaton are partitioned as either controllable or uncontrollable.
Hence, each process can have both controllable and uncontrollable behavior.
A strategy comprises a family of one individual controller for each process that can restrict controllable actions based on the causal past of the process but has to account for all uncontrollable actions.
Together, the local controllers aim to fulfill an objective against all possible unrestricted behavior. 
There are non-elementary decidability results for acyclic communication architectures~\cite{DBLP:conf/icalp/GenestGMW13,DBLP:conf/fsttcs/MuschollW14}.
Decidability has also been obtained for restrictions on the dependencies of actions~\cite{DBLP:conf/fsttcs/GastinLZ04} or on the synchronization behavior~\cite{DBLP:conf/concur/MadhusudanT02,DBLP:conf/fsttcs/MadhusudanTY05} and, recently, for decomposable games~\cite{DBLP:conf/fsttcs/Gimbert17}.

\emph{Petri games}~\cite{DBLP:journals/iandc/FinkbeinerO17} are based on Petri nets.
They partition the places of the underlying Petri net into system places and environment places and, thereby, group the tokens into system players and environment players. 
For tokens in system places, the outgoing transitions can be restricted by the strategy whereas tokens in environment places cannot be controlled, i.e., every possible transition has to be accounted for.
Strategies are defined as restrictions of the global unfolding and aim to fulfill an objective against all possible unrestricted behavior.
Petri games are \textsc{EXPTIME}-complete for a bounded number of system players and one environment player~\cite{DBLP:journals/iandc/FinkbeinerO17} as well as for one system player and a bounded number of environment players~\cite{DBLP:conf/fsttcs/FinkbeinerG17}.
Both models are based on causal information: Control games utilize local views whereas Petri games utilize unfoldings.

\noindent\textbf{\textsf{Translations~~}}
The precise connection between control games and Petri games, and hence, the question whether results can be transferred between them, was, so far, open.
We translate control games into Petri games, and vice versa.
Both game types admit strategies based on causal information but the formalisms for the possibilities of system and environment differ.
In control games, an action is either controllable or uncontrollable and therefore can be restricted by either all or none of the involved players.
From the same state of a process, both controllable and uncontrollable behavior is possible.
By contrast, Petri games utilize a partitioning into system and environment places. 
While this offers more precise information about which player can control a shared transition, a given place can no longer comprise both system and environment behavior. 
The challenge is to resolve the controllability while preserving the causal information in the game.
For both translations, we adopt the concept of \emph{commitment sets}: 
The local players do not enable behavior directly but move to a state or place that explicitly encodes their decision of what to enable. 
Using this explicit representation, we can express the controllability aspects of one game in the respective other one, i.e., make actions in a control game controllable by only a subset of players and allow places in Petri games that comprise both environment and system behavior.

Our translations preserve the structure of winning strategies in a weak bisimilar way. 
In addition to the upper bounds established by our exponential translations, we provide matching lower bounds.
The translations show that contrasting formalisms can be overcome whereas our lower bounds highlight an intrinsic difficultly to achieve this. 
The equivalence of both models, as witnessed by our results, gives rise to more practical applications by allowing the transfer of existing decidability results between both models. 
As an example, we can transfer decidability of single-process systems for Petri games~\cite{DBLP:conf/fsttcs/FinkbeinerG17} to control games and decidability for acyclic communication architectures for control games~\cite{DBLP:conf/icalp/GenestGMW13} to Petri games.

\section{Examples}
\label{sec:motivation}

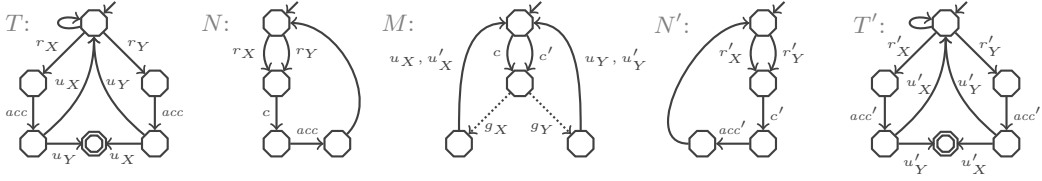
\begin{figure}[t!]
	\centering
	\begin{tikzpicture}[scale=0.8, every label/.append style={font=\tiny}, label distance=-1mm]
	\node[aastate] at (0,0) (s11){};
	\node[aastate] at (-1,-1) (s12){};
	\node[aastate] at (1,-1) (s13){};
	\node[aastate] at (-1,-2) (s14){};
	\node[aastate] at (1,-2) (s15){};
	
	\node[specialAA] at (0,-2) (s16){};
	
	\draw[arrow] (s11)+(0.35, 0.35) to (s11);
	
	\draw[arrow, loop left] (s11) to (s11);
	
	\draw[arrow] (s11) to node[left=-1mm,yshift=1mm] {\tiny$r_X$} (s12);
	\draw[arrow] (s11) to node[right=-1mm,yshift=1mm] {\tiny$r_Y$} (s13);
	\draw[arrow] (s12) to node[left=-0.5mm] {\tiny$\mathit{acc}$} (s14);
	\draw[arrow] (s13) to node[right=-0.5mm] {\tiny$\mathit{acc}$} (s15);
	
	\draw[arrow] (s14) to[out=40,in=270] node[left=-1.5mm,yshift=1mm] {\tiny$u_X$} (s11);
	\draw[arrow] (s15) to[out=140,in=270] node[right=-1.5mm,yshift=1mm] {\tiny$u_Y$} (s11);
	
	\draw[arrow] (s14) to[out=0,in=180] node[below] {\tiny$u_Y$} (s16);
	\draw[arrow] (s15) to[out=180,in=0] node[below] {\tiny$u_X$} (s16);
	
	\node[] at (-1.25,0) () {\small\color{gray}$T$:};

	%%%%%%%%%%%%%%%%%%%%%%%%%%%%%%%%%%%%%%%%%%%%%%%%%%%%%%%%%%%%%%
	
	\node[aastate] at (3,0) (s21){};
	\node[aastate] at (3,-1) (s22){};
	\node[aastate] at (3,-2) (s23){};
	\node[aastate] at (4,-2) (s24){};
	
	\draw[arrow] (s21)+(0.35, 0.35) to (s21);
	
	\draw[arrow] (s21) to[out=240,in=120] node[left=-0.5mm] {\tiny$r_X$} (s22);
	\draw[arrow] (s21) to[out=300,in=60] node[right=-0.5mm] {\tiny$r_Y$} (s22);
	\draw[arrow] (s22) to node[left=-0.5mm] {\tiny$c$} (s23);
	\draw[arrow] (s23) to node[above] {\tiny$\mathit{acc}$} (s24);
	\draw[arrow] (s24) to[out=45,in=0] (s21);
	
	\node[] at (2,0) () {\small\color{gray}$N$:};
	
	%%%%%%%%%%%%%%%%%%%%%%%%%%%%%%%%%%%%%%%%%%%%%%%%%%%%%%%%%%
	
	\node[aastate] at (7,-0) (s1){};
	\node[aastate] at (7,-1) (s2){};
	\node[aastate] at (6,-2) (s3){};
	\node[aastate] at (8,-2) (s4){};
	
	\draw[arrow] (s1)+(0.35, 0.35) to (s1);
	
	\draw[arrow] (s1) to[out=240,in=120] node[left=-0.5mm] {\tiny$c$} (s2);
	\draw[arrow] (s1) to[out=300,in=60] node[right=-0.5mm] {\tiny$c'$} (s2);
	
	\draw[arrow, densely dotted] (s2) to node[below,xshift=1mm] {\tiny$g_X$} (s3);
	\draw[arrow, densely dotted] (s2) to node[below,xshift=-1mm] {\tiny$g_Y$} (s4);
	
	\draw[arrow] (s3) to[out=90,in=180] node[left] {\tiny$u_X, u'_X$} (s1);
	\draw[arrow] (s4) to[out=90,in=0] node[right] {\tiny$u_Y, u'_Y$} (s1);
	
	\node[] at (5,0) () {\small\color{gray}$M$:};
	
	%%%%%%%%%%%%%%%%%%%%%%%%%%%%%%%%%%%%%%%%%%%%%%%%%%%%%%%%%%%%%%%%%%%%%%%%%%
	
	\node[aastate] at (11,0) (s31){};
	\node[aastate] at (11,-1) (s32){};
	\node[aastate] at (11,-2) (s33){};
	\node[aastate] at (10,-2) (s34){};
	
	\draw[arrow] (s31)+(0.35, 0.35) to (s31);
	
	\draw[arrow] (s31) to[out=240,in=120] node[left=-0.5mm] {\tiny$r'_X$} (s32);
	\draw[arrow] (s31) to[out=300,in=60] node[right=-0.5mm] {\tiny$r'_Y$} (s32);
	\draw[arrow] (s32) to node[right=-0.5mm] {\tiny$c'$} (s33);
	\draw[arrow] (s33) to node[above] {\tiny$\mathit{acc}'$} (s34);
	\draw[arrow] (s34) to[out=180,in=180] (s31);
	
	\node[] at (9.5,0) () {\small\color{gray}$N'$:};
	%%%%%%%%%%%%%%%%%%%%%%%%%%%%%%%%%%%%%%%%%%%%%%%%%%%%%%%%%%%%%%%%%%%%%%%%%%%%

	\node[aastate] at (14,0) (s41){};
	\node[aastate] at (13,-1) (s42){};
	\node[aastate] at (15,-1) (s43){};
	\node[aastate] at (13,-2) (s44){};
	\node[aastate] at (15,-2) (s45){};
	
	\node[specialAA] at (14,-2) (s46){};
	
	\draw[arrow] (s41)+(0.35, 0.35) to (s41);
	
	\draw[arrow, loop left] (s41) to (s41);
	
	\draw[arrow] (s41) to node[left=-1mm,yshift=1mm] {\tiny$r'_X$} (s42);
	\draw[arrow] (s41) to node[right=-1mm,yshift=1mm] {\tiny$r'_Y$} (s43);
	\draw[arrow] (s42) to node[left=-0.5mm] {\tiny$\mathit{acc}'$} (s44);
	\draw[arrow] (s43) to node[right=-0.5mm] {\tiny$\mathit{acc}'$} (s45);
	
	\draw[arrow] (s44) to[out=40,in=270] node[left=-1.5mm,yshift=1mm] {\tiny$u'_X$} (s41);
	\draw[arrow] (s45) to[out=140,in=270] node[right=-1.5mm,yshift=1mm] {\tiny$u'_Y$} (s41);
	
	\draw[arrow] (s44) to[out=0,in=180] node[below] {\tiny$u'_Y$} (s46);
	\draw[arrow] (s45) to[out=180,in=0] node[below] {\tiny$u'_X$} (s46);
	
	\node[] at (12.75,0) () {\small\color{gray}$T'$:};
	\end{tikzpicture}
	\caption{A control game for a manager $M$ of resources $X$ and $Y$ between threads $T$ and $T'$ with networks $N$ and $N'$ is depicted.
		Communication occurs by synchronization on shared actions.
		Dotted actions are controllable, all others are uncontrollable.
		Losing states are double circles.}
	\label{fig:exp-aa}
\end{figure}

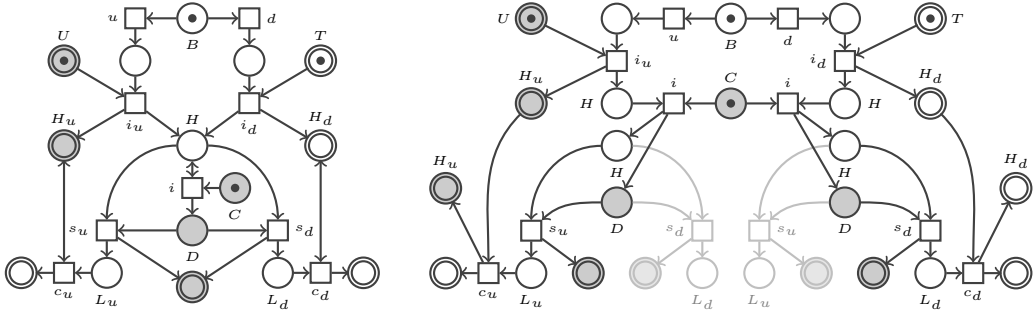
\begin{figure}[t!]
	\begin{subfigure}[c]{0.37\textwidth}
		\centering
		\begin{tikzpicture}[scale=0.75, every label/.append style={font=\tiny}, label distance=-0.5mm]
		\node[envplace,label=south:\color{black}$B$] at (0,-0.75) (p1){};
		%\node[blueNode] at (0,0) (){};
		\node[envplace] at (-1,-1.5) (p2){};
		%\node[blueNode] at (-1,-1.5) (){};
		\node[envplace] at (1,-1.5) (p3){};
		%\node[blueNode] at (1,-1.5) (){};
		
		\node[sysplace,specialSys,label=north:\color{black}$U$] at (-2.25,-1.5) (p4){};
		\node[envplace,specialEnv,label=north:\color{black}$T$] at (2.25,-1.5) (p5){};
		\node[sysplace, specialSys,label=north:\color{black}$H_u$] at (-2.25,-3) (p6){};
		\node[envplace,specialEnv,label=north:\color{black}$H_d$] at (2.25,-3) (p7){};
		
		\node[envplace,label=north:\color{black}$H$] at (0,-3) (p8){};
		%\node[blueNode] at (0,-3) (){};
		
		\node[sysplace,label=south:\color{black}$C$] at (0.75,-3.75) (p9){};
		\node[sysplace,label=south:\color{black}$D$] at (0,-4.5) (p10){};
		\node[sysplace,specialSys] at (0,-5.5) (p11){};
		
		\node[envplace,label=south:\color{black}$L_u$] at (-1.5,-5.25) (p12){};
		%\node[blueNode] at (-1.5,-6) (){};
		\node[envplace,label=south:\color{black}$L_d$] at (1.5,-5.25) (p13){};
		%\node[blueNode] at (1.5,-6) (){};
		
		\node[envplace,specialEnv] at (-3,-5.25) (p14){};
		%\node[blueNode] at (-2.5,-6) (){};
		\node[envplace,specialEnv] at (3,-5.25) (p15){};
		%\node[blueNode] at (2.5,-6) (){};

		\node[transition,label=west:\color{black}$u$] at (-1,-0.75) (t1) {};
		\node[transition,label=east:\color{black}$d$] at (1,-0.75) (t2) {};
		
		\node[transition,label=south:\color{black}$i_u$] at (-1,-2.25) (t3) {};
		\node[transition,label=south:\color{black}$i_d$] at (1,-2.25) (t4) {};
		
		\node[transition,label=west:\color{black}$i$] at (0,-3.75) (t5) {};
		
		\node[transition,label=west:\color{black}$s_u$] at (-1.5,-4.5) (t6) {};
		\node[transition,label=east:\color{black}$s_d$] at (1.5,-4.5) (t7) {};
		
		\node[transition,label=south:\color{black}$c_u$] at (-2.25,-5.25) (t8) {};
		\node[transition,label=south:\color{black}$c_d$] at (2.25,-5.25) (t9) {};

		\node[token] at (0,-0.75) (){};
		\node[token] at (-2.25,-1.5) (){};
		\node[token] at (2.25,-1.5) (){};
		\node[token] at (0.75,-3.75) (){};
		
		\draw[arrow] (p1) -- (t1);
		\draw[arrow] (p1) -- (t2);
		\draw[arrow] (t1) -- (p2);
		\draw[arrow] (t2) -- (p3);
		
		\draw[arrow] (p2) -- (t3);
		\draw[arrow] (p3) -- (t4);
		\draw[arrow] (t3) -- (p8);
		\draw[arrow] (t4) -- (p8);
		
		\draw[arrow] (p4) -- (t3);
		\draw[arrow] (p5) -- (t4);
		\draw[arrow] (t3) -- (p6);
		\draw[arrow] (t4) -- (p7);
		
		\draw[darrow] (p8) -- (t5);
		\draw[arrow] (p9) -- (t5);
		\draw[arrow] (t5) -- (p10);
		
		\draw[arrow] (p10) -- (t6);
		\draw[arrow] (p10) -- (t7);
		\draw[arrow] (t6) -- (p11);
		\draw[arrow] (t7) -- (p11);
		
		\draw[arrow] (p8) to[out=180,in=90] (t6);
		\draw[arrow] (p8) to[out=0,in=90] (t7);
		\draw[arrow] (t6) -- (p12);
		\draw[arrow] (t7) -- (p13);
		
		\draw[arrow] (p12) -- (t8);
		\draw[arrow] (p13) -- (t9);
		\draw[arrow] (t8) to (p14);
		\draw[arrow] (t9) to (p15);
		\draw[darrow] (p6) to (t8);
		\draw[darrow] (p7) to (t9);
		\end{tikzpicture}
		
		\subcaption{Petri game for a police strategy.}
		\label{fig:pgA}
	\end{subfigure}
	\begin{subfigure}[c]{0.63\textwidth}
		\centering
		\begin{tikzpicture}[scale=0.75, every label/.append style={font=\tiny}, label distance=-0.5mm]
		\node[envplace,label=south:$B$] at (0,-1.5) (pe1){};
		\node[envplace] at (-2,-1.5) (pe2){};
		\node[envplace] at (2,-1.5) (pe3){};
		
		\node[envplace,label=west:$H$] at (-2,-3) (pe4){};
		\node[envplace,label=east:$H$] at (2,-3) (pe5){};
		
		\node[envplace,label=south:$H$] at (-2,-3.75) (pe6){};
		\node[envplace,label=south:$H$] at (2,-3.75) (pe7){};
		
		\node[envplace,label=south:$L_u$] at (-3.5,-6) (pe8){};
		\node[envplace,draw=gray!50,label=south:\color{gray}$L_d$] at (-0.5,-6) (pe9){};
		\node[envplace,draw=gray!50,label=south:\color{gray}$L_u$] at (0.5,-6) (pe10){};
		\node[envplace,label=south:$L_d$] at (3.5,-6) (pe11){};
		
		\node[sysplace,specialSys,label=west:$U$] at (-3.5,-1.5) (pb1){};
		\node[envplace,specialEnv,label=east:$T$] at (3.5,-1.5) (pb2){};
		\node[sysplace, specialSys,label=north:$H_u$] at (-3.5,-3) (pb3){};
		\node[envplace,specialEnv,label=north:$H_d$] at (3.5,-3) (pb4){};
		
		\node[sysplace,label=north:$C$] at (0,-3) (ps1){};
		
		\node[sysplace,label=south:$D$] at (-2,-4.75) (ps2){};
		\node[sysplace,label=south:$D$] at (2,-4.75) (ps3){};
		
		\node[sysplace,specialSys] at (-2.5,-6) (ps4){};
		\node[sysplace,specialSys,draw=gray!50,fill=gray!20,double=gray!20] at (-1.5,-6) (ps5){};
		\node[sysplace,specialSys,draw=gray!50,fill=gray!20,double=gray!20] at (1.5,-6) (ps6){};
		\node[sysplace,specialSys] at (2.5,-6) (ps7){};

		\node[transition,label=south:$u$] at (-1,-1.5) (te1) {};
		\node[transition,label=south:$d$] at (1,-1.5) (te2) {};
		
		\node[transition,label=east:$i_u$] at (-2,-2.25) (te3) {};
		\node[transition,label=west:$i_d$] at (2,-2.25) (te4) {};
		
		\node[transition,label=north:$i$] at (-1,-3) (te5) {};
		\node[transition,label=north:$i$] at (1,-3) (te6) {};
		
		\node[transition,label=east:$s_u$] at (-3.5,-5.25) (te7) {};
		\node[transition,draw=gray!50,label=west:\color{gray}$s_d$] at (-0.5,-5.25) (te8) {};
		\node[transition,draw=gray!50,label=east:\color{gray}$s_u$] at (0.5,-5.25) (te9) {};
		\node[transition,label=west:$s_d$] at (3.5,-5.25) (te10) {};

		\node[transition,label=south:$c_u$] at (-4.25,-6) (tt1) {};
		\node[sysplace,specialSys,label=north:$H_u$] at (-5,-4.5) (pp1){};
		\node[envplace,specialEnv] at (-5,-6) (pp2){};

		\node[transition,label=south:$c_d$] at (4.25,-6) (tt2) {};
		\node[envplace,specialEnv,label=north:$H_d$] at (5,-4.5) (pp3){};
		\node[envplace,specialEnv] at (5,-6) (pp4){};
		
		\node[token] at (0,-1.5) (){};
		\node[token] at (-3.5,-1.5) (){};
		\node[token] at (3.5,-1.5) (){};
		\node[token] at (0,-3) (){};

		\draw[arrow] (pe1) -- (te1);
		\draw[arrow] (pe1) -- (te2);
		\draw[arrow] (te1) -- (pe2);
		\draw[arrow] (te2) -- (pe3);
		
		\draw[arrow] (pe2) -- (te3);
		\draw[arrow] (pe3) -- (te4);
		\draw[arrow] (te3) -- (pe4);
		\draw[arrow] (te4) -- (pe5);
		
		\draw[arrow] (pe4) -- (te5);
		\draw[arrow] (pe5) -- (te6);
		\draw[arrow] (te5) -- (pe6);
		\draw[arrow] (te6) -- (pe7);
		
		\draw[arrow] (pe6) to[out=180,in=90] (te7);
		\draw[arrow,color=gray!50] (pe6) to[out=0,in=90] (te8);%Compute without drawing
		
		\draw[arrow] (te7) -- (pe8);
		\draw[arrow,color=gray!50] (te8) -- (pe9);
		
		\draw[arrow,color=gray!50] (pe7) to[out=180,in=90] (te9); %Compute without drawing
		\draw[arrow] (pe7) to[out=0,in=90] (te10);
		\draw[arrow,color=gray!50] (te9) -- (pe10);
		\draw[arrow] (te10) -- (pe11);

		%%%%
		\draw[arrow] (pb1) -- (te3);
		\draw[arrow] (pb2) -- (te4);
		\draw[arrow] (te3) -- (pb3);
		\draw[arrow] (te4) -- (pb4);
		
		%%%
		\draw[arrow] (ps1) -- (te5);
		\draw[arrow] (ps1) -- (te6);
		\draw[arrow] (te5) -- (ps2);
		\draw[arrow] (te6) -- (ps3);
		
		\draw[arrow] (ps2) to[out=180,in=45] (te7);
		\draw[arrow,color=gray!50] (ps2) to[out=0,in=135] (te8);
		\draw[arrow] (te7) -- (ps4);
		\draw[arrow,color=gray!50] (te8) -- (ps5);
		
		\draw[arrow,color=gray!50] (ps3) to[out=180,in=45] (te9);
		\draw[arrow] (ps3) to[out=0,in=135] (te10);
		\draw[arrow,color=gray!50] (te9) -- (ps6);
		\draw[arrow] (te10) -- (ps7);

		%%%%%%%%%%%%%%%%%%%%%%%%%%%%%
		
		\draw[arrow] (pe8) -- (tt1);
		\draw[arrow] (pb3) to[out=225,in=90] (tt1);
		\draw[arrow] (tt1) -- (pp1);
		\draw[arrow] (tt1) -- (pp2);
		
		\draw[arrow] (pe11) -- (tt2);
		\draw[arrow] (pb4) to[out=315,in=90] (tt2);
		\draw[arrow] (tt2) -- (pp3);
		\draw[arrow] (tt2) -- (pp4);
		
		%\draw[-,red,dashed] (pb3) to[out=200,in=90] (pe8);
		%\draw[-,red,dashed] (pb4) to[out=340,in=90] (pe11);
		\end{tikzpicture}
		
		\subcaption{Unfolding and winning strategy (without grayed parts).}
		\label{fig:pgB}
	\end{subfigure}

	\caption{A Petri game, an unfolding, and a winning strategy are given.
		Gray places belong to the system whereas white places belong to the environment.
		Winning places are double circles.}
	\label{fig:pg}
\end{figure}

We illustrate the models with two examples. The examples demonstrate the use of control games and Petri games and their differences, which our translations overcome.
Both examples highlight decidable classes~\cite{DBLP:conf/fsttcs/FinkbeinerG17,DBLP:conf/icalp/GenestGMW13}, that are transferable through our results.

As a control game, consider the example of a manager for resources in \refFig{exp-aa}. 
The control game consists of five players: A manager $M$ and two pairs of thread and network connection ($T$, $N$ and $T'$, $N'$).
Both pairs of thread and network connection are identical but act on disjoint actions (primed and not).
There are two resources $X$ and $Y$ that are managed by $M$. 
Each thread ($T$, $T'$) can request access to one of them ($r_X$, $r_Y$) and afterwards wait for the acknowledgement from its network connection ($\mathit{acc}$). 
After the acknowledgement, the thread can use one of the resources ($u_X$, $u_Y$).
Each network connection ($N$, $N'$) synchronizes with its thread on the actions for requests and synchronizes with the manager for communication ($c$). 
Afterwards, each network connection sends the acknowledgement to its thread.
The manager is the only process that comprises controllable actions. 
Upon communication with one of the two network connections, the manager can grant access to the resources $X$ or $Y$ using the controllable actions $g_X$ or $g_Y$. 
The enabled resource can afterwards be accessed and used ($u_X$, $u_Y$).
A losing state can be reached for either thread if an unwanted resource is enabled, i.e., after the acknowledgement, the requested and granted resource do not match. 

This control game can be won by the system. 
After every communication with a network connection, the manager enables the resource that the respective thread requested. 
A winning controller relies on the information transfer associated with every synchronization.
The request of the process is transferred to the manager upon communication with the network connection. 
Then, the correct resource can be enabled. 
This control game falls into a decidable class by our translation to Petri games as it is a single-process system with bad places~\cite{DBLP:conf/fsttcs/FinkbeinerG17}.
Note that the control game has a cyclic communication architecture.

As a Petri game, consider the example of a burglary in \refFig{pgA}.
A crime boss in environment place $B$ decides to either burgle up- or downtown by firing transition $u$ or $d$. 
Depending on the choice, an undercover agent in system place $U$ or a thug in environment place $T$ is instructed by transition $i_u$ or $i_d$ and commits the burglary, i.e., moves to place $H_u$ or $H_d$.
This returns the crime boss to her hideout $H$ where she gets caught and interrogated ($i$) by a cop in system place~$C$.
Afterwards, the cop can send ($s_u$, $s_d$) the flipped crime boss up- or downtown to place $L_u$ or $L_d$ in order to intercept the burglary ($c_u$, $c_d$).

Causal past is key for the existence of winning strategies.
Only upon synchronization players exchange all information about their past.
After the crime boss instructs for a location to burgle, only she and the respective burglar know about the decision.
The cop learns about the location of the burglary after catching the crime boss. 
A winning strategy for the cop catches and interrogates the crime boss and then uses the obtained information to send the flipped crime boss to the correct location.
For this Petri game, our translation results in a control game with acyclic communication architecture~\cite{DBLP:conf/icalp/GenestGMW13}.
Note that the Petri game has two system and two environment players.

\section{Background}
\label{sec:background}

We recall asynchronous automata~\cite{DBLP:journals/ita/Zielonka87}, control games~\cite{DBLP:conf/icalp/GenestGMW13}, Petri nets~\cite{DBLP:books/sp/Reisig85a}, and Petri games~\cite{DBLP:journals/iandc/FinkbeinerO17}. 
Further details can be found in Appendix \ref{app:background}.

\subsection{Zielonka's Asynchronous Automata}
\label{sec:aa}

An \emph{asynchronous automaton}~\cite{DBLP:journals/ita/Zielonka87} is a family of finite automata, called \emph{processes}, synchronizing on shared actions.
Our definitions follow \cite{DBLP:conf/icalp/GenestGMW13}.
The finite set of processes of an asynchronous automaton is defined as $\processes$.
A \emph{distributed alphabet} $(\Sigma, \dom)$ consists of a finite set of \emph{actions} $\Sigma$ and a \emph{domain} function $\dom: \Sigma \rightarrow 2^{\processes} \setminus \{\emptyset\}$. 
For an action $a \in \Sigma$, $\dom(a)$ are all processes that have to synchronize on $a$. 
For a process $p \in \processes$, $\Sigma_p = \{ a \in \Sigma \mid p \in \dom(a) \}$ denotes all actions $p$ is involved in.
A (deterministic) asynchronous automaton $\AA = (\{S_p\}_{p\in\processes}, s_\mathit{in}, \{\delta_a\}_{a\in\Sigma})$ is defined by a finite set of local states $S_p$ for every process $p\in\processes$, the initial state $s_\mathit{in} \in \prod_{p\in\processes} S_p$, and a partial function $\delta_a : \prod_{p\in\dom(a)} S_p \xrightarrow{.} \prod_{p\in\dom(a)} S_p$.
We call an element $\{s_p\}_{p\in\processes} \in \prod_{p\in\processes} S_p$ \emph{a global state}.
For a set of processes $R \subseteq \processes$, we abbreviate $s_R = \{s_p\}_{p\in R}$ as the restriction of the global state to $R$.
We denote that a local state $s' \in S_p$ is part of a global state $s_R$ by $s' \in s_R$.
For a local state $s'$, we define the set of outgoing actions by $\enabled{s'} = \{ a \in \Sigma \mid \exists s_{\dom(a)} \in \domain(\delta_a) : s' \in s_{\dom(a)} \}$.
We can view an asynchronous automaton as a sequential automaton with state space $\prod_{p\in\processes} S_p$ and transitions $s \xrightarrow{a} s'$ if $(s^{}_{\dom(a)}, s'_{\dom(a)}) \in \delta_a$ and $s^{}_{\processes\setminus\dom(a)} = s'_{\processes\setminus\dom(a)}$.
By $\Plays(\AA) \in \Sigma^* \cup \Sigma^\omega$, we denote the set of finite and infinite sequences in this global automaton.
For a finite $u \in \Plays(\AA)$, $\state{u}$ denotes the global state after playing $u$ and $\statep{p}{u}$ the local state of process $p$.

The domain function $\dom$ induces an independence relation $I$: Two actions $a, b \in \Sigma$ are independent, denoted by $(a, b) \in I$, if they involve different processes, i.e., $\dom(a) \cap \dom(b) = \emptyset$.
Adjoint independent actions of sequences of actions can be swapped.
This leads to an equivalence relation $\sim_I$ between sequences, where $u \sim_I w$ if $u$ and $w$ are identical up to multiple swaps of consecutive independent actions. 
The equivalence classes of $\sim_I$ are called \emph{traces} and denoted by $[u]_I$ for a sequence $u$.
Given the definition of asynchronous automata, it is natural to abstract from concrete sequences and consider $\Plays(\AA)$ as a set of traces.

In our translation, an alternative characterization of a subset of asynchronous automata turns out to be practical: We describe every process $p$ by a finite local automaton $\leo_p = (Q_p,  s_{0,p}, \vartheta_p)$ acting on actions from $\Sigma_p$. 
Here, $Q_p$ is a finite set of states, $s_{0,p}$ the initial state and $\vartheta_p \subseteq Q_p \times \Sigma_p \times Q_p$ a deterministic transition relation.
For a family of local processes $\{\leo_p\}_{p\in\processes}$, we define the \emph{parallel composition} $\bigotimes_{p\in\processes} \leo_p$ as an asynchronous automaton with \textbf{\textsf{(1)}} $\forall p \in \processes : S_p = Q_p$, \textbf{\textsf{(2)}} $s_\mathit{in} = \{s_{0,p}\}_{p \in \processes}$, and \textbf{\textsf{(3)}} $\delta_a\big(\{s_p\}_{p \in \dom(a)}\big)$: If for all $p \in \dom(a)$, there exists a state $s'_p \in S_p$ with $(s_p, a, s_p') \in \vartheta_p$ then define $\delta_a(\{s_p\}_{p \in \dom(a)}) = \{s_p'\}_{p\in\dom(a)}$, otherwise it is undefined.
Figure \ref{fig:exp-aa} is an example of such a parallel composition.
Note that not every asynchronous automaton can be described as a composition of local automata.

\subsection{Control Games}
\label{sec:cg}

A \emph{control game}~\cite{DBLP:conf/icalp/GenestGMW13} $\CG = (\AA, \Sigma^{\sys}, \Sigma^{\env}, \{\pSpecial_p\}_{p\in\processes})$ consists of an asynchronous automaton~$\AA$ as a game arena, a distribution of actions into \emph{controllable} actions $\Sigma^\sys$ and \emph{uncontrollable} actions $\Sigma^\env$, and \emph{special states} $\{\pSpecial_p\}_{p\in\processes}$ for a winning objective.
We define the set of plays in the game as $\Plays(\cGame) = \Plays(\AA)$.
Intuitively, a strategy for $\cGame$ can restrict controllable actions but cannot prohibit uncontrollable actions.
Given a play~$u$, a process $p$ only observes parts of it. 
The local $p$-view, denoted by $\view_p(u)$, is the shortest trace $[v]_I$ such that $u \sim_I v \, w$ for some $w$ not containing any actions from $\Sigma_p$.
The $p$-view describes the \emph{causal past} of process $p$ and contains all actions the process is involved in and all actions it learns about via communication. 
We define the set of $p$-views as $\Plays_p(\cGame) = \{ \view_p(u) \mid u \in \Plays(\cGame) \}$.

To avoid confusion with Petri games, we refer to strategies for control games as controllers.
A \emph{controller} for $\cGame$ is a family of local controllers for all processes $\varrho = \{f_p\}_{p \in \processes}$.
A local controller for a process $p$ is a function $f_p : \Plays_p(\cGame) \rightarrow \Sigma^\sys \cap \Sigma_p$.
$\Plays(\cGame, \varrho)$ denotes the set of plays respecting $\varrho$.
It is defined as the smallest set containing the empty play $\epsilon$ and such that for every $u\in\Plays(\cGame, \varrho)$: \textbf{\textsf{(1)}} if $a\in\Sigma^\env$ and $ua \in Plays(\cGame)$ then $ua \in \Plays(\cGame, \varrho)$ and \textbf{\textsf{(2)}} if $a\in\Sigma^\sys$, $ua \in \Plays(\cGame)$, and $\forall p \in \dom(a) : a \in f_p(\view_p(u))$ then $ua \in \Plays(\cGame, \varrho)$.
Environment actions are always possible whereas system actions are only possible if allowed by the local controllers of \emph{all} participating processes.
Local controllers base their decisions on their local view and thereby act only on their causal past.

We define the (possibly empty) set of \emph{final plays} $\PlaysInf(\cGame, \varrho)$ as all finite plays $u \in \Plays(\cGame, \varrho)$ such that there is no $a$ with $u \, a \in \Plays(\cGame, \varrho)$.
We consider either reachability or safety objectives for the system. 
Therefore, $\{\pSpecial_p\}_{p\in\processes}$ describes sets of winning ($\win_p$) or losing ($\bad_p$) states. 
A controller $\varrho$ is \emph{reachability-winning} if it only admits finite plays and on each final play all processes terminate in a winning state.
For safety objectives, we need to ensure progress.
A controller $\varrho$ is \emph{deadlock-avoiding} if $\PlaysInf(\cGame, \varrho) \subseteq \PlaysInf(\cGame, \top)$ for the controller $\top$ allowing all actions, i.e., the controller only terminates if the asynchronous automaton does.
A controller~$\varrho$ is \emph{safety-winning} if it is deadlock-avoiding and no play in $\Plays(\cGame, \varrho)$ visits any local, losing state from $\bigcup_{p\in\processes} \bad_p$. 

\subsection{Petri Nets}
\label{sec:pn}

A \textit{Petri net}~\cite{DBLP:books/sp/Reisig85a,DBLP:journals/tcs/NielsenPW81} $\petriNet$ consists of disjoint sets of \textit{places} $\pl$ and \textit{transitions}~$\tr$, the \textit{flow relation} $\fl$ as multiset over $(\pl \times \tr)\cup (\tr \times \pl)$, and the \textit{initial marking} $\init$ as multiset over $\pl$.
We call elements in $\places \cup \transitions$ \emph{nodes} and $\pNet$ \emph{finite} if the set of nodes is finite.
For node~$x$, the \emph{precondition} (written $\pre{}{x}$) is the multiset defined by $\pre{}{x}(y) = \flow(y, x)$ and \emph{postcondition} (written $\post{}{x}$) the multiset defined by $\post{}{x}(y) = \flow(x, y)$. 
For multiple nets $\pNet^\sigma, \pNet^1, \cdots$, we refer to the components by  $\places^{\pNet^\sigma}$ and write $\pre{\pNet^\sigma}{x}$ unless clear from the context.
Configurations of Petri nets are represented by multisets over places, called \emph{markings}. 
$\init$ is the initial marking.
For a transition $t$, $\pre{}{t}$ is the multiset of places from which tokens are consumed. 
A transition $t$ is \emph{enabled} in marking $M$ if $\pre{}{t} \subseteq M$, i.e., every place in $M$ contains at least as many tokens as required by $t$.  
If no transition is enabled from marking~$M$ then we call $M$ \emph{final}.
An enabled transition~$t$ can \textit{fire} from a marking~$M$ resulting in the successor marking $M' = M - \pre{}{t} + \post{}{t}$ (denoted $\fireTranTo{M}{t}{M'}$).
For markings~$M$ and $M'$, we write $\fireTranTo{M}{t_0,\dotsc,t_{n-1}}{M'}$ if there exist markings $ M = M_0,\dotsc,M_{n} = M'$ s.t.\ $\fireTranTo{M_i}{t_i}{M_{i+1}}$ for all $0 \leq i \leq n-1$.
The set of \textit{reachable markings} of $\pNet$ is defined as $\reach(\pNet) = \{M \mid \exists n \in \mathbb{N}, t_0,\dotsc, t_{n-1} \in \tr : \fireTranTo{\init}{t_0,\dotsc,t_{n-1}}{M}\}$.
A net $\pNet'$ is a \emph{subnet} of $\pNet$ (written $\pNet' \sqsubseteq \pNet$) if $\places' \subseteq \places$, $\transitions' \subseteq \transitions$, $\init' \subseteq \init$, and $\flow' = \flow \upharpoonright (\places' \times \transitions') \cup (\transitions' \times \places')$. 
A Petri net is \emph{1-bounded} if every reachable marking contains at most one token per place. 
It is \emph{concurrency-preserving} if $|\pre{}{t}| = |\post{}{t}|$ for all transitions $t$.

For nodes $x$ and $y$, we write $x \lessdot y$ if $x\in\pre{}{y}$, i.e., there is an arc from $x$ to $y$.
With $\leq$, we denote the reflexive, transitive closure of $\lessdot$.
The \textit{causal past} of $x$ is $\past(x) = \{ y \mid y \leq x \}$.
$x$ and $y$ are \emph{causally related} if $x \leq y \lor y \leq x$.
They are \textit{in conflict} (written $\conflict{x}{y}$) if there exists a place $q\in\pl\setminus \{x,y\}$ and two distinct transitions $t_1, t_2 \in \post{}{q}$ s.t.\ $t_1 \leq x$ and $t_2 \leq y$.
Node $x$ is in \textit{self-conflict} if $\conflict{x}{x}$.
We call $x$ and $y$ \emph{concurrent} if they are neither causally related nor in conflict.
An \textit{occurrence net} is a Petri net~$\pNet$, where the pre- and postcondition of all transitions are sets, the initial marking coincides with places without ingoing transitions ($\forall q \in \pl : q\in\init \Leftrightarrow |\pre{}{q}| = 0$), all other places have exactly one ingoing transition (${\forall q \in \pl \setminus \init : |\pre{}{q}| = 1}$), $\leq$ is \textit{well-founded} (no infinite path following the inverse flow relation exists), and no transition is in self-conflict. 
An \textit{initial homomorphism} from $\pNet$ to $\pNet'$ is a function $\lambda : \pl\cup \tr \rightarrow \pl'\cup \tr'$ that respects node types ($\lambda(\pl) \subseteq \pl' \wedge \lambda(\tr) \subseteq \tr'$),
is structure-preserving on transitions ($\forall t \in \tr : \lambda[\pre{\pNet}{t}] = \pre{\pNet'}{\lambda(t)} \wedge \lambda[\post{\pNet}{t}] = \post{\pNet'}{\lambda(t)}$),
and agrees on the initial markings ($\lambda[\init] = \init'$). 

A branching process \cite{DBLP:journals/acta/Engelfriet91,DBLP:journals/tcs/MeseguerMS96,DBLP:series/eatcs/EsparzaH08} describes parts of the behavior of a Petri net.
Formally, an \textit{(initial) branching process} of a Petri net $\pNet$ is a pair $\iota = (\pNet^\iota,\lambda^\iota)$ where $\pNet^\iota$ is an occurrence net and $\lambda^\iota :\pl^\iota\cup \tr^\iota\rightarrow\pl\cup \tr$ is an initial homomorphism  from $\pNet^\iota$ to $\pNet$ that is injective on transitions with the same precondition ($\forall t, t' \in \tr^\iota : (\pre{\pNet^\iota}{t} = \pre{\pNet^\iota}{t'} \wedge \lambda^\iota(t) = \lambda^\iota(t')) \Rightarrow t = t'$).
A branching process describes subsets of possible behaviors of a Petri net. 
Whenever a place or transition can be reached on two distinct paths it is split up. 
$\lambda$ can be thought of as label of the copies into nodes of $\pNet$. 
The injectivity condition avoids additional unnecessary splits: Each transition must either be labelled differently or occur from different preconditions. 
The \emph{unfolding} $\mathfrak{U}$ of $\pNet$ is the maximal branching process: Whenever there is a set of pairwise concurrent places $C$ s.t.\ $\lambda[C] = \pre{\pNet}{t}$ for some transition $t$ then there exists $t'$ with $\lambda(t')=t$ and $\pre{\pNet^\mathfrak{U}}{t'} = C$. It represents ever possible behavior of $\pNet$.

\subsection{Petri Games}
\label{sec:pg}

A \textit{Petri game}~\cite{DBLP:journals/iandc/FinkbeinerO17} is a tuple $\pGame=(\plS,\plE,\tr,\fl,\init,\mathit{Sp})$. 
\emph{System places} $\plS$ and \emph{environment places} $\plE$ partition the places of the underlying, finite net $\petriNet$ with $\pl = \plS \uplus \plE$.
We extend notation from the underlying net to $\pGame$ by, e.g., defining $\pre{\pGame}{\cdot} = \pre{\pNet}{\cdot}$ and $\places^\pGame = \places^\pNet$.
The game progresses by firing transitions in the underlying net. 
Intuitively, a strategy can control the behavior of tokens on system places by deciding which transitions to allow. 
Tokens on environment places belong to the environment and cannot be restricted by strategies. 
$\mathit{Sp}\subseteq\pl$ denotes \emph{special places} used to pose a winning objective.
For graphical representation, we depict a Petri game as the underlying net and color system places gray, environment places white, and special places as double circles (cf.\ \refFig{pg}).

A \emph{strategy} for $\pGame$ is an initial branching process $\sysstrat = (\pNet^\sysstrat, \lambda^\sysstrat)$ satisfying \emph{justified refusal}:
If there is a set of pairwise concurrent places $C$ in $\pNet^\sigma$ and a transition $t \in \transitions^\pGame$ with $\lambda[C] = \pre{\pGame}{t}$ then there either is a transition $t'$ with $\lambda(t')=t$ and $C = \pre{\pNet^\sigma}{t'}$ or there is a system place $q \in C \cap \lambda^{-1}[\places_\psys]$ with $t \not\in \lambda[\post{\pNet^\sigma}{q}]$. 
Since a branching process describes subsets of the behavior of a Petri net, a strategy is a restriction of possible moves in the game. 
Justified refusal enforces that only system places can prohibit transitions based on their causal past. 
From every situation in the game, a transition possible in the underlying net is either allowed, i.e., in the strategy, or there is a \emph{system} place that never allows it.
In particular, transitions involving only environment places are always possible. 
A strategy $\sigma$ is \emph{reachability-winning} for a set of winning places $\mathit{Sp} = \win$ if $\pNet^\sigma$ is a finite net and in each final, reachable marking every token is on a winning place.
A strategy is \emph{deadlock-avoiding} if for every final, reachable marking $M$ in the strategy, $\lambda[M]$ is final as well, i.e., the strategy is only allowed to terminate if the underlying Petri net does so. 
A strategy $\sigma$ is \emph{safety-winning} for bad places $\mathit{Sp} = \bad$ if it is deadlock-avoiding and no reachable marking contains a bad place.
For both objectives, we can require $\sigma$ to be \emph{deterministic}: For every reachable marking $M$ and system place $q \in M$ there is at most one transition from $\post{\pNet^\sysstrat}{q}$ enabled in $M$.
In \refFig{pgB}, the unfolding of \refFig{pgA} is depicted labeled by $\lambda$. 
Excluding the grayed parts, this is a winning strategy for the system.

For safety as winning objective, unbounded Petri games are undecidable in general~\cite{DBLP:journals/iandc/FinkbeinerO17} whereas bounded ones with either one system player~\cite{DBLP:conf/fsttcs/FinkbeinerG17} or one environment player~\cite{DBLP:journals/iandc/FinkbeinerO17} are \textsc{EXPTIME}-complete.
Bounded synthesis is a semi-decision procedure to find winning strategies~\cite{DBLP:conf/birthday/Finkbeiner15,DBLP:journals/corr/abs-1711-10637,ATVA19bounded}.
Both approaches are implemented in the tool \textsc{Adam}~\cite{DBLP:conf/cav/FinkbeinerGO15,DBLP:journals/corr/abs-1711-10637}.

\section{Game Equivalence}
\label{sec:equivalence}

A minimum requirement for translations between games is to be \emph{winning-equivalent}.
The system has a winning strategy in one game if and only if it has a winning strategy in the translated other one.
One trivial translation fulfilling this is to solve the game and to return a minimal winning-equivalent game.
Such a translation is not desirable, especially since decidability in both control games and Petri games is still an open question~\cite{DBLP:conf/icalp/Muscholl15,DBLP:journals/iandc/FinkbeinerO17}.
Instead, our translations preserve the underlying structure of the games.
We propose \textit{strategy-equivalence} as an adequate equivalence notion. 
Our notion is based on weak bisimulation which is popular and powerful to relate concurrent systems represented as Petri nets~\cite{DBLP:journals/acta/BestDKP91,DBLP:conf/apn/AutantS92,olderogbooks}.

For our purpose, a bisimulation between the underlying Petri net and the asynchronous automaton is not sufficient.
Instead, we want to express that any strategy can be matched by a strategy that allows equivalent (bisimilar) behavior, i.e., allows identical actions/transitions.
In both models, strategies are defined based on the causal past of the players.
A Petri game $\pGame$ utilizes unfoldings whereas a control game $\cGame$ utilizes local views.
We consider a strategy and a controller equivalent if there is a weak bisimulation between the branching process of the strategy and the plays that are compatible with the controller.
We base our definition on a set of shared actions and transitions between the Petri game and the control game. 
We refer to them as \emph{observable}.
All non-shared transitions and actions are considered \emph{internal} ($\tau$). 
If we, e.g., translate a Petri game to a control game we aim for a control game that contains all transitions as observable actions but might add internal ones.

\begin{definition}
	A strategy $\sigma$ for $\pGame$ and controller $\varrho$ for $\mathcal{C}$ are \emph{bisimilar} if there exists a relation $\relation \, \subseteq \reach(\pNet^\sigma) \times \Plays(\AA, \varrho)$ s.t.\ $\init^\sigma \relation \epsilon$ and all following conditions hold:
	\begin{itemize}
		\item If $M \relation u$ and $\fireTranTo{M}{a}{M'}$ there exists $u' \in \Plays(\AA, \varrho)$ with $u' = u \tau^* a \tau^*$ and $M' \relation u'$

		\item If $M \relation u$ and $\fireTranTo{M}{\tau}{M'}$ there exists $u' \in \Plays(\AA, \varrho)$ with $u' = u \tau^*$ and $M' \relation u'$

		\item If $M \relation u$ and $u' = u \, a$ there exists $M' \in \reach(\pNet^\sigma)$ with $\fireTranTo{M}{\tau^* a \tau^*}{M'}$ and $M' \relation u'$

		\item If $M \relation u$ and $u' = u \, \tau$ there exists $M' \in \reach(\pNet^\sigma)$ with $\fireTranTo{M}{\tau^*}{M'}$ and $M' \relation u'$

	\end{itemize}
\end{definition}

A Petri game $\pGame$ and a control game $\cGame$ are called \emph{strategy-equivalent} if for every winning strategy~$\sigma$ for~$\pGame$ there exists a bisimilar winning controller $\varrho_\sigma$ for $\cGame$ and for every winning controller~$\varrho$~for~$\cGame$ there exists a bisimilar winning strategy $\sigma_\varrho$ for $\pGame$.

\section{Translating Petri Games to Control Games}
\label{sec:PGtoAA}

We give our translation from Petri games to control games and prove that it yields strategy-equivalent (and therefore winning-equivalent) games. 
Moreover, we provide an exponential lower bound, showing that our translation is asymptomatically optimal when requiring strategy-equivalence. 
We present the translation for reachability objectives.

\subsection{Construction}

We describe the construction of our translation for a restrictive class of Petri games called \emph{sliceable}.
In \refSection{pg-cp}, the construction is generalized to \emph{concurrency-preserving} Petri games.

\noindent\textbf{\textsf{Slices~~}}
A Petri game describes the \emph{global} behavior of the players. 
By contrast, a control game is defined in terms of \emph{local} processes. 
Similarly, a Petri game strategy is a global branching process opposed to a family of local controllers for control games. 
The first difference our translation needs to overcome is to distribute a Petri game into parts describing the local behavior of players. 
Therefore, we dismantle the Petri game into \emph{slices} for each token.

\begin{definition}
A \textit{slice} of a Petri net $\pNet$ is a Petri net $\slice = (\places^\slice, \transitions^\slice, \flow^\slice, \init^\slice)$ s.t.,
{\normalfont\textbf{\textsf{(1)}}} $\slice \sqsubseteq \pNet$,
{\normalfont\textbf{\textsf{(2)}}} $|\init^\slice| = 1$,
{\normalfont\textbf{\textsf{(3)}}}  $\forall t \in \transitions^\slice : \abs{\pre{\slice}{t}} = \abs{\post{\slice}{t}} = 1$,
{\normalfont\textbf{\textsf{(4)}}}  $\forall q \in \places^\slice : \post{\pNet}{q} \subseteq \transitions^\slice$
\end{definition}

A slice is a subnet of $\pNet$ \textbf{\textsf{(1)}} that describes the course of exactly one token \textbf{\textsf{(2, 3)}} and includes every possible move of this token \textbf{\textsf{(4)}}. 
A slice characterizes the exact behavior of a single token in the global net $\pNet$.
For a family of slices $\{\slice\}_{\slice \in \slices}$, the parallel composition $\parallel_{\slice \in \slices} \slice$ is the Petri net with places $\biguplus_{\slice \in \slices} \places^\slice$, transitions $\bigcup_{\slice \in \slices} \transitions^\slice$, flow relation $\biguplus_{\slice \in \slices} \flow^\slice$, and initial marking $\biguplus_{\slice \in \slices} \init^\slice$.
All unions, except for the union of transitions, are disjoint.
Transitions can be shared between multiple slices, creating synchronization.
A Petri net $\pNet$ is \textit{sliceable} if there is a family of slices $\{\slice\}_{\slice \in \slices}$ s.t.\ $\pNet = \; \parallel_{\slice \in \slices} \slice$ and $\biguplus_{\slice \in \slices} \places^\slice$ is a partition of $\places^\pNet$, i.e., $\pNet$ can be described by the local movements of tokens.
Sliceable Petri nets are concurrency-preserving and 1-bounded.
We extend slices to Petri games in the natural way by distinguishing system, environment, and special places. 
Figure~\ref{fig:firstRed} depicts a Petri game (a) and a possible distribution into slices (b).
Note that even concurrency-preserving and 1-bounded Petri games must not be sliceable and that a distribution in slices is not unique.

\noindent\textbf{\textsf{Commitment Sets~~}}
In control games, actions are either controllable or uncontrollable whereas, in Petri games, players are distributed between the system and the environment.
In our construction, we represent transitions as actions and need to guarantee that only certain players can control them. 
In control games, this cannot be expressed directly.
We overcome this difference by using commitment sets.
Each process that should be able to control an action chooses a commitment set, i.e., moves to a state that explicitly encodes its decision. 

\begin{figure}[t!]
	\begin{tcolorbox}[colback=white, colframe=myYellow, arc=3mm, boxrule=1mm]
		\small
		Define $\processes = \slices$ and the distributed alphabet as $(\Sigma, dom)$ with:
		\vspace{-0.2cm}
		\begin{align*}
		\Sigma = \transitions &\;\cup\; \{ \tau_{(q, A)}  \mid  q \in \places_\psys \; \land \;  A \subseteq \post{\pGame}{q}\}\\
		&\color{red}\;\cup\; \{ \; \text{\Lightning}^{(q, A)}_{[t_1, t_2]} \mid q \in \places_\psys \; \land \; A \subseteq \post{\pGame}{q} \; \land \; t_1,  t_2 \in A \; \land \;  t_1 \neq t_2\}
		\end{align*}
		
		\vspace{-0.15cm}
		and $\dom: \Sigma \to 2^{\processes} \setminus \{\emptyset\}$:
		\vspace{-0.15cm}
		\begin{align*}
			\dom(\,t\,) &= \{\slice \in \slices \mid t \in \transitions^\slice \} \quad \text{for} \; t \in \transitions\\
			\dom(\,\tau_{(q, A)}\,) &= \{ \slice \} \; \textrm{where } \slice \in \slices \textrm{ is the unique slice  s.t.\ } q \in \places^\slice\\
			\color{red}\dom(\,\text{\Lightning}^{(q, A)}_{[t_1, t_2]}\,) &\color{red}= \{\slice \in \slices \mid t_1 \in \transitions^\slice \; \lor \; t_2 \in \transitions^\slice\}
		\end{align*}
		
		\vspace{-0.15cm}
		For each slice $\slice = (\places^\slice, \transitions^\slice, \flow^\slice, \init^\slice) \in \slices$, we define a local process $\leo_\slice = (Q_\slice, q_{0,\slice}, \vartheta_\slice)$ with $\vartheta_\slice \subseteq Q_\slice \times \Sigma_\slice \times Q_\slice$ as:
		\begin{itemize}
			\item $Q_\slice = \places^\slice \cup \{ (q, A) \mid q \in \places^\slice \cap \places_\psys \; \land \; A \subseteq \post{\slice}{q} \} \textcolor{red}{ \;\cup\; \{\badw_\slice\}}$
			
			\item $q_{0,\slice}$ is the unique state s.t.\ $\init^\slice = \{q_{0,\slice}\}$	
		\end{itemize}
		
		and $\vartheta_\slice$ is given by: 
		\vspace{-0.2cm}
		\begin{center}
			\begin{tikzpicture}[scale=0.88]
			
			\draw [draw, thick,rounded corners,] (0,0) rectangle (4.5,-2);
			
			\draw [draw, thick,rounded corners,] (5,0) rectangle (9.5,-2);
			
			\draw [draw, thick,rounded corners,] (10,0) rectangle (14.5,-2);
			
			\draw[-, thick, densely dashed] (0,-0.7) to (4.5,-0.7);
			\draw[-, thick, densely dashed] (5,-0.7) to (9.5,-0.7);
			\draw[-, thick, densely dashed] (10,-0.7) to (14.5,-0.7);

			\node[align=center] at (2.25,-0.35) () {\footnotesize$q \xmapsto[]{\tau_{(q, A)}} (q, A)$};
			\node[align=center] at (2.25,-1.35) () {\scriptsize$q \in \places_\psys \; \land$\\
				\scriptsize$A \subseteq \post{\slice}{q}$};
			
			\node[align=center] at (0.5,-0.35) () {\footnotesize \textbf{\textsf{(1)}}};
			
			\node[align=center] at (7.25,-0.35) () {\footnotesize$q \xmapsto[]{\;\;\, t \;\;\,} q'$};
			\node[align=center] at (7.25,-1.35) () {\scriptsize$q \in \places_\penv \; \land \; t \in \transitions \; \land $\\
				\scriptsize$q \in \pre{\slice}{t} \; \land \; q' \in \post{\slice}{t}$};
			
			\node[align=center] at (5.5,-0.35) () {\footnotesize \textbf{\textsf{(2)}}};
			
			\node[align=center] at (12.25,-0.35) () {\footnotesize$(q, A) \xmapsto[]{\;\;\, t\;\;\,} q'$};
			\node[align=center] at (12.25,-1.35) () {\scriptsize$q \in \places_\psys \; \land \; t \in A \; \land $\\
				\scriptsize$q \in \pre{\slice}{t} \; \land \; q' \in \post{\slice}{t}$};
			
			\node[align=center] at (10.5,-0.35) () {\footnotesize \textbf{\textsf{(3)}}};
			
			%%%%%%%%%%%%%%%%%%%%%%%%%%%%%%%%%%%%%%%%%%%%%%%%%%%%%%%%%%%%%%%%%
			
			\draw [draw=red, thick,rounded corners,] (0,-2.5) rectangle (4.5,-5);
			
			\draw [draw=red, thick,rounded corners,] (5,-2.5) rectangle (9.5,-5);
			
			\draw [draw=red, thick,rounded corners,] (10,-2.5) rectangle (14.5,-5);
			
			\draw[-,red, thick, densely dashed] (0,-3.3) to (4.5,-3.3);
			\draw[-,red, thick, densely dashed] (5,-3.3) to (9.5,-3.3);
			\draw[-,red, thick, densely dashed] (10,-3.3) to (14.5,-3.3);

			\node[align=center] at (2.25,-2.9) () {\footnotesize$\color{red} (q, A) \xmapsto[]{\text{\Lightning}^{(q, A)}_{[t_1, t_2]}}  \badw_\slice$ };
			\node[align=center] at (2.25,-4.15) () {};
			
			\node[align=center] at (0.5,-2.9) () {\footnotesize \color{red}\textbf{\textsf{(4)}}};

			\node[align=center] at (7.25,-2.9) () {\footnotesize$\color{red} q \xmapsto[]{\text{\Lightning}^{(q', A')}_{[t_1, t_2]}} \badw_\slice$};
			\node[align=center] at (7.25,-4.15) () {\scriptsize $\color{red} q' \not\in Q_\slice \; \land \; q \in \places_\penv \; \land$\\
				\scriptsize $\color{red} (t_1 \in \transitions^\slice \Rightarrow t_1 \in \post{\slice}{q}) \; \land $\\
				\scriptsize $\color{red} (t_2 \in \transitions^\slice \Rightarrow t_2 \in \post{\slice}{q})$};
			
			\node[align=center] at (5.5,-2.9) () {\footnotesize \color{red}\textbf{\textsf{(5)}}};

			\node[align=center] at (12.25,-2.9) () {\footnotesize$\color{red} (q, A) \xmapsto[]{\text{\Lightning}^{(q', A')}_{[t_1, t_2]}}  \badw_\slice$ };
			\node[align=center] at (12.25,-4.15) () {\scriptsize $\color{red} q' \not\in Q_\slice \; \land$\\
				\scriptsize $\color{red} (t_1 \in \transitions^\slice \Rightarrow t_1 \in A) \; \land $\\
				\scriptsize $\color{red} (t_2 \in \transitions^\slice \Rightarrow t_2 \in A)$};
			
			\node[align=center] at (10.5,-2.9) () {\footnotesize \color{red}\textbf{\textsf{(6)}}};

			\end{tikzpicture}
		\end{center}
		\vspace{-0.2cm}
		Define $\AA_\pGame = \bigotimes_{\slice \in \slices} \leo_\slice$ and 
		%and $\widehat{\cGame_\pGame} = \bigotimes\limits_{\slice \in \slices} \leo_\slice$ \textit{including the red parts}.
		%
		the control game as $\cGame_\pGame = ( \AA_\pGame, \Sigma^{\sys}, \Sigma^{\env}, \{\win_\slice\}_{\slice \in \slices} )$ %(or $( \widehat{\cGame_\pGame}, \Sigma^{\sys}, \Sigma^{\env}, \{\win_p\}_{p \in \slices} )$) 
		where
		\begin{itemize}
			\item $\Sigma^{\sys} = \{ \tau_{(q, A)} \mid q \in \places_\psys \; \land \; A \subseteq \post{\pGame}{q} \}$
			\item $\Sigma^{\env} = \transitions \color{red}\;\cup\; \{ \;\text{\Lightning}^{(q, A)}_{[t_1, t_2]} \mid q \in \places_\psys \; \land \; A \subseteq \post{\pGame}{q} \; \land \; t_1,  t_2 \in A \; \land \;  t_1 \neq t_2\}$
			\item $\win_\slice = (\win \; \cap \; \places^\slice) \; \cup \; \{(q, A) \mid q \in (\win \cap \places^\slice) \; \land \; A \subseteq \post{\pGame}{q} \}$
		\end{itemize}
		
	\end{tcolorbox}
	\caption{The construction of the translated control game for a Petri game $\pGame = (\places_\psys, \places_\penv, \transitions, \flow, \init, \win)$, distributed in slices $\{\slice\}_{\slice \in \slices}$, is depicted.
		Excluding the red parts, this is the definition of $\cGame_\pGame$.
		Including the red parts, this is the definition of $\widehat{\cGame_\pGame}$.
	}
	\label{fig:PGtoAA}
\end{figure}
We fix a sliceable game $\pGame = (\places_\psys, \places_\penv, \transitions, \flow, \init, \win)$ and a distribution in slices $\{\slice\}_{\slice \in \slices}$.
We begin by defining a control game $\cGame_\pGame$. Afterwards, we describe a possible modification $\widehat{\cGame_\pGame}$, that enforces determinism. The construction is depicted in \refFig{PGtoAA}.

We transform every slice $\slice$ into a process that is described by a local automaton $\leo_\slice$.
Hence, we use the terms slice and process interchangeably.
Every place in $\slice$ becomes a local state in $\leo_\slice$. The process starts in the state that corresponds to the initial place of the slice. 
For every \emph{system} place $q$, we furthermore add the aforementioned commitment sets. These are states $(q, A)$ representing every possible commitment, i.e., every $A \subseteq \post{\pGame}{q}$.

Every transition $t$ is added as an uncontrollable action. Action $t$ involves all processes with slices synchronizing on $t$. 
To choose a commitment set, we furthermore add controllable actions ($\tau$-actions) that are local to each process. 
We assume that each process chooses at most one commitment set. 
The transition relation $\vartheta_\slice$ is given by three rules:
From every \emph{system} place $q \in \places_\psys$, a process can choose a commitment set using the corresponding $\tau$-action \textbf{\textsf{(1)}}.
From an environment place $q \in \places_\penv$, $t$ can fire if $q$ is in the precondition of $t$ ($q \in \pre{\slice}{t}$). The process is then moved to the state $q'$ that corresponds to the place that is reached when firing $t$ in the slice ($q' \in \post{\slice}{t}$) \textbf{\textsf{(2)}}.
A process on an environment place can hence never restrict any actions; as in Petri games.
For a system place, the rule is almost identical but only admits $t$ if a commitment set has been chosen, that contains $t$ \textbf{\textsf{(3)}}. 
Therefore, a process on a system place can control actions by choosing commitment sets; as in Petri games.
States corresponding to winning places become winning states. 

An example translation is depicted in \refFig{firstRed}.
The Petri game (a) comprise two players starting in $A$ and $C$. 
They can move to $B$ and $D$ using $e_1, e_2$, or $i$ and afterwards synchronize on $a$ or $b$. 
The Petri game can be distributed into slices (b).
In the control game from our construction (c), the slice containing only environment places results in the local process on the left.
For the system places in the other slice, commitment sets are added as states $\{C\} \times 2^{\{i\}}$ and $\{D\} \times 2^{\{a,b\}}$. 
The process can choose them using controllable $\tau$-actions and the actions $a, b$, and $i$ can only occur if included in the current set. 
The construction guarantees that only the second process can control transitions $a$ and $b$, as in the Petri game.

\begin{figure}[t!]
	\begin{subfigure}[c]{0.25\textwidth}
		\centering
		\begin{tikzpicture}[scale=0.75, every label/.append style={font=\tiny}, label distance=-1mm]
		\node[envplace,label=west:$A$] at (0,0) (e1){};
		
		\node[envplace,specialEnv,label=west:$B$] at (0,-1.5) (e2){};
		
		\node[sysplace,label=east:$C$] at (2.5,0) (s1){};
		\node[sysplace,specialSys,label=east:$D$] at (2.5,-1.5) (s2){};
		
		\node[transition,label=west:$e_1$] at (-0.5, -0.75) (t1){};
		\node[transition,label=east:$e_2$] at (0.5, -0.75) (t2){};
		
		\node[transition,label=east:$i$] at (2.5, -0.75) (t3){};
		
		\node[transition,label=east:$a$] at (1.5, -0.75) (t4){};
		\node[transition,label=north:$b$] at (1.5, -2.25) (t5){};
		
		\node[token] at (0,0) (){};
		\node[token] at (2.5,0) (){};
		
		\draw[arrow] (e1) to (t1);
		\draw[arrow] (e1) to (t2);
		\draw[arrow] (t1) to (e2);
		\draw[arrow] (t2) to (e2);
		
		\draw[arrow] (e2) to[out=0,in=270] (t4);
		\draw[darrow] (e2) to[out=270,in=180] (t5);
		\draw[arrow] (t4) to[out=90,in=0] (e1);

		\draw[arrow] (s1) to (t3);
		\draw[arrow] (t3) to (s2);
		
		\draw[arrow] (s2) to[out=180,in=270] (t4);
		\draw[darrow] (s2) to[out=270,in=0] (t5);
		\draw[arrow] (t4) to[out=90,in=180] (s1);
		
		\end{tikzpicture}
		\subcaption{}
	\end{subfigure}
	\begin{subfigure}[c]{0.3\textwidth}
		\centering
		\begin{tikzpicture}[scale=0.75, every label/.append style={font=\tiny}, label distance=-1mm]
		\node[envplace,label=west:$A$] at (0,0) (e1){};
		\node[envplace,specialEnv,label=west:$B$] at (0,-1.5) (e2){};
		
		\node[transition,label=west:$e_1$] at (-0.5, -0.75) (t1){};
		\node[transition,label=west:$e_2$] at (0.5, -0.75) (t2){};
		
		\node[transition,label=west:$a$] at (1.5, -0.75) (t4){};
		\node[transition,label=west:$b$] at (0, -2.25) (t5){};
		
		\node[token] at (0,0) (){};

		\draw[arrow] (e1) to (t1);
		\draw[arrow] (e1) to (t2);
		\draw[arrow] (t1) to (e2);
		\draw[arrow] (t2) to (e2);
		
		\draw[arrow] (e2) to[out=0,in=270] (t4);
		\draw[darrow] (e2) to (t5);
		\draw[arrow] (t4) to[out=90,in=0] (e1);
		
		%%%%%%%%%%%%%%%%%%%%%%%%%%%
		
		\node[sysplace,label=east:$C$] at (3,0) (s1){};
		\node[sysplace,specialSys,label=east:$D$] at (3,-1.5) (s2){};
		
		\node[transition,label=east:$i$] at (3, -0.75) (ts3){};
		
		\node[transition,label=east:$a$] at (2.25, -0.75) (ts4){};
		\node[transition,label=east:$b$] at (3, -2.25) (ts5){};
		
		\node[token] at (3,0) (){};
		
		\draw[arrow] (s1) to (ts3);
		\draw[arrow] (ts3) to (s2);
		
		\draw[arrow] (s2) to[out=180,in=270] (ts4);
		\draw[darrow] (s2) to (ts5);
		\draw[arrow] (ts4) to[out=90,in=180] (s1);
		
		\end{tikzpicture}
		\subcaption{}
	\end{subfigure}
	\begin{subfigure}[c]{0.45\textwidth}
		\centering
		\begin{tikzpicture}[scale=0.8, every label/.append style={font=\tiny}, label distance=-0.5mm]
		\node[aastate,label=north:$A$] at (0,0) (s11){};
		\node[] at (0,-1.5) (s12){};
		
		\node[aastate,draw=red,label=east:{$\color{red}\bot_{\slice_1}$}] at (0,-3) (er1){};
		
		\draw[arrow] (s11)+(0.35, 0.35) to (s11);
		
		\draw[arrow] (s11) to [out=270,in=90] node[left=-1mm] {\tiny$e_1$} (s12);
		\draw[arrow] (s11) to [out=340 ,in=20] node[right=-1mm] {\tiny$e_2$} (s12);
		
		\draw[arrow] (s12) to [out=160 ,in=200] node[left=-0.5mm] {\tiny$a$} (s11);
		
		\draw[arrow,loop right] (s12) to node[right=-0.5mm] {\tiny$b$} (s12);

		\draw[arrow,red] (s12) to node[left=-0.5mm] {\tiny$\text{\Lightning}$} (er1);
		
		%%%%%%%%%%%%%
		
		\node[aastate,label=north:$C$] at (4,0) (s1){};
		
		\node[aastate,label=north:{$(C, \emptyset)$}] at (2,0) (c1){};
		\node[aastate,label=north:{$(C, \{i\})$}] at (6,0) (c2){};
		
		\node[specialAA,label=south:$D$] at (4,-2.2) (s2){};
		
		%\node[specialAA,label=west:{$(D, \{a,b\})$}] at (2,-2.2) (c3){};
		\node[] at (2,-2.2) (c3){};
		\node[specialAA,label=east:{$(D, \{b\})$}] at (4,-3.7) (c4){};
		\node[specialAA,label=south:{$(D, \emptyset)$}] at (6,-2.2) (c5){};
		\node[specialAA,label=east:{$(D, \{a\})$}] at (4,-1) (c6){};

		\node[aastate,draw=red,label=west:{$\color{red}\bot_{\slice_2}$}] at (2,-3.7) (er2){};
		
		\draw[arrow] (s1)+(0.35, 0.35) to (s1);
		
		\draw[arrow, densely dotted] (s1) to node[above=-0.5mm] {\tiny$\tau_{(C, \emptyset)}$} (c1);
		\draw[arrow, densely dotted] (s1) to node[above=-0.5mm] {\tiny$\tau_{(C, \{i\}}$} (c2);

		\draw[arrow] (c2) to[out=270,in=10] node[right=0mm] {\tiny$i$} (s2);
		
		\draw[arrow, densely dotted] (s2) to[out=150,in=30] node[above=-0.5mm] {\tiny$\tau_{(D, \{a,b\})}$} (c3);
		\draw[arrow, densely dotted] (s2) to[out=300,in=60] node[right=-0.5mm] {\tiny$\tau_{(D, \{b\})}$} (c4);
		\draw[arrow, densely dotted] (s2) to node[below=-0.5mm] {\tiny$\tau_{(D, \emptyset)}$} (c5);
		\draw[arrow, densely dotted] (s2) to node[right=-0.5mm] {\tiny$\tau_{(D, \{a\})}$} (c6);
		
		\draw[arrow] (c3) to[out=90,in=190] node[left] {\tiny$a$} (s1);
		\draw[arrow] (c3) to[out=330,in=210] node[below=-0.5mm] {\tiny$b$} (s2);
		\draw[arrow] (c4) to[out=120,in=240] node[left=-0.5mm] {\tiny$b$} (s2);
		\draw[arrow] (c6) to node[right=-0.5mm] {\tiny$a$} (s1);
		
		\draw[arrow,red] (c3) to node[left=-0.5mm] {\tiny$\text{\Lightning}$} (er2);

		%Redraw
		\node[specialAA,label=west:{$(D, \{a,b\})$}] at (2,-2.2) (c3){};
		\node[specialAA,label=west:$B$] at (0,-1.5) (s12){};
		%\draw[-,thick, red, dashed] (s12) to (c3);

		\end{tikzpicture}
		\subcaption{}
	\end{subfigure}

	\caption{A sliceable Petri game $\pGame$ (a), a possible (in this case unique) distribution in slices ($\slice_1, \slice_2$) (b), and the asynchronous automaton $\cGame_\pGame$ obtained by our translation (c).
	$\widehat{\cGame_\pGame}$ comprises additional $\bot$-states and one $\text{\Lightning}^{(D, \{a,b\})}_{[a, b]}$-action (named \Lightning) depicted in red.
	}
	\label{fig:firstRed}
\end{figure}

\noindent\textbf{\textsf{Non-Determinism~~}}
In deterministic strategies, every system place allows transitions s.t.\ in every situation, there is at most one of them enabled.
In $\cGame_\pGame$, the controller can choose arbitrary commitment sets and, thus, a winning controller can result in a \emph{non-deterministic} strategy for $\pGame$.
To ensure deterministic strategies, we want to penalize situations where a commitment set in $\cGame_\pGame$ is chosen s.t.\ two or more distinct actions from this set can be taken. 

To achieve this, we define the modified game $\widehat{\cGame_\pGame}$.
We equip each process with a $\badw$-state from which no winning configurations are reachable.
Uncontrollable \Lightning-actions move processes to $\badw$-states and thereby cause the system to lose. 
The situation to be covered comprises a process that has chosen a commitment set, i.e., is in a state $(q, A)$, and two distinct actions $t_1$ and $t_2$ in $A$.
For every such combination, we add a $\text{\Lightning}^{(q, A)}_{[t_1, t_2]}$-action that involves all processes participating in $t_1$ or $t_2$ and can be taken exactly if $(q, A)$ is a current state and both $t_1$ \emph{and} $t_2$ could occur from the current global state.
The three rules in $\vartheta$ add the $\text{\Lightning}^{(q, A)}_{[t_1, t_2]}$-action to each process.
It fires if one process is in state $(q,A)$ \textbf{\textsf{(4)}} and all other involved processes are in states such that both $t_1, t_2 \in A$ are possible \textbf{\textsf{(4, 5)}}.
To ensure that $t_1$ and $t_2$ can both be taken, we distinguish between system and environment places:
Every process $\slice$ on an environment place needs to be in the right state, i.e., if $\slice$ is involved in $t_i$ ($t_i \in \transitions^\slice$), then $t_i$ is in the postcondition of its current place for $i=1,2$ \textbf{\textsf{(5)}}.
If on a system place, $t_i$ must not only be in the postcondition but also in the currently chosen commitment set \textbf{\textsf{(6)}}.

\noindent\textbf{\textsf{Size~~}}
In both $\cGame_\pGame$ and $\widehat{\cGame_\pGame}$, the size of the alphabet and number of local states is exponential in the number of transitions and linear in the number of places. 
For $\cGame_\pGame$, the blow-up in the alphabet can be kept polynomial by using a tree construction to choose commitment sets.
For a bound on the number of outgoing transitions, both $\cGame_\pGame$ and $\widehat{\cGame_\pGame}$ are of polynomial size.

\subsection{Correctness}
We show that our translation yields strategy-equivalent games by outlining the translation of winning strategies and controllers between $\pGame$ and $\cGame_\pGame$.
Both game types rely on causal information, i.e., a strategy/controller bases its decisions on every action/transition it took part in as well as all information it received upon communication. 
In $\pGame$, the information is carried by individual tokens. 
In our translation, we transform each slice for a token into a process that is involved in exactly the transitions that the slice it is build from takes part in, i.e., we preserve the communication architecture.
At every point, all processes in $\cGame_\pGame$ possess the same information as their counterpart slices. 
Using the commitment set, our translation ensures that only processes on a state based on a system place can control any behavior. 
Therefore, a process and its counterpart slice have the same possibilities for control.

\noindent\textbf{\textsf{Translating a Strategy for $\pGame$ to a Controller for $\cGame_\pGame$~~}}
Given a winning strategy $\sigma$, we construct a controller $\varrho_\sigma$. 
The only states from which a process $p$ can control any behavior (in terms of controllable actions) are of the form $q \in \places_\psys$.
$\sigma$ decides for every system place which transitions to enable. 
Due to our construction, $p$ can copy the decision of $\sigma$ by choosing an appropriate commitment set.
Therefore, $\varrho_\sigma$ allows the same behavior as $\sigma$.
If $\sigma$ is deterministic then the commitments sets are chosen such that no \Lightning-actions are possible. 

\noindent\textbf{\textsf{Translating a Controller for $\cGame_\pGame$ to a Strategy for $\pGame$~~}}
Given a winning controller $\varrho$, we incrementally construct a strategy $\sigma_\varrho$.
Every system place $q$ in the partially constructed strategy can control which transitions are enabled. 
The place $q$ belongs to some process. 
If on a state corresponding to a system place, this process can control all actions using its commitment sets. $q$ enables exactly the transitions that the process has chosen as a commitment set. 
An environment place cannot control any behavior and neither can the process it belongs to.
Hence, $\varrho$ and $\sigma_\varrho$ allow the same actions and transitions. 
A winning controller for $\widehat{\cGame_\pGame}$ additionally avoids any uncontrollable \Lightning-actions and results in a deterministic strategy.

For a detailed translation of strategies and controllers, we refer to Appendix \ref{sec:firstDir} where we formally prove the following:
\begin{theorem}
	\label{theo:theor1}
	$\pGame$ and $\cGame_\pGame$ are strategy-equivalent.
	$\pGame$ and $\widehat{\cGame_\pGame}$ are strategy-equivalent if we require deterministic Petri game strategies.
\end{theorem}

\subsection{Generalization to Concurrency-Preserving Games}
\label{sec:pg-cp}
Our translation builds processes from a slice distribution of the Petri game. 
This limits the translation to sliceable games.
The notion of slices is too strict: 
Our translation only requires to distribute the global behavior of the Petri game into local behavior, a \emph{partitioning} of the places is not necessarily needed. 
We introduce the new concept of \emph{singular nets} (SN). 
Similar to a slice, an SN describes the course of one token. 
Instead of being a subnet, it is equipped with a labeling function assigning to each node in the singular net a node in the original net. 
This labeling allows us to split up places and transitions by equally labelled copies enabling us to distribute every concurrency-preserving Petri net and game into singular nets. 
We can build our previous translation with an SN-distribution instead of a slice-distribution.
This allows us to generalize our result by not restricting to sliceable games (c.f., \refTheo{theor1}):

\begin{theorem}
	\label{theo:theor2}
	For every concurrency-preserving Petri game $\pGame$, there exist control games $\cGame_\pGame$ and $\widehat{\cGame_\pGame}$ with an equal number of players such that {\normalfont\textbf{\textsf{(1)}}} $\pGame$ and $\cGame_\pGame$ are strategy-equivalent and {\normalfont\textbf{\textsf{(2)}}} $\pGame$ and $\widehat{\cGame_\pGame}$ are strategy-equivalent if we require deterministic Petri game strategies. 
\end{theorem}
For a detailed discussion of singular nets and a sketch of the generalization, we refer the reader to Appendix \ref{sec:appSND}.

\subsection{Lower Bound}

We can show that there is a family of Petri games such that every strategy-equivalent control game must have exponentially many local states.
In a control game, either all or none of the players can restrict an action. 
By contrast, Petri games offer a finer granularity of control by allowing only some players to restrict a transition. 
The insight for the lower bound is to create a situation where a transition is shared between players but can only be controlled by one of them.
Using careful reasoning, we can show that in any strategy-equivalent control game there must be actions that can only be controlled by a single process, resulting in exponentially many local states.
Our translation shows that the difference between both formalism can be overcome but our lower bound shows an intrinsic difficulty to achieve this. 

\begin{theorem}
	\label{theo:theor3}
	There is a family of Petri games such that every strategy-equivalent control game (with an equal number of players) must have at least $\Omega(d^n)$ local states for $d > 1$.
\end{theorem}
The proof of the lower bound can be found in Appendix \ref{sec:appLB1}.

\section{Translating Control Games to Petri Games}
\label{sec:AAtoPG}

We give our translation from control games to Petri games, prove that it yields strategy-equivalent (and therefore winning-equivalent) games, and give an exponential lower bound. 
We present our translation for safety objectives.

\subsection{Construction}

\begin{figure}[t]
	\begin{tcolorbox}[colback=white, colframe=myYellow, arc=3mm, boxrule=1mm]
		\fontsize{8.5}{10.5}
		%\small
		Define $\pGame_\cGame = (\places_\psys, \places_\penv, \transitions, \flow, \init, \mathcal{B})$ where
		\def\arraystretch{1.2}
		\begin{itemize}
			\item $\places_\psys = \bigcup_{p \in  \processes} S_p$, \, $\places_\penv =  \{\, (s, A) \,\mid s \in \bigcup_{p \in  \processes} S_p,\; A \subseteq \enabled{s} \cap \Sigma^{\sys} \} \color{gray} \,\cup \, \{\badw_\DL^p \mid p \in \processes\}$
			
			\item $\transitions = \begin{aligned}[t]
					&\{ \, (a, B, \{A_s\}_{s \tuplein B}) \, \mid a \in \Sigma^{\env}  ,\; B \in \domain(\delta_a),\; A_s \subseteq \enabled{s} \cap \Sigma^{\sys}  \} \; \cup &\text{\footnotesize\textbf{\textsf{(1)}}}\\
					&\{ \, (a, B, \{A_s\}_{s \tuplein B}) \, \mid a \in \Sigma^{\sys},\; B \in \domain(\delta_a), \; A_s \subseteq \enabled{s} \cap \Sigma^{\sys},\; a \in A_s \} \; \cup &\text{\footnotesize\textbf{\textsf{(2)}}}\\
					&\{\,\tau_{(s, A)} \,\mid s \in \textstyle\bigcup\nolimits_{p \in  \processes} S_p, \; A \subseteq \enabled{s} \cap \Sigma^{\sys} \} \color{gray}\; \cup \;\{t_\DL^M \mid M \in \mathfrak{D}_\DL \} &\text{\footnotesize\textbf{\textsf{(3), \textcolor{gray}{(7)}}}}
				\end{aligned}$
			\item $\flow = \begin{aligned}[t]
				&\{ \, \big( \, (s, A) \, , \,(a, B, \{A_s\}_{s \tuplein B}) \,\big) \mid s \tuplein B,\; A_s = A  \} \; \cup \hspace{4.573cm} &\text{\footnotesize\textbf{\textsf{(4)}}} \\ 
				&\{ \, \big(\,(a, B, \{A_s\}_{s \tuplein B})\,,\, s' \,\big) \mid s' \tuplein \delta_a(B) \} \; \cup &\text{\footnotesize\textbf{\textsf{(5)}}} \\
				&\{\,  \big(s, \tau_{(s, A)}\big) \} \; \cup \; \{ \big(\tau_{(s, A)}, (s, A) \big)\, \} \color{gray} \;\cup &\text{\footnotesize\color{black}\textbf{\textsf{(6)}}}\\
				&\color{gray}\{ \,\big(q, t_\DL^M\big)\, \mid q \in M \} \; \cup \; \{\, \big(t_\DL^M, \badw_\DL^p\big)\, \mid p \in \processes \} &\text{\footnotesize\textbf{\textsf{\textcolor{gray}{(8)}}}}
			\end{aligned}$ 
			
			\item $\init = s_{in}^\AA$ and $\mathcal{B} = \bigcup_{p \in  \processes} \mathcal{B}_p \color{gray} \; \cup \; \bigcup_{p \in  \processes} \badw_\DL^p$
		\end{itemize}
	\end{tcolorbox}
	\caption{We give the construction of the translated Petri game $\pGame_\cGame$ for a control game $\CG = (\AA, \Sigma^{\sys}, \Sigma^{\env}, \{\bad_p\}_{p \in \processes})$ where $\AA = (\{S_p\}_{p \in \processes}, s_{in}^\AA, \{\delta_a\}_{a \in \Sigma})$. The initial state $s_{in}^\AA$ is viewed as a set.
	The gray parts penalize the artificial deadlocks, i.e., all markings in $\mathfrak{D}_\DL$.
		}
	\label{fig:secReduction}
\end{figure}

We fix a control game $\CG = (\AA, \Sigma^{\sys}, \Sigma^{\env}, \{\bad_p\}_{p \in \processes})$ with safety objective.
The translation to $\pGame_\cGame$  is depicted in \refFig{secReduction}.
We represent each local state $s$ as a system place.
We add environment places $(s, A)$, which encode every possible commitment set of actions that can be allowed by a controller ($A \subseteq \enabled{s} \cap \Sigma^{\sys}$).
From each system place, the player can move to places for the commitment sets using a $\tau_{(s, A)}$-transition \textbf{\textsf{(3, 6)}}.
Each action $a$ in $\cGame$ can occur from different configurations of the processes in $\dom(a)$, i.e., all states in $\domain(\delta_a)$, whereas in Petri games transitions fire from fixed preconditions. 
We want to represent $a$ as a transition that fires from places representing commitment sets that correspond to configurations from which $a$ can occur in $\cGame$.
We hence duplicate $a$ into multiple transitions to account for every configuration in $\domain(\delta_a)$ and for every combination of commitment sets.
Transitions have the form $(a, B, \{A_s\}_{s \tuplein B})$ where $a$ is the action in the control game, $B \in \domain(\delta_a)$ is the configuration from which $a$ can fire, and $\{A_s\}_{s \tuplein B}$ are the involved commitment sets.
If action $a$ is uncontrollable the corresponding transitions are added independently of the commitment sets \textbf{\textsf{(1)}}. 
If $a$ is controllable a transition is only added if $a$ is in the commitment sets of all involved players, i.e., $a \in A_s$ for every $s \tuplein B$ \textbf{\textsf{(2)}}. 
If $(a, B, \{A_s\}_{s \tuplein B})$ is added it fires from precisely the precondition that is encoded in it, i.e., the places $(s, A)$ where $s \in B$ and $A_s = A$, and moves every token to the system place that corresponds to the resulting local state when firing $a$ in $\cGame$ \textbf{\textsf{(4, 5)}}. 
A strategy can restrict controllable actions by moving to an appropriate commitment set but cannot forbid uncontrollable ones, since they can occur from every combination of commitment sets.
If a system player decides to refuse any commitment set it could prohibit transitions that correspond to uncontrollable actions.
In \refSection{enforecommitment}, we show how to force the system to always choose a commitment set.

In safety games, every winning strategy must avoid deadlocks. 
By introducing explicit commitment sets, we add \emph{artificial} deadlocks, i.e., configurations that are deadlocked in $\pGame_\cGame$ but where the corresponding state in $\cGame$ could still act. 
This permits trivial strategies that, e.g., always choose the empty commitment set.
We define $\mathfrak{D}_\DL$ as the set of all reachable markings that are final in $\pGame_\cGame$ but where the corresponding global state in $\cGame$ can still perform an action, i.e., all artificial deadlocks.
Similar to the \Lightning-actions, we introduce $t^M_\DL$-transitions that fire from every marking $M$ in $\mathfrak{D}_\DL$ and move every token to a losing place $\badw_\DL$ \textbf{\textsf{(7, 8)}}.
The mechanism to detect artificial deadlocks is depicted as the gray parts in \refFig{secReduction}.
For a formal construction of the deadlock detection mechanisms, we refer the reader to Appendix~\ref{sec:appAD}.

Figure \ref{fig:redAAtoPG} depicts an example translation. 
The system cannot win this game: The uncontrollable action $b$ can always happen, independent of the commitment set for place $A$. If one of the two tokens refuses $c$ (moves to a blue place) a (losing) transition $t_\DL$ can fire.

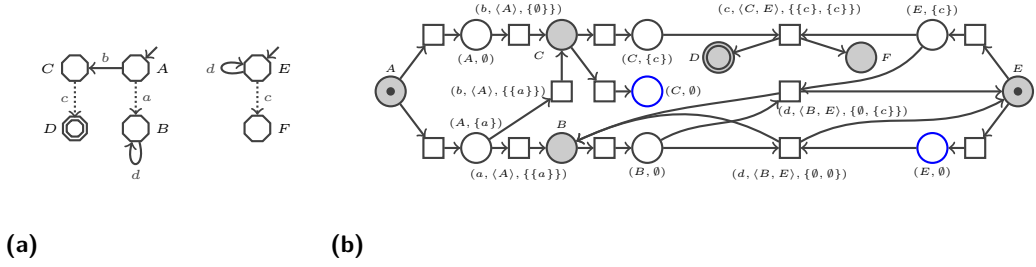
\begin{figure}[t!]
	\begin{subfigure}[t]{0.3\textwidth}
		\begin{center}
			\begin{tikzpicture}[scale=0.8, every label/.append style={font=\tiny}, label distance=-0.5mm]
			\node[aastate, label=east:\tiny$A$] at (0,0) (s11) {};
			\node[aastate, label=west:\tiny$C$] at (-1,0) (s12) {};
			\node[aastate, label=east:\tiny$B$] at (0,-1) (s13) {};
			\node[specialAA, label=west:\tiny$D$] at (-1,-1) (s14) {};
			
			\node[aastate, label=east:\tiny$E$] at (2,0) (s21) {};
			\node[aastate, label=east:\tiny$F$] at (2,-1) (s22) {};
			
			\draw[arrow] (s11)+(0.35, 0.35) to (s11);
			\draw[arrow] (s11) to node[above=-0.5mm] {\tiny$b$} (s12);
			\draw[arrow, densely dotted] (s11) to node[right=-0.5mm] {\tiny$a$} (s13);
			\draw[arrow, densely dotted] (s12) to node[left=-0.5mm] {\tiny$c$} (s14);
			
			\draw[arrow,loop below] (s13) to node[below=-0.5mm] {\tiny$d$} (s13);
			
			\draw[arrow] (s21)+(0.35, 0.35) to (s21);
			\draw[arrow, densely dotted] (s21) to node[right=-0.5mm] {\tiny$c$} (s22);
			
			\draw[arrow,loop left] (s21) to node[left=-0.5mm] {\tiny$d$} (s21);
			\end{tikzpicture}
		\end{center}
		\subcaption{}
	\end{subfigure}
	\begin{subfigure}[t]{0.7\textwidth}
		\begin{center}
			\begin{tikzpicture}[scale=0.75, every label/.append style={font=\fontsize{4}{4}}, label distance=-0.8mm]
			\node[sysplace,label=north:$A$] at (0,0) (p11){};
			\node[envplace,label=south:{$(A, \emptyset)$}] at (1.5,1) (p12){};
			\node[envplace,label=north:{$(A, \{a\})$}] at (1.5,-1) (p13){};
			
			\node[sysplace,label=225:$C$] at (3,1) (p14){};
			\node[sysplace,label=north:$B$] at (3,-1) (p15){};
			\node[envplace,label=south:{$(B, \emptyset)$}] at (4.5,-1) (pt1){};
			
			\node[envplace,label=south:{$(C, \{c\})$}] at (4.5,1) (p16){};
			\node[envplace,draw=blue,label={[label distance=-1.1mm] east:{$(C, \emptyset)$}}] at (4.5,0) (p17){};
			
			\node[sysplace, specialSys,label=west:$D$] at (5.75,0.6) (p18){};

			\node[transition] at (0.75, 1) (t11){};
			\node[transition] at (0.75, -1) (t12){};
			
			\node[transition,label={[label distance=0mm] north:{$(b, \langle A \rangle, \{ \emptyset\} \})$}}] at (2.25, 1) (t13){};
			\node[transition,label=left:{$(b, \langle A \rangle, \{ \{a\} \})$}] at (3, 0) (t14){};
			
			\node[transition,label={[label distance=0mm] south:{$(a, \langle A \rangle, \{ \{a\} \})$}}] at (2.25, -1) (t15){};
			
			\node[transition] at (3.75, 1) (t16){};
			\node[transition] at (3.75, 0) (t17){};
			
			\node[transition] at (3.75, -1) (tt1){};
			
			\draw[arrow] (p11) to (t11);
			\draw[arrow] (p11) to (t12);
			
			\draw[arrow] (t11) to (p12);
			\draw[arrow] (t12) to (p13);
			
			\draw[arrow] (p12) to (t13);
			
			\draw[arrow] (p13) to (t14);
			\draw[arrow] (p13) to (t15);
			
			\draw[arrow] (t13) to (p14);
			\draw[arrow] (t14) to (p14);
			
			\draw[arrow] (t15) to (p15);
			
			\draw[arrow] (p14) to (t16);
			\draw[arrow] (p14) to (t17);
			
			\draw[arrow] (t16) to (p16);
			\draw[arrow] (t17) to (p17);
			
			\node[token] at (0,0) (){};
			
			\draw[arrow] (p15) to (tt1);
			\draw[arrow] (tt1) to (pt1);
			
			%%%%%
			\node[sysplace,label=north:$E$] at (11, 0) (p21){};
			
			\node[envplace,draw=blue,label=south:{$(E, \emptyset)$}] at (9.5, -1) (p22){};
			\node[envplace,label=north:{$(E, \{c\})$}] at (9.5, 1) (p23){};
			
			\node[sysplace,label=east:$F$] at (8.25, 0.6) (p24){};
			
			\node[transition] at (10.25, -1) (t21){};
			\node[transition] at (10.25, 1) (t22){};
			
			\node[token] at (11,0) (){};
			
			\draw[arrow] (p21) to (t21);
			\draw[arrow] (p21) to (t22);
			
			\draw[arrow] (t21) to (p22);
			\draw[arrow] (t22) to (p23);

			\node[transition,label={[label distance=0mm]north:{$(c, \langle C, E \rangle, \{ \{c\}, \{c\} \})$}}] at (7, 1) (tcom){};
			
			\node[transition,label={[label distance=-0.8mm, xshift=7mm]south:{$(d, \langle B, E \rangle, \{ \emptyset, \{c\} \})$}}] at (7, 0) (tcom1){};
			
			\node[transition,label={[label distance=0mm]south:{$(d, \langle B, E \rangle, \{ \emptyset, \emptyset \})$}}] at (7, -1) (tcom2){};
			
			\draw[arrow] (p23) to (tcom);
			\draw[arrow] (p16) to (tcom);
			
			\draw[arrow] (tcom) to (p24);
			\draw[arrow] (tcom) to (p18);

			\draw[arrow] (pt1) to[out=30,in=220] (tcom1);
			\draw[arrow] (pt1) to (tcom2);
			\draw[arrow] (tcom1) to[out=190,in=25] (p15);
			%\draw[arrow] (tcom2) to[out=210,in=330] (p15);
			\draw[arrow] (tcom2) to[out=150,in=25] (p15);
			
			\draw[arrow] (p23) to[out=230,in=11] (tcom1);
			\draw[arrow] (p22) to (tcom2);
			\draw[arrow] (tcom1) to (p21);
			\draw[arrow] (tcom2) to[out=30,in=210] (p21);
			\end{tikzpicture}
		\end{center}
		\subcaption{}
	\end{subfigure}
	
	\caption{Control game $\cGame$ (a) and translated Petri game $\pGame_\cGame$ (b) are given. Commitment sets without outgoing transitions are omitted. The set of artificial deadlocks $\mathfrak{D}_\DL$ comprises every final marking that contains at least one blue place. The resulting $t_\DL^M$-transitions are omitted. 
	}
	\label{fig:redAAtoPG}
\end{figure}

\subsection{Correctness}

We show strategy-equivalence of $\cGame$ and $\pGame_\cGame$ by translating strategies (that always commit) and controllers between both of them. 
We observe that each token moves on the local states of one process and takes part in precisely the actions of the process.
At every point, a token hence possesses the same local information as the process. 
A token can restrict the controllable actions using the commitment sets but cannot restrict the uncontrollable ones. 
The token therefore has the same possibilities as the process counterpart.

\noindent\textbf{\textsf{Translating Controllers to Strategies~~}}
Given a winning controller $\varrho$, we incrementally build a (possibly infinite) winning, deterministic strategy $\sigma_\varrho$.
Every system place $q$ in a partially constructed strategy can choose one of the commitment sets. 
$q$ copies $\varrho$ by committing to exactly the actions that the process it belongs to has allowed. 
The commitment sets can only restrict controllable actions, as the process can. 
Hence, $\sigma_\varrho$ allows the same behavior as~$\varrho$.

\noindent\textbf{\textsf{Translating Strategies to Controllers~~}}
Given a winning, deterministic strategy $\sigma$, we construct a winning controller $\varrho_\sigma$.
A process $p$ that resides on a local state $s$ can decide which of the controllable actions should be allowed. 
Every token in $\sigma$ can decide for a commitment set and therefore implicitly chooses which controllable actions should be enabled. 
$p$ allows exactly the actions that $\sigma$ chooses as a commitment set. 
Both can only restrict controllable actions and, by copying, $\varrho_\sigma$ achieves the same behavior as $\sigma$.

For a formal translation, we refer the reader to Appendix \ref{sec:appsecondDir}. 
Under the assumption that any strategy for $\pGame_\cGame$ always commits (we will see that this is valid in \refSection{enforecommitment}), we can prove:

\begin{theorem}
	\label{theo:theor4}
	$\cGame$ and $\pGame_\cGame$ are strategy-equivalent.
\end{theorem}

\subsection{Enforcing Commitment}
\label{sec:enforecommitment}

Our construction assumes wining strategies to always choose a commitment set. 
We can modify $\pGame_\cGame$ such that every non-committing strategy cannot win.
The insight is to use the deadlock-avoidance of winning strategies.
Deadlocks define a \emph{global} situation of the game. 
To enforce commitment, we require \emph{local} deadlock-avoidance in the sense that every token has to choose a commitment set.
This is not prevented by global deadlock-avoidance, where, e.g., a single player being able to play locally enables every other player to refuse to commit without being deadlocked.
We reduce local to global deadlocks by adding transitions to challenge the players to have reached a local deadlock. 
Using challenge transitions, every player currently residing on a place that corresponds to a chosen commitment set moves to a terminating place. 
Every player that has chosen commitment sets can terminate, resulting in the players that are locally deadlocked to cause a global deadlock. Although the challenge is always possible, the scheduler decides the point of challenge.
The game with the added challenger has a winning strategy iff $\pGame_\cGame$ has a winning strategy that always commits.
For a formal construction, we refer the reader to Appendix \ref{sec:appEC}.

\subsection{Lower Bounds}

We can provide a family of control games where every strategy-equivalent Petri game must be of exponential size. 
In control games, both controllable and uncontrollable actions can occur from the same state. 
In Petri games, a given place can either restrict all transitions (system place) or none. 
A control game where both actions types are possible already results in Petri games of exponential size.
We assume the absence of infinite $\tau$-sequences.

\begin{theorem}
	\label{theo:theor5}
	There is a family of control games such that every strategy-equivalent Petri game (with an equal number of players) must have at least $\Omega(d^n)$ places for $d > 1$.
\end{theorem}
The proof can be found in Appendix~\ref{sec:appLB2}.

\section{New Decidable Classes}
\label{sec:decidableClasses}

We exemplarily show one transferrable class of decidability for both control games and Petri games to highlight the applicability of our translations.

\noindent\textbf{\textsf{New Decidable Control Games~~}}
A process in a control game is an \emph{environment process} if all its action are uncontrollable. A \emph{system process} is one that is not an environment process. We can modify our second translation by not adding system places if there are no outgoing controllable actions. Therefore, environment processes do not add system places to the Petri game and we can use the results from \cite{DBLP:conf/fsttcs/FinkbeinerG17}.
\begin{corollary}
	Control games with safety objectives and one system process are decidable.
\end{corollary}

\noindent\textbf{\textsf{New Decidable Petri Games~~}}
Given a Petri game $\pGame$ and a distribution into slices (or SNs) $\{\slice\}_{\slice \in \slices}$, we analyze the communication structure between the slices by building the undirected graph $(V, E)$ where $V = \slices$ and $E = \{ (\slice_1, \slice_2) \mid \transitions^{\slice_1} \cap \transitions^{\slice_2} \neq \emptyset \}$.
$(V, E)$ is isomorphic to the \emph{communication architecture} of the constructed asynchronous automaton $\cGame_\pGame$ (as introduced in \cite{DBLP:conf/icalp/GenestGMW13}).
We define $\pGame_{\text{\Cancer}}$ as every Petri game that has a distribution $\{\slice\}_{\slice \in \slices}$ where ${(V, E)}$ is \emph{acyclic}. 
We can show that such distributions are hard to find. From \cite{DBLP:conf/icalp/GenestGMW13}, we obtain decidability.
\begin{lemma}
	Deciding whether a Petri net has an acyclic slice-distribution is \textsc{NP}-complete. 
\end{lemma}
The proof can be found in Appendix~\ref{sec:acyclicSliceNPcomplete}.
\begin{corollary}
	Petri games in $\pGame_{\text{\Cancer}}$ with reachability objectives are decidable.
\end{corollary}

\section{Conclusion}
\label{sec:conclusion}

We have provided the first formal connection between control games and Petri games by showing that both are equivalent.
This indicates that synthesis models for asynchronous systems with causal memory are stable under the concrete formalisms of system and environment responsibilities for the two most common models. 
Conversely, our lower bounds show an intrinsic difference between control games and Petri games. 
By our translations, existing and future decidability results can be combined and transferred between both game types.
Our translations could be adapted to other winning objectives. 
An interesting direction for future work is to investigate how action-based control games~\cite{muscholl2009look} relate to Petri games and to study unified models that combine features from control games and Petri games.

%%
%% Bibliography and Appendix
%%

\bibliography{lipics-v2019-sample-article}

\newpage

\appendix

\section{Preliminaries}
\label{app:background}

For convenience, we overload notation: 
The transitions in a branching process of a Petri game are not the ones from the game but are merely equipped with a $\lambda$-label to them. 
Writing $\fireTranTo{M}{t}{M'}$ for some marking $M$ in a branching process and transitions $t$ in the underlying game is therefore not defined. 
Unless for very specific occasions, we are, however, not interested in the precise transition in a branching process but solely for the label of it. 
In the proofs and notions defined below, we hence always identify transitions in the branching process (strategy) with the original ones. 
$\fireTranTo{M}{t}{M'}$ should therefore be understood as: There is a transition $t'$ in the branching process with $\fireTranTo{M}{t'}{M'}$ and $\lambda(t') = t$. 
Using this notational shortcut, we can, for instance, write $\fireSeq{\pNet^\mathfrak{U}}{\kappa}$ for a sequence $\kappa$ of transitions in the underlying game. 
It should be noted that this notional shortcut is not well-defined for arbitrary Petri nets; there could be multiple equally labelled transitions enabled from the same marking in the branching process. 
For branching processes of safe games (and therefore of sliceable games), however, there is at most one transition with a matching $\lambda$-label enabled. 
%In the following we hence never talk about transitions of a branching process and never apply $\lambda$ to transitions. 

\paragraph*{Partially Ordered Sets}
Recall that a partially ordered set (poset) is a pair $(\mathcal{X}, \leq)$ where $\leq$ is a partial order on elements from $\mathcal{X}$. 
We introduce a labelled partially ordered set as a triple $(\mathcal{X}, \leq, \beta)$ where $(\mathcal{X}, \leq)$ is a poset and $\beta: \mathcal{X} \to \mathcal{Y}$ labels the elements from $\mathcal{X}$ in some set $\mathcal{Y}$. 
Two posets $(\mathcal{X}_1, \leq_1)$ and $(\mathcal{X}_2, \leq_2)$ are \emph{isomorphic} if there is a bijection $g$ between $\mathcal{X}_1$ and $\mathcal{X}_2$ such that for all $x, y \in \mathcal{X}_1: x \leq_1 y \Leftrightarrow g(x) \leq_2 g(y)$. In the literature, such a function is referred to as an \emph{order isomorphism}. 
Two labelled posets $(\mathcal{X}_1, \leq_1, \beta_1)$ and $(\mathcal{X}_2, \leq_2, \beta_2)$ that are labelled in the same set $\mathcal{Y}$ are \emph{isomorphic} if there exists an order isomorphism $g$ between $\mathcal{X}_1$ and $\mathcal{X}_2$ where $\forall x \in \mathcal{X}_1: \beta_1(x) = \beta_2(g(x))$, i.e., the labels agree.
We call two posets \emph{equal} and write ``$=$'' between them if they are isomorphic.

\subparagraph{Causal Past as Partially Ordered Sets}

Both game formalisms allow for strategies that depend on causal information but do so in different ways. For Petri games, the restriction of the branching process (justified refusal) enforces that causal memory is obeyed whereas control games allow decisions based on an explicit local view on the previous play. 
Partially ordered sets are a natural representation of concurrent execution, i.e., sequences of events (in our case transitions or actions) that are not executed subsequently but can be interleaved to a certain degree. 
In Petri games, the causal information is represented as the causal past of a place. If this place is a system place it must make a decision of what transitions to allow solely based on this causal past (otherwise it would violate justified refusal). The causal past of a place hence characterizes the causal information of a player on that place. 
In control games, the causal past is represented as a local view on a previous play. 
Both the causal past of a place and the local view on a play can be characterized precisely using posets. This gives us a way to compare partial information between Petri games and control games despite their substantially different formulations.

We call a trace $u$ \emph{prime} if all linearizations of $u$, i.e., all sequences in the equivalence class, end with the same action. For a prime trace $u$, $\mathit{last}(u)$ denotes the last action. Note that $\mathit{last}$ is only defined for prime traces. 
A trace $u$ is a prefix of $w$, denoted $u \sqsubseteq w$ if there are linearizations $u'$ of $u$ and $w'$ of $w$ with $u' \sqsubseteq_{Seq} w'$. 
Here, $\sqsubseteq_{Seq}$ denotes the usual prefix relation on sequences. 
\begin{itemize}
	\item For the causal past of place $q$ in an branching process, we can define the labelled poset $(\pastt{q}, \leq, \lambda)$ where $\leq$ is the causal dependency relation and $\lambda$ the homomorphism associated to each branching process. 
	\item For a trace $u$, we can define the labelled poset $(\mathit{Pre}^{\mathit{prime}}(u), \sqsubseteq, \mathit{last})$ where $\mathit{Pre}^{\mathit{prime}}(u)$ are all primed prefixes of $u$, $\sqsubseteq$ is the prefix relation, and $\mathit{last}$ labels each prefix with its last action.
\end{itemize}
For both Petri games and control games, the poset representation is an intuitive concept to represent the causal past. Note that the poset of a trace describes the dependency between the actions.

\section{Translating Petri Games to Control Games}
\label{sec:firstDir}

\subsection{Proving Strategy-Equivalence}
\label{sec:formal1}

In this section, we discuss causal information in both game types. 
Afterwards, we give a detailed translation of strategies and controllers and derive a proof of \refTheo{theor1}.

\subparagraph{On the relation $\relation$}

Any state in $\bigcup_{\slice \in \slices} Q_\slice \,\setminus\, \{\badw_\slice\}$ corresponds to a place in $\pGame$ in the natural way. This correspondence is formalized by $\zeta$ where:
\begin{align*}
\zeta (q) &= q\\
\zeta ( \,(q, A)\, ) &= q
\end{align*}
We extend $\zeta$ to global states by defining for each global state $\{q_p\}_{p \in \processes}$ a corresponding marking by: $\zeta (\{q_p\}_{p \in  \processes}) = \bigcup_{p \in  \processes} \{\zeta(q_p) \}$.
For a process $p$, we define the shortcut $\pts{p}$ for the slice that $p$ has been build from\footnote{In the construction, $p$ is exactly this slice. However, having explicit notion is more convenient.}. 
Conversely, for a slice $\slice$, $\stp{\slice}$ denotes the process that is build from $\slice$.
By definition of $\cGame_\pGame$, we have that $\transitions \subseteq \Sigma$. For a sequence of actions $u \in \Sigma^*$, we denote the projection on $\transitions$ by $\projtr{u}$. It is defined by:
\begin{align*}
\projtr{\epsilon} &= \epsilon \\
\projtr{u \, \tau} &= \projtr{u} \\
\projtr{u \, t} &= \projtr{u} \, t \quad \text{if } t \in \transitions
\end{align*}
We can now formalize the relation $\relation \subseteq \reach(\pGame^\mathfrak{U}) \times \Plays(\cGame_\pGame)$ by defining:
\begin{tcolorbox}[colback=white]
	\centering
	$M \relation u$ iff $\fireSeq{\pGame^\mathfrak{U}}{\projtr{u}} = M$
\end{tcolorbox}
\noindent This captures the idea that a marking and play are similar/related if they are reached with the same observable trace. $\fireSeq{\pGame^\mathfrak{U}}{\projtr{u}}$ should be understood as firing any linearization of $\projtr{u}$. 
We hence need to prove that $\relation$ is \emph{well-defined}, i.e., $\fireSeq{\pGame^\mathfrak{U}}{\projtr{u}}$ is invariant under elements of the equivalence class $u$. 
Since the actions in $\cGame_\pGame$ are constructed from transitions they inherit the dependency from the transition. 
If two actions are independent the corresponding transitions are concurrent in the Petri net and can be executed in any order:
\begin{lemma}
	\label{lem:wellDefined}If $\fireSeq{\pGame^\mathfrak{U}}{\projtr{u}} = M$ for $u \in \Sigma^*$ and $u \sim_\ind w$ for some $w \in \Sigma^*$ then $\fireSeq{\pGame^\mathfrak{U}}{\projtr{w}} = M$.
\end{lemma}
\begin{proof}
	If actions $t_1, t_2 \in \transitions$ are independent in $\cGame_\pGame$ they belong to different slices (by definition of the dependency relation), so $(\pre{\pGame}{t_1} \cup \post{\pGame}{t_1}) \cap (\pre{\pGame}{t_2} \cup \post{\pGame}{t_2}) = \emptyset$. 
	Swapping $t_1$ and $t_2$ hence results in the same marking in the unfolding $\pGame^\mathfrak{U}$. 
	The claim follows by induction on the number of swaps in the proof of $u \sim_\ind w$. 
\end{proof}

In our construction, every place in the Petri game is represented as possibly many states in the control games. 
These additional copies, used to represent commitment sets, are equipped with the same $\zeta$ label.
Every observable action $t$ in the control game precisely captures the movement of the tokens involved in $t$. %That is $t$ involves process $p$ if and only if slice $\pts{p}$ takes part in $t$.
We hence see that for a related marking and play the underlying net/automaton is in an equally labelled state: 
\begin{lemma}
	If $M \relation u$ then $\zeta(\state{u}) = \lambda[M]$.
	\label{lem:PtoC_sameLabel}
\end{lemma}
\begin{proof}
	Follows by induction on the length of $u$ using the fact that for all $t \in \transitions$ and all $B \in \domain(\delta_t)$ it holds that $\zeta(B) = \pre{\pGame}{t}$ and $\zeta(\delta_t(B)) = \post{\pGame}{t}$.
\end{proof}

\subparagraph{Causal Information Flow}

In our construction, we represent each slice as a distinct process. 
The actions of a process $p$ ($\Sigma_p$) are precisely the transitions that $\pts{p}$ is involved in (and additional $\tau$-actions). 
Now consider a marking $M$ and play $u$ where $M \relation u$. By construction, firing the observable action from $u$ in the unfolding results in $M$.
The marking $M$ and trace~$u$ do not only represent the global state of the system but also include the local information of each token or process. 
The crucial observation of our translation is that this information is ``the same''. 
The local view of process $p$ on $u$ is the same as the causal past of the token in $M$ from slice $\pts{p}$.
This holds as in $\cGame_\pGame$ the communication behavior of $\pGame$ is modeled truthfully. Every process hence participates in exactly the actions that its slice takes part in. 
For our translation, we need a more formal notion of what ``having the same information'' means. 
We thus need to find a way to relate causal information between both game types. Unfortunately, Petri games and control games represent causal information in a fundamentally different way utilizing either the causal past of a place or the local view on a play.

If we consider a play $u$ and the poset representation of $\projtr{u}$ we observe that the poset is labelled in $\transitions$. For any place $q$ in the unfolding of a $\pGame$, the poset is also labelled in $\transitions$. 
This allows us to express equality between the causal past of a place and a trace. 
We can, for instance, write $\pastt{q} = \projtr{u}$, which should be understood as the fact that both sides have equal poset representations, i.e., the labelled poset representations for $\pastt{q}$ and $\projtr{u}$ are isomorphic. 
Labelled posets thus allow us to compare the causal information between both game types.  

We can now state the following result which gives us a direct characterization of the local information of individual player. It tells us that in $\relation$-related situations, the local view of each process aligns with the causal past of the corresponding place. 

\begin{lemma}
	If $M \relation u$ and $q \in M \cap \inv{\lambda}[\places^{\pts{p}}]$ (for some $p \in \processes$) then $\pastt{q} = \projtr{\view_p(u)}$
	\label{lem:PtoC_sameCausalPast}
\end{lemma}
\begin{proof}
	From $M \relation u$, we conclude that $\fireSeq{\pGame^\mathfrak{U}}{\projtr{u}} = M$.
	The simulation is invariant under elements from $u$ as we argued in \refLemma{wellDefined}.\\
	The local view of $p$ on $u$ is defined as the smallest trace $[v]_\ind$ such that $u \sim_\ind v \, w$ for some $w$ that contains no actions from $\Sigma_p$.
	We can hence write $u = \view_p(u) \, w$.
	%The local view of $p$ on $u$ is defined as the smallest prefix s.t.\ $u = \view_p(u) \, w$ and $w$ contains no actions from $\Sigma_p$.
	%It is easy to see that $\projtr{(\view_p(u))} = \view_p(\projtr{u})$
	Since the $\tau$-actions are local to one process it holds that $\projtr{u} = \projtr{\view_p(u)} \, \projtr{w}$ and $\projtr{w}$ contains no actions from $\Sigma_p$. We hence obtain that $$\fireSeq{\pGame^\mathfrak{U}}{\projtr{u}} = \fireSeq{\pGame^\mathfrak{U}}{\projtr{\view_p(u)} \, \projtr{w}}$$
	and, in particular,
	$$\fireSeq{\pGame^\mathfrak{U}}{\projtr{u}} \cap \inv{\lambda}[\places^{\pts{p}}] = \fireSeq{\pGame^\mathfrak{U}}{\projtr{\view_p(u)} \, \projtr{w}} \cap \inv{\lambda}[\places^{\pts{p}}]$$
	The observable actions in $\Sigma_p$ are exactly the transitions that the slice $\pts{p}$ is involved in. $\projtr{w}$ contains no actions from $\Sigma_p$ and therefore contains no transitions that involve $\pts{p}$.
	We can hence see that 
	\begin{align*}
	M \cap \inv{\lambda}[\places^{\pts{p}}] &= \fireSeq{\pGame^\mathfrak{U}}{\projtr{u}} \cap \inv{\lambda}[\places^{\pts{p}}]\\
	&= \fireSeq{\pGame^\mathfrak{U}}{\projtr{\view_p(u)} \, \projtr{w}} \cap \inv{\lambda}[\places^{\pts{p}}] \\
	&= \fireSeq{\pGame^\mathfrak{U}}{\projtr{\view_p(u)}} \cap \inv{\lambda}[\places^{\pts{p}}]
	\end{align*}
	Firing $\projtr{\view_p(u)}$ and firing $\projtr{u}$ results in the same place for slice $\pts{p}$.
	We later recover exactly this statement (\refLemma{PtoC_localeViewSamePlace}) from our current lemma.\\
	We next show that $\fireSeq{\pGame^\mathfrak{U}}{\projtr{\view_p(u)}} = \fireSeq{\pGame^\mathfrak{U}}{\pastt{q}}$, i.e., firing the transitions in the causal past of $q$ results in the same marking as firing $\projtr{\view_p(u)}$. 
	Note that by definition every linearization of $\pastt{q}$ results in the same marking, so, $\fireSeq{\pGame^\mathfrak{U}}{\pastt{q}}$ is well-defined. 
	It trivially holds that $\fireSeq{\pGame^\mathfrak{U}}{\pastt{q}} \cap \inv{\lambda}[\places^{\pts{p}}] = M \cap \inv{\lambda}[\places^{\pts{p}}]$, so we get that
	\begin{align*}
	\fireSeq{\pGame^\mathfrak{U}}{\projtr{\view_p(u)}} \cap \inv{\lambda}[\places^{\pts{p}}] = \fireSeq{\pGame^\mathfrak{U}}{\pastt{q}} \cap \inv{\lambda}[\places^{\pts{p}}] \tag*{\myItem{(1)}}
	\end{align*}
	We want to show the more general statement that not only the place that belongs to process~$p$ is shared in $\fireSeq{\pGame^\mathfrak{U}}{\projtr{\view_p(u)}}$ and $\fireSeq{\pGame^\mathfrak{U}}{\pastt{q}}$ but the place of every process.  \\
	We can first observe that $\pastt{q}$ is the smallest set of transitions that needs to fire to reach~$q$. As soon as we remove a single transition from the set, the simulation will no longer reach place $q$. 
	From \myItem{(1)}, we get that simulating $\projtr{\view_p(u)}$ also results in place $q$. Simulating $\projtr{\view_p(u)}$ instead of $\pastt{q}$ therefore results in a marking that has progressed more, i.e., a marking where the game has progressed further \myItem{(2)}. \\
	We assume for contradiction that $\fireSeq{\pGame^\mathfrak{U}}{\projtr{\view_p(u)}} \neq \fireSeq{\pGame^\mathfrak{U}}{\pastt{q}}$.
	There hence is a process $p'$ with 
	$$\fireSeq{\pGame^\mathfrak{U}}{\projtr{\view_{p}(u)}} \cap \inv{\lambda}[\places^{\pts{p'}}] \neq \fireSeq{\pGame^\mathfrak{U}}{\pastt{q}} \cap \inv{\lambda}[\places^{\pts{p'}}]$$
	Let $q_1$ and $q_2$ be the unique places with 
	\begin{align*}
	q_1 &\in \fireSeq{\pGame^\mathfrak{U}}{\projtr{\view_{p}(u)}} \cap \inv{\lambda}[\places^{\pts{p'}}] \\
	q_2 &\in \fireSeq{\pGame^\mathfrak{U}}{\pastt{q}} \cap \inv{\lambda}[\places^{\pts{p'}}]
	\end{align*}
	By assumption $q_1 \neq q_2$ and from \myItem{(2)}, it is easy to see that $q_2 < q_1$, i.e., the token of slice $\pts{p'}$ has progressed further when firing $\projtr{\view_{p}(u)}$ instead of $\pastt{q}$.\\
	Let $t$ be the unique transition in $\pre{\pGame^\mathfrak{U}}{q_1}$. It holds that $q_2 < t < q_1$. 
	We know that $t$ must be included in $\view_p(u)$ and since $t$ has no successor transitions we observe that $\view_p(u) = r \, t$ \myItem{(3)} for some play $r$, i.e., there is a linearization of $\view_p(u)$ that ends with~$t$. 
	Since $t$ does not involve the token from slice $\pts{p}$ we can conclude that $t \not\in \Sigma_p$.
	\myItem{(3)} is, however, a contradiction to the minimality of $\view_p(u)$.\\
	Hence, $\fireSeq{\pGame^\mathfrak{U}}{\projtr{\view_p(u)}} = \fireSeq{\pGame^\mathfrak{U}}{\pastt{q}}$. 
	If two transitions in $\pastt{q}$ are unordered they are independent in $\projtr{\view_p(u)}$. Conversely, consecutive independent actions in $\projtr{\view_p(u)}$ involve disjoint sets of slices and are hence unordered in $\pastt{q}$.
	It is therefore easy to see that $\pastt{q} = \projtr{(\view_p(u))}$.
	\vspace{0.1cm}
\end{proof}

 \refLemma{PtoC_sameCausalPast} tells us that our relation $\relation$ does not only capture the global configuration of both game types (as stated in \refLemma{PtoC_sameLabel}) but also respects the local information. 
This is of tremendous importance for a translation of strategies/controller. 
If $M \relation u$ then every process in $p$ possesses the same information (in terms of the local view on $u$) as the corresponding place in $M$ has (in terms of the causal past). 

\subsection{Translating Strategies to Controllers}
\label{sec:PtoC_translateStrtoCont}

In this section, we provide a formal translation of strategies to controllers. 
Given a winning strategy $\sigma$ for $\pGame$, 
we construct a winning controller $\varrho_\sigma = \{f^{\varrho_\sigma}_p\}_{p \in \processes}$ for $\cGame_\pGame$ and, furthermore, show that if $\sigma$ is deterministic, $\varrho_\sigma$ is a winning controller for $\widehat{\cGame_\pGame}$.
The description of the local controller $f^{\varrho_\sigma}_p$ for process $p$ is depicted in \refFig{PtoC_translation1}.

\begin{figure}[t!]
	\begin{tcolorbox}[colback=white, colframe=myYellow, arc=3mm, boxrule=1mm]
		For $p \in \processes$ and $u \in \Plays_p(\cGame_\pGame)$:
		\begin{enumerate}
			\item[\myItem{1.}] If $\statep{p}{u} \in \places_\penv$ all outgoing transitions are uncontrollable. Define $f^{\varrho_\sigma}_p(u) = \emptyset$. 
			
			\item[\myItem{2.}] If $\statep{p}{u} = (q, A)$ for some $q \in \places_\psys$ and $A \subseteq \post{\pGame}{q}$ all outgoing transitions are uncontrollable. Define $f^{\varrho_\sigma}_p(u) = \emptyset$.
			
			\item[\myItem{3.}] If $\statep{p}{u} \in \places_\psys$ we distinguish two cases
			\begin{enumerate}
				\item[\myItem{a)}] $\projtr{u}$ is a valid sequence of transitions in $\pNet^\sigma$:\\
				Let $M = \fireSeq{\pNet^\sigma}{\projtr{u}}$. There exists a unique place $q \in M \cap \inv{\lambda}[\places^{\pts{p}}]$. \\
				Define $f^{\varrho_\sigma}_p(u) = \{\; \tau_{(\statep{p}{u}, \lambda[A])} \; \}$ where $A = \post{\pNet^\sigma}{q}$.
				
				\item[\myItem{b)}] $\projtr{u}$ is no valid sequence of transition $\pNet^\sigma$:\\
				Define $f^{\varrho_\sigma}_p(u) = \emptyset$. \\
				\textit{This case will never occur if $u$ is a controller-compatible play.}
				
			\end{enumerate}
			\item[\myItem{4.}] If $\statep{p}{u} = \badw$ there are no outgoing transitions. Define $f^{\varrho_\sigma}_p(u) = \emptyset$.
		\end{enumerate}
	\end{tcolorbox}
	
	\caption{Description of local controller $f^{\varrho_\sigma}_p$ for process $p \in \processes$. The controller is build from a strategy $\sigma$ for $\pGame$ with branching process $\pNet^\sigma$. }
	\label{fig:PtoC_translation1}
\end{figure}

Every process $p$ in $\varrho_\sigma$ does what we described informally. 
Given a play $u \in \Plays_p(\cGame_\pGame)$, every process computes its current state. 
Only if this state corresponds to a system place of $\pGame$ (case \myItem{3.}) any controllable actions are available. In this case, the observable actions in $u$ are simulated in $\pNet^\sigma$, i.e., the branching process of $\sigma$.
In \refLemma{wellDefined}, we already argued that simulation of traces is well-defined, i.e., invariant under linearizations.
For an arbitrary $u$, $\projtr{u}$ might not be a valid sequence in the strategy. 
We therefore include case \myItem{b)} to obtain a total function $f^{\varrho_\sigma}_p$.
In case of a successful simulation (case \myItem{a)}), the simulation reaches some marking $M$. 
Now, $p$ should copy the decision of the strategy made in $M$. It therefore computes the place in $M$ that corresponds to the slice $p$ is build from, i.e., the place $q \in M \cap \inv{\lambda}[\places^{\pts{p}}]$. 
The set of transitions allowed by this place are $\post{\pNet^\sigma}{q}$.
To copy the decision, $p$ hence chooses the commitment set that contains exactly those transitions. 
We later show that for controller-compatible plays $u$,  $\projtr{u}$ is always a valid sequence, i.e., we never land in case \myItem{b)}.

\begin{example}
	As an example, we consider the translation from \refFig{firstRed} and the winning (non-deterministic) strategy $\dot{\sigma}$ depicted in \refFig{example_Strat}.
	Whenever possible, $\dot{\sigma}$ allows transition~$i$ to move to place $D$. 
	If in place $D$ for the first time the strategy allows communication on both $a$ and~$b$. 
	Upon communication on either $a$ or $b$, the strategy can, furthermore, deduce whether the environment played $e_1$ or $e_2$, as this information is conceptually transmitted in the communication. There are hence four different cases possible:
	In case of synchronization on~$a$, $\dot{\sigma}$ allows the token on a system place to move to $D$ using transition $i$. It then distinguishes whether $e_1$ or $e_2$ have been played. In case of $e_1$, it terminates and, in case of $e_2$, it allows communication on $b$ one more time. 
	If synchronization occurred on $b$ the strategy again distinguishes the two cases. If it can deduce $e_1$ the strategy allows $b$ one more time. In case of $e_2$, it terminates directly. 
	Even though $\dot{\sigma}$ seems unnecessary complicated\footnote{In the sense that there are much simpler winning strategies.}, strategy-equivalence requires us to build a controller that copies this behavior.
	We can now translate $\dot{\sigma}$ according to our translation of strategies and obtain a controller $\dot{\varrho}_\sigma$ for~$\dot{\cGame_\pGame}$.
	Since there is no intuitive way to represent a controller graphically, we depict $\dot{\varrho}_\sigma$ as a table that summarizes a selection of plays in $\dot{\cGame_\pGame}$ and the decision made by $p_2$ (the local controller~$f_{p_2}^{\dot{\varrho}_\sigma}$). 
	The table is shown in \refFig{possiblePlays}. As we only depict the decisions of $p_2$, we listed the $p_2$-view on all plays.  
	$\dot{\varrho}_\sigma$ initially allows action $i$ by choosing commitment set $(C, \{i\})$. Afterwards, it admits communication on both $a$ and $b$ by moving to commitment set $(D, \{a, b\})$. Then, $\dot{\varrho}_\sigma$ copies the ``case analysis'' of $\dot{\sigma}$. 
	We can observe that every decision of the controller is made in accordance with our construction.
	As an example, consider the play $[\tau_{(C, \{i\})}, i, \tau_{(D, \{a, b\})}, e_1, b]_\ind$ (play \myItem{(1)} in \refFig{possiblePlays}).
	The observable actions of that play comprise $e_1$, $i$, and $b$.
	The simulation of this play in $\pNet^{\dot{\sigma}}$ results in the red marking $M_1$. 
	Since the system place (the place in $M_1 \cap \inv{\lambda}[\places^{\pts{p_2}}]$) allows $b$ in its postcondition, the controller chooses $(D, \{b\})$ as a commitment set. 
	The interested reader is advised to convince herself that all decisions listed in \refFig{possiblePlays} are in accordance with both our construction and $\dot{\sigma}$.
	
	\begin{figure}[t]
		\centering
		
		\begin{tikzpicture}[scale=1.0, every label/.append style={font=\scriptsize}, label distance=-1mm]
		\node[envplace,label=east:$\color{myDGray}A$] at (0,0.75) (p1){};
		\node[transition,label=north:$\color{myDGray}e_1$] at (-3.5,-0.75) (t1){};
		\node[envplace,label=east:$\color{myDGray}B$] at (-3.5,-1.5) (p2){};
		\node[transition,label=north:$\color{myDGray}e_2$] at (3.5,-0.75) (t2){};
		\node[envplace,label=west:$\color{myDGray}B$] at (3.5,-1.5) (p3){};
		
		\node[sysplace,label=east:$\color{myDGray}C$] at (0,-0) (ps0){};
		\node[transition,label=east:$\color{myDGray}i$] at (0,-0.75) (ts0){};
		\node[sysplace,label=south:$\color{myDGray}D$] at (0,-1.5) (p4){};
		
		\node[transition,label=north:$\color{myDGray}b$] at (-5,-3) (t3){};
		\node[transition,label=north:$\color{myDGray}a$] at (-2,-3) (t4){};
		\node[transition,label=north:$\color{myDGray}a$] at (2,-3) (t5){};
		\node[transition,label=north:$\color{myDGray}b$] at (5,-3) (t6){};
		
		\node[envplace,label=south:$\color{myDGray}B$] at (-5.75,-3.75) (p5){};
		\node[sysplace,label=south:$\color{myDGray}D$] at (-4.25,-3.75) (p6){};
		\node[envplace,label=east:$\color{myDGray}A$] at (-2,-3.75) (p7){};
		\node[sysplace,label=east:$\color{myDGray}C$] at (-2,-4.5) (p8){};
		
		\node[envplace,label=west:$\color{myDGray}A$] at (2,-3.75) (p9){};
		\node[sysplace,label=west:$\color{myDGray}C$] at (2,-4.5) (p10){};
		\node[sysplace,label=south:$\color{myDGray}D$] at (4.25,-3.75) (p11){};
		\node[envplace,label=south:$\color{myDGray}B$] at (5.75,-3.75) (p12){};

		%%%%%%%%%%%%%%%%%%%%%%%%%%%%%%%%%%%%%%%%%%%%%%%%%%%%%%%%%%%%%%%%%%%%%%%%%%%%%%%%%%%%%%%

		%%%%%%%%%%%%%%%%%
		\node[transition,label=north:$\color{myDGray}e_1$] at (-3.25,-5.25) (t7){};
		\node[transition,label=west:$\color{myDGray}i$] at (-2,-5.25) (t8){};
		\node[transition,label=north:$\color{myDGray}e_2$] at (-0.75,-5.25) (t9){};

		\node[envplace,label=south:$\color{myDGray}B$] at (-3.25,-6) (p13){};
		\node[sysplace,label=south:$\color{myDGray}D$] at (-2,-6) (p14){};
		\node[envplace,label=south:$\color{myDGray}B$] at (-0.75,-6) (p15){};
		
		\node[transition,label=north:$\color{myDGray}e_1$] at (0.75,-5.25) (t12){};
		\node[transition,label=west:$\color{myDGray}i$] at (2,-5.25) (t13){};
		\node[transition,label=north:$\color{myDGray}e_2$] at (3.25,-5.25) (t14){};

		\node[envplace,label=west:$\color{myDGray}B$] at (0.75,-6) (p20){};
		\node[sysplace,label=south:$\color{myDGray}D$] at (2,-6) (p21){};
		\node[envplace,label=east:$\color{myDGray}B$] at (3.25,-6) (p22){};
		
		\node[transition,label=west:$\color{myDGray}b$] at (0.75,-7) (t15){};
		\node[transition,label=east:$\color{myDGray}b$] at (3.25,-7) (t16){};

		\node[sysplace,label=south:$\color{myDGray}D$] at (0,-7.75) (p23){};
		\node[envplace,label=south:$\color{myDGray}B$] at (1.5,-7.75) (p24){};
		
		\node[sysplace,label=south:$\color{myDGray}D$] at (2.5,-7.75) (p25){};
		\node[envplace,label=south:$\color{myDGray}B$] at (4,-7.75) (p26){};
		%%%%%%%%%%%%%%%%%%%%%%%%%%%%%%%%%%%%%%%%%%%%%%%
		
		\node[transition,label=south:$\color{myDGray}b$] at (-5,-4.5) (tt1){};
		
		\node[envplace,label=south:$\color{myDGray}B$] at (-5.75,-5.25) (pp1){};
		\node[sysplace,label=south:$\color{myDGray}D$] at (-4.25,-5.25) (pp2){};

		%%%%%%%%%%%%%%%%%%%%%%%%%%%%%%%%%%%%%%%%%%%%%%%%%%%

		\node[token] at (0,0.75) (){};
		\node[token] at (0,0) (){};
		
		\draw[arrow] (p1) to (t1);
		\draw[arrow] (p1) to (t2);
		\draw[arrow] (t1) to (p2);
		\draw[arrow] (t2) to (p3);
		
		\draw[arrow] (ps0) to (ts0);
		\draw[arrow] (ts0) to (p4);
		
		\draw[arrow] (p2) to (t3);
		\draw[arrow] (p2) to (t4);
		\draw[arrow] (p4) to (t3);
		\draw[arrow] (p4) to (t4);
		
		\draw[arrow] (p3) to (t5);
		\draw[arrow] (p3) to (t6);
		\draw[arrow] (p4) to (t5);
		\draw[arrow] (p4) to (t6);
		
		%%%%%%%
		
		\draw[arrow] (t3) to (p5);
		\draw[arrow] (t3) to (p6);
		
		\draw[arrow] (t4) to (p7);
		\draw[arrow] (t4) to[out=180, in=180] (p8);
		
		\draw[arrow] (t5) to (p9);
		\draw[arrow] (t5) to[out=0, in=0] (p10);
		
		\draw[arrow] (t6) to (p11);
		\draw[arrow] (t6) to (p12);
		
		%%%
		
		\draw[arrow] (p7) to (t7);
		\draw[arrow] (p7) to (t9);
		\draw[arrow] (t7) to (p13);
		\draw[arrow] (t9) to (p15);
		
		\draw[arrow] (p8) to (t8);
		\draw[arrow] (t8) to (p14);
		
		%%%%%%%%%%%%%%%%%%%%%%%%%%%%%%%%%%%%%%%%%%%%
		
		\draw[arrow] (p9) to (t12);
		\draw[arrow] (p9) to (t14);
		\draw[arrow] (t12) to (p20);
		\draw[arrow] (t14) to (p22);
		
		\draw[arrow] (p10) to (t13);
		\draw[arrow] (t13) to (p21);
		
		\draw[arrow] (p20) to (t15);
		\draw[arrow] (p22) to (t16);
		
		\draw[arrow] (p21) to (t15);
		\draw[arrow] (p21) to (t16);
		
		\draw[arrow] (t15) to (p23);
		\draw[arrow] (t15) to (p24);
		
		\draw[arrow] (t16) to (p25);
		\draw[arrow] (t16) to (p26);
		%%%%%%%%%%%%%%%%%%%%%%%%%%%%%%%%%%%%%%%%%%%%%%%%%%
		\draw[arrow] (p5) to (tt1);
		\draw[arrow] (p6) to (tt1);
		
		\draw[arrow] (tt1) to (pp1);
		\draw[arrow] (tt1) to (pp2);
		
		\draw[red,dotted,very thick,rounded corners=10pt] (-6.25,-3.35) rectangle (-3.75,-4.25);
		\node[red] at (-5, -3.75) () {\scriptsize $M_1$};
		
		\draw[green,dotted,very thick,rounded corners=10pt] (1.5,-5.6) rectangle (4.4,-6.5);
		\node[green] at (4.15, -6) () {\scriptsize $M_2$};

		\draw[blue,dotted,very thick,rounded corners=10pt] (0.9,-3.35) rectangle (2.5,-4.15);
		
		\node[blue] at (1.2,-3.75) () {\scriptsize $M_3$};
		
		\draw[blue,dotted,very thick,rounded corners=10pt] (1.4,-5.5) rectangle (2.6,-7);

		\node[blue] at (2,-6.75) () {\scriptsize $M_3$};
		
		\end{tikzpicture}
		
		\caption{A winning strategy $\dot{\sigma}$ for the Petri game $\dot{\pGame}$ in \refFig{firstRed}. The marking of winning places has been omitted. After first allowing both $a$ and $b$ the strategy makes a case distinction on which commination of $e_1$ or $e_2$ and $a$ or $b$ occurred. Reachable marking $M_1$, $M_2$ and $M_3$ are surrounded in red, green and blue. \vspace{-0.4cm}}
		\label{fig:example_Strat}
	\end{figure}
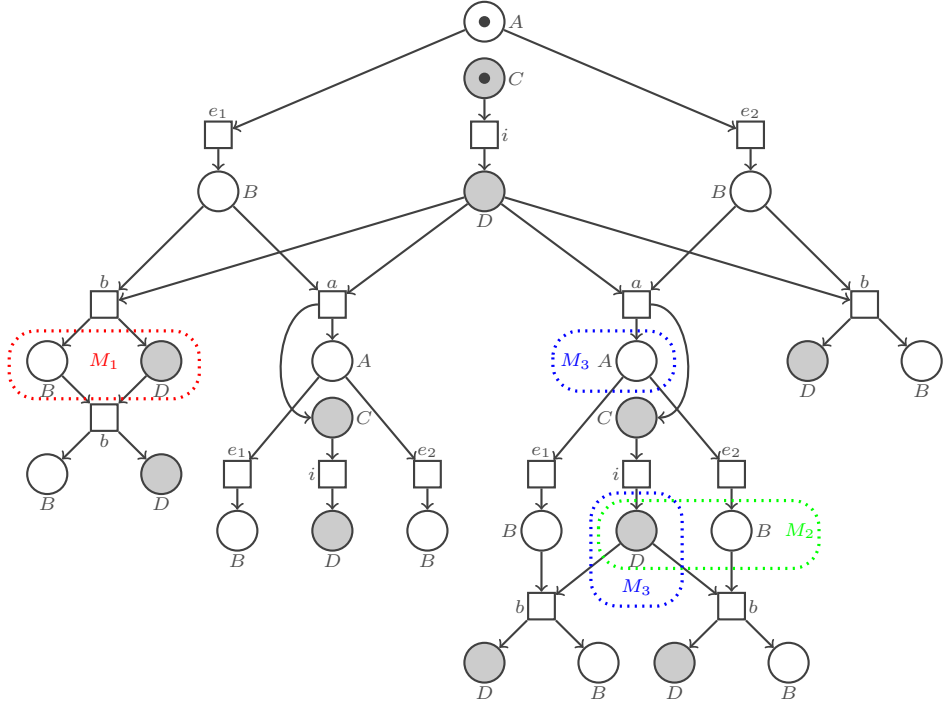
	
	\begin{figure}[t]
		\centering
		\small
		\renewcommand{\arraystretch}{1.2}
		\begin{tabular}{c|l|c|} 
			\cline{2-3}
			& \multicolumn{1}{ |c| }{$u \in \Plays(\dot{\cGame_\pGame}, \dot{\varrho}_\sigma) \cap \Plays_{p_2}(\dot{\cGame_\pGame})$} & $f^{\dot{\varrho}_\sigma}_{p_2}(u)$  \\ 
			\cline{2-3}\multicolumn{1}{c}{}&\multicolumn{2}{c}{}\\[-\bls]\cline{2-3}
			& $\epsilon$ & $\{ \tau_{(C, \{i\})}\}$  \\ 
			\cline{2-3} 
			& $\tau_{(C, \{i\})}$ & $\emptyset$ \\ 
			\cline{2-3} 
			& $\tau_{(C, \{i\})}, i$& $\{\tau_{(D, \{a, b\})}\}$ \\ 
			\cline{2-3} 
			& $\tau_{(C, \{i\})}, i, \tau_{(D, \{a, b\})}$ & $\emptyset$ \\ 
			\cline{2-3} 
			& $\tau_{(C, \{i\})}, i, \tau_{(D, \{a, b\})}, e_1, a$ & $\{ \tau_{(C, \{i\})}\}$ \\ 
			\cline{2-3} 
			\myItem{(1)} & $\tau_{(C, \{i\})}, i, \tau_{(D, \{a, b\})}, e_1, b$  & $\{ \tau_{(D, \{b\})}\}$ \\ 
			\cline{2-3} 
			& $\tau_{(C, \{i\})}, i, \tau_{(D, \{a, b\})}, e_2, a$ & $\{ \tau_{(C, \{i\})}\}$ \\ 
			\cline{2-3} 
			& $\tau_{(C, \{i\})}, i, \tau_{(D, \{a, b\})}, e_2, b$ & $\emptyset$ \\   
			\cline{2-3} 
			& $\tau_{(C, \{i\})}, i, \tau_{(D, \{a, b\})}, e_2, a, \tau_{(C, \{i\})}, i$ & $\{\tau_{(D, \{b\})}\}$ \\ 
			\cline{2-3} 
			& \multicolumn{2}{|c|}{$\cdots$} \\
			\cline{2-3}
		\end{tabular} 
		
		\caption{Controller $\dot{\varrho}_\sigma$ build from the winning strategy $\dot{\sigma}$ in \refFig{example_Strat}. 
			The controller is depicted by listing possible plays and the decision of $f^{\dot{\varrho_\sigma}}_{p_2}$ on them.  \vspace{-0.4cm}}
		\label{fig:possiblePlays}
	\end{figure}
\end{example}

\paragraph*{Strategy-Equivalence}

Given the constructed controller $\varrho_\sigma$, we can prove it strategy-equivalent to $\sigma$.
For our bisimulation $\relation$, we use the one we already defined, but restrict it to reachable markings in $\pNet^\sigma$ and plays in $\Plays(\cGame_\pGame, \varrho_\sigma)$.
The previous statements (\refLemma{PtoC_sameLabel} and \refLemma{PtoC_sameCausalPast}) extend to this restricted relation.
We begin by showing a direct consequence of \refLemma{PtoC_sameCausalPast}:

\begin{lemma}
	\label{lem:PtoC_localeViewSamePlace}
	If $M \relation u$ and $p \in \processes$ then
	\begin{align*}
	M \cap \inv{\lambda}[\places^{\pts{p}}] &= \fireSeq{\pGame^\mathfrak{U}}{\projtr{u}} \cap \inv{\lambda}[\places^{\pts{p}}] \\
	&= \fireSeq{\pGame^\mathfrak{U}}{\projtr{\view_p(u)}} \cap \inv{\lambda}[\places^{\pts{p}}]
	\end{align*}
	
\end{lemma}
\begin{proof}
	Let $q$ be the unique place with $q \in M \cap \inv{\lambda}[\places^{\pts{p}}]$.
	It holds that $M \cap \inv{\lambda}[\places^{\pts{p}}] = \fireSeq{\pGame^\mathfrak{U}}{\pastt{q}} \cap \inv{\lambda}[\places^{\pts{p}}]$ since firing the transitions in the past of $q$ is always sufficient to reach $q$. Note that writing down $\fireSeq{\pGame^\mathfrak{U}}{\pastt{q}}$ is well-defined. 
	By \refLemma{PtoC_sameCausalPast}, it holds that $\pastt{q} = \projtr{\view_p(u)}$.
	We know conclude that 
	\begin{align*}
	M \cap \inv{\lambda}[\places^{\pts{p}}] &= \fireSeq{\pGame^\mathfrak{U}}{\pastt{q}} \cap \inv{\lambda}[\places^{\pts{p}}]\\
	&= \fireSeq{\pGame^\mathfrak{U}}{\projtr{\view_p(u)}} \cap \inv{\lambda}[\places^{\pts{p}}]
	\qedhere
	\end{align*}
\end{proof}

 Our definition of $\varrho_\sigma$ is completely independent from the definition of $\relation$. 
\refLemma{PtoC_localeViewSamePlace}, however, establishes an important relation between them.
Suppose $u$ is the global play in $\cGame_\pGame$ and $M$ a marking such that $M \relation u$. From the definition of $\relation$, we know that $\fireSeq{\pNet^\sigma}{\projtr{u}} = M$.
Since $\view_p(u)$ differs (in general) from $u$, simulating $\view_p(u)$ instead of~$u$ results in a different marking $M'$. 
\refLemma{PtoC_localeViewSamePlace} now states that for process $p$, the place that belongs to $\pts{p}$ is identical in $M$ and $M'$. 
This establishes a connection to our controller definition as, in $\varrho_\sigma$, each process simulates its local view and copies the decisions on the resulting marking.
By \refLemma{PtoC_localeViewSamePlace} in related situations, every process therefore copies the decision of one of the places in $M$. 

\refLemma{PtoC_localeViewSamePlace} allows us to show that the defined $\varrho_\sigma$ actually enables the same behavior if $M \relation u$. Essentially, it allows us to conclude that $\varrho_\sigma$ copies $\sigma$ in $\relation$-related situations. 
We can reason in both direction:
\begin{itemize}
	\item If $u \, t \in \Plays(\cGame, \varrho_\sigma)$ then all involved processes allowed $t$.
	Every process $p \in \dom(t)$ either resides in an environment place (a state corresponding to an environment place) where it has no control or it is on a system place where it must have chosen a commitment set where $t$ is included. 
	$p$ chose its commitment set by simulating its local view on $u$ in the branching process of $\sigma$. By \refLemma{PtoC_localeViewSamePlace}, it thereby copied the decision of a system place in~$M$ (the system place $q \in M \cap \inv{\lambda}[\places^{\pts{p}}]$). 
	As $t$ is in the commitment set of every process involved in $t$, it must be in the postcondition of every of every system place involved in $t$. Therefore, $t$ is enabled in $M$.
	
	\item If the strategy allows a transition $t$ from $M$, all system places must have agreed, i.e., included $t$ in their postcondition.
	In $\varrho_\sigma$, each process decided on what to allow as a commitment set by simulating its local view and, by \refLemma{PtoC_localeViewSamePlace}, copies the decision of one system place in $M$. 
	As $t$ is included in the postcondition of all involved places, every process involved in $t$ thus chooses a commitment set where $t$ is included. 
	We can hence see that $t$ is an extension of $u$ (after playing sufficiently many $\tau$-actions to choose a commitment set).
\end{itemize}

 $\varrho_\sigma$ is a controller for both $\cGame_\pGame$ and $\widehat{\cGame_\pGame}$. To prove that $\sigma$ and $\varrho_\sigma$ are bisimilar, we can treat $\cGame_\pGame$ and $\widehat{\cGame_\pGame}$ as the same, i.e., ignoring all \Lightning-actions in $\widehat{\cGame_\pGame}$. 
We later show that,  if $\sigma$ is deterministic, \Lightning-actions are never part in any play compatible with $\varrho_\sigma$ and can hence be neglected for bisimulation. 

\begin{lemma}
	If $M \relation u$ and $u' = u \, t \in \Plays(\cGame_\pGame, \varrho_\sigma)$ then there exists a marking $M' \in \reach(\pNet^\sigma)$ with $\fireTranTo{M}{t}{M'}$ and $M' \relation u'$.
	\label{lem:PtoAl1}
\end{lemma}
\begin{proof}
	Since $M \relation u$, \refLemma{PtoC_sameLabel} allows us to conclude that $\zeta(\state{u}) = \lambda[M]$. \\
	We want to show that $t$ is enabled in $M$. This would imply that $\fireTranTo{M}{t}{M'}$ and $M' \relation u$ is a trivial consequence.
	From $\zeta(\state{u}) = \lambda[M]$ and since $t$ is possible from $\state{u}$, the construction of $\cGame_\pGame$ allows us to conclude that $t$ is enabled in $\lambda[M]$. 
	There hence is a set $C \subseteq M$ with $\lambda[C] = \pre{\pGame}{t}$. \\
	We assume for contradiction that $t$ is not allowed by the strategy.
	Because of justified refusal, there is a \textit{system} place $q \in C$ with $t \not\in \lambda[\post{\pNet^\sigma}{q}]$ \myItem{(1)}.
	Place $q$ belongs to some process $p$, i.e., $q \in M \cap \inv{\lambda}[\places^{\pts{p}}]$.
	We know that $\lambda(q) = \zeta(\statep{p}{u})$.
	By construction of $\dom$, we know that $p \in \dom(t)$.
	Since $u' = u\,t \in \Plays(\cGame_\pGame, \varrho_\sigma)$ and $q$ is a system place we know that $\statep{p}{u} = (\lambda(q), B)$ for some $B$ with $t \in B$, i.e., process $p$ has chosen a commitment set that includes $t$ \myItem{(2)}.
	We derive the contradiction by showing that the set of transitions leaving $q$ ($\lambda[\post{\pNet^\sigma}{q}]$) agrees with the decision of $p$ and must hence, by \myItem{(2)}, include $t$.\\
	As $\statep{p}{u} = (\lambda(q), B)$, there must be a $\tau_{(\lambda(q), B)}$-action in $u$, since this is the only action leading to state $(\lambda(q), B)$.
	Let $u_{\tau} \sqsubseteq u$ be the prefix obtained by removing the last such action.
	$u_\tau$ is a $\varrho_\sigma$-compatible play.
	It holds that $\statep{p}{u_\tau} = \lambda(q)$.
	We conclude that $\tau_{(\lambda(q), B)} \in f^{\varrho_\sigma}_p(\view_p(u_{\tau}))$.\\
	We can now study how $\varrho_\sigma$ chooses $B$ as its commitment set. 	
	By definition of $\varrho_\sigma$, we know that $B = \lambda[\post{\pNet^\sigma}{q'}]$ for the unique system place $q'$ with 
	$$q' \in \fireSeq{\pNet^\sigma}{\projtr{\view_p(u_{\tau})}} \cap \inv{\lambda}[\places^{\pts{p}}]$$
	Now by \refLemma{PtoC_localeViewSamePlace}: 
	\begin{align*}
	\{ q' \} &= \fireSeq{\pNet^\sigma}{\projtr{\view_p(u_{\tau})}} \cap \inv{\lambda}[\places^{\pts{p}}]\\
	&= \fireSeq{\pNet^\sigma}{\projtr{\view_p(u)}} \cap \inv{\lambda}[\places^{\pts{p}}]\\
	&= \fireSeq{\pNet^\sigma}{\projtr{u}} \cap \inv{\lambda}[\places^{\pts{p}}]\\
	&= M \cap \inv{\lambda}[\places^{\pts{p}}] \\
	&= \{q\}
	\end{align*}
	The system place reached by simulating $\projtr{\view_p(u_{\tau})}$ is hence exactly the system place in $M$. It follows that
	$$B = \lambda[\post{\pNet^\sigma}{q'}] = \lambda[\post{\pNet^\sigma}{q}]$$ 
	This is a contradiction to $t \in B$ \myItem{(2)} but we assumed $t \not\in \lambda[\post{\pNet^\sigma}{q}]$ \myItem{(1)}.
\end{proof}

\begin{lemma}
	If $M \relation u$ and $u' = u \, \tau \in \Plays(\cGame_\pGame, \varrho_\sigma)$ then $M \relation u'$.
	\label{lem:PtoAl2}
\end{lemma}
\begin{proof}
	Obvious consequence from the definition of $\relation$.
\end{proof}

In the proofs above, we never have to deal with case \myItem{b)} in the definition of $\varrho_\sigma$.
We always conclude statements under the assumption that $M \relation u$ for some $M$ and $u$. By definition of~$\relation$, $\projtr{u}$ is a valid sequence in $\pNet^\sigma$ and therefore $\projtr{\view_p(u)}$ as well (since $\projtr{\view_p(u)}$ is a prefix of $\projtr{u}$). 
We can show the next corollary which shows that case \myItem{3b)} can be ignored for \emph{any} $\varrho_\sigma$-compatible play\footnote{This is no statement required by strategy-equivalence.}.
\begin{corollary}
	If $u$ is a $\varrho_\sigma$-compatible play then $\projtr{\view_p(u)}$ is a valid sequence in $\pNet^\sigma$.
\end{corollary} 
\begin{proof}
	It holds that $\init^\sigma \relation \epsilon$. 
	By playing $u$ and using \refLemma{PtoAl1} and \refLemma{PtoAl2}, we get a reachable marking $M$ in $\pNet^\sigma$ with $M \relation u$. 
	By definition of $\relation$, $\projtr{u}$ is a valid sequence in~$\pNet^\sigma$. Since $\projtr{\view_p(u)}$ is a prefix of $\projtr{u}$ it is a valid sequence as well.  
\end{proof}

\begin{lemma}
	If $M \relation u$ and $\fireTranTo{M}{t}{M'}$ for some $M' \in \reach(\pNet^\sigma)$ there exists $u' = u \, \tau^* \, t \in \Plays(\cGame_\pGame, \varrho_\sigma)$ with $M' \approx u'$.
	\label{lem:PtoAl3}
\end{lemma}
\begin{proof}
	Since $M \relation u$, \refLemma{PtoC_sameLabel} allows us to conclude that $\zeta(\state{u}) = \lambda[M]$ \myItem{(1)}. 
	Transition $t$ is enabled in $M$ and hence for every place $q$ in $M$ with $\lambda(q) \in \pre{\pGame}{t}$, it holds that $t \in \lambda[\post{\pNet^\sigma}{q}]$ \myItem{(2)}.\\
	Let $u_\tau$ be $u$ extended with as many $\tau$-actions as possible such that no $\tau$-action is possible after~$u_\tau$. Since $\varrho_\sigma$ always allows a commitment set, after playing $u_\tau$, every process that can choose a commitment set, has chosen a commitment set, i.e., for every process $p$ with $\zeta(\statep{p}{u_\tau}) \in \places_\psys$ we know that $\statep{p}{u_\tau} = (\_, \_)$.\\
	Assume for contradiction $u_\tau \, t \not\in \Plays(\cGame_\pGame, \varrho_\sigma)$. Since \myItem{(1)} holds, every process on a system state has chosen a commitment set and action $t$ is uncontrollable, we can conclude that there exists a process $p \in \dom(t)$ with $\statep{p}{u_\tau} = (q_{p}, B)$ but where $t \not \in B$.
	That is, $p$ has chosen a commitment set where $t$ is not included.
	Let $q'_{p} \in M \cap \inv{\lambda}[\places^{\pts{p}}]$ be the corresponding place in $M$.
	Since $p \in \dom(t)$ and $t$ is enabled in $\lambda[M]$, we conclude that $\lambda(q'_{p}) \in \pre{\pGame}{t}$ and by \myItem{(2)} we get that $t \in \lambda[\post{\pNet^\sigma}{q'_{p}}]$ \myItem{(3)}, i.e., from the place in $M$ that corresponds to~$p$, $t$ is enabled (in the postcondition). \\
	Since $p$ is in state $(q_{p}, B)$ there is an $\tau_{(q_{p}, B)}$-action in $u_\tau$.
	Let $u_\tau^-$ be $u_\tau$ where the last such action is removed such that $p$ has not chosen a commitment set (i.e., $\statep{{p}}{u_\tau} = q_{p}$).
	We can conclude that $\tau_{(q_{p}, B)} \in f_{p}(\view_{p} (u_\tau^-))$.\\
	By the definition of $\varrho_\sigma$, it holds that $B = \lambda[\post{\pNet^\sigma}{q''_{p}}]$ for the unique place $q''_{p}$ with
	$$q''_{p} \in \fireSeq{\pNet^\sigma}{\projtr{\view_{p}(u_\tau^-)}} \cap \inv{\lambda}[\places^{\pts{p}}]$$
	Because of \refLemma{PtoC_localeViewSamePlace}:
	\begin{align*}
	\{ q''_{p} \} &= \fireSeq{\pNet^\sigma}{\projtr{\view_{p}(u_\tau^-)}} \cap \inv{\lambda}[\places^{\pts{p}}]\\
	&= \fireSeq{\pNet^\sigma}{\projtr{\view_{p}(u_\tau)}} \cap \inv{\lambda}[\places^{\pts{p}}]\\
	&= \fireSeq{\pNet^\sigma}{\projtr{u}} \cap \inv{\lambda}[\places^{\pts{p}}]\\
	&= M \cap \inv{\lambda}[\places^{\pts{p}}] \\
	&= \{q'_{p} \}
	\end{align*}
	We hence conclude that $q''_{p} = q'_{p}$ and get 
	$$B = \lambda[\post{\pNet^\sigma}{q''_{p}}] = \lambda[\post{\pNet^\sigma}{q'_{p}}]$$
	The chosen commitment set $B$ agrees with the transitions leaving $q'_{p}$. This is a contradiction to $t \in \lambda[\post{\pNet^\sigma}{q'_{p}}]$ \myItem{(3)} and our assumption $t \not\in B$. 
\end{proof}

\begin{corollary}
	$\sigma$ and $\varrho_\sigma$ are bisimilar.
\end{corollary}
\begin{proof}
	By definition, $\init^\sigma \relation \epsilon$ holds.
	Since there are no local (unobservable) transitions in $\pGame$ the statement follows from \refLemma{PtoAl1}, \refLemma{PtoAl2}, and \refLemma{PtoAl3}.
\end{proof}
 Having proven bisimilarity, we can show that winningness is preserved by our translation. 

\begin{lemma}
	If $\sigma$ is a winning strategy for $\pGame$ then $\varrho_\sigma$ is a winning controller for $\cGame_\pGame$.
\end{lemma}
\begin{proof}
	We first show that all plays in $\Plays(\cGame_\pGame, \varrho_\sigma)$ are finite: 
	Assume for contradiction there is an infinite play $u$. Due to $\cGame_\pGame$ not permitting infinite sequences of consecutive $\tau$-actions, $u$ must contain infinitely many observable actions. 
	By bisimulation, we have an infinite sequence of markings $\pNet^\sigma$. This is a contradiction since $\sigma$ is by assumption winning and therefore by definition finite. \\
	We now show that all maximal plays terminate in a winning configuration:
	Suppose $u \in \Plays(\cGame_\pGame, \varrho_\sigma)^M$ is a maximal $\varrho_\sigma$-compatible play, i.e., cannot be extended by any action. Using our bisimulation, there exists a reachable marking $M$ in $\pNet^\sigma$ with $M \relation u$. Since $u$ is maximal, $M$ is final. Since $\sigma$ is winning, $M$ must be a winning marking. Now, $\zeta(\state{u}) = \lambda[M]$ (by \refLemma{PtoC_sameLabel}) and from our construction of the winning states in $\cGame_\pGame$ it follows that $\state{u}$ is winning as well.
\end{proof}

\paragraph*{Deterministic Strategies}
So far, we ignored all \Lightning-actions introduced with $\widehat{\cGame_\pGame}$. 
We can justify this by showing that the \Lightning-actions can actually never be taken, if $\varrho_\sigma$ is constructed from a deterministic $\sigma$.
The \Lightning-transition can occur when the processes have chosen their commitment sets such that two transitions are enabled from the same set. 
By construction, $\varrho_\sigma$ chooses its commitment sets in accordance with the strategy $\sigma$, i.e., the actions in a commitment set are exactly the ones that are enabled by a place in $\sigma$. 
If $\sigma$ is deterministic there is at most one transition enabled from every system place and thereby at most one action possible from each commitment set; the \Lightning-actions are thus never enabled.
Formally:

\begin{lemma}
	\label{lem:pToAdet1}
	If $\sigma$ is deterministic, then there is no play in $\Plays(\widehat{\cGame_\pGame}, \varrho_\sigma)$ that contains a \Lightning-action.
\end{lemma}
\begin{proof}
	Suppose the opposite, i.e., there is a $\varrho_\sigma$-compatible play $u$ that contains a \Lightning-action.
	W.l.o.g.\ $u = u'\, \text{\Lightning}^{(q,A)}_{[t_1,t_2]}$ with $q \in \places_\psys$, $A \subseteq \post{\pGame}{q}$, $t_1, t_2 \in A$, and there is no $\text{\Lightning}$-action in $u'$. 
	By construction of the \Lightning-actions, it is easy to see that if $u'\, \text{\Lightning}^{(q,A)}_{[t_1,t_2]}$ is a play then $u' \, t_1$ and $u' \, t_2$ are as well. 
	
	If $\text{\Lightning}^{(q,A)}_{[t_1,t_2]}$ is possible, the transition relation in $\widehat{\cGame_\pGame}$ also requires a process $p$ in state $(q, A)$, i.e., $\statep{p}{u'} = (q, A)$.
	Since there are no \Lightning-action in $u'$ we can use the previous bisimulation result and obtain a marking $M \in \reach(\pNet^\sigma)$ with $M \relation u'$. By bisimulation, we know that $t_1$ and $t_2$ (transitions with that label) are enabled from $M$. 
	
	Let $q' \in M \cap \inv{\lambda}[\places^{\pts{p}}]$ be the \emph{system} place in $M$ that corresponds to process $p$. From \refLemma{PtoC_sameLabel}, we know that $\lambda(q') = q$. 
	Since $t_1, t_2 \in \dom(p)$, we get $t_1, t_2 \in \transitions^{\pts{p}}$. Place $q'$ is therefore involved in both $t_1$ and $t_2$. So, we see that $t_1, t_2 \in \lambda[\post{\pNet^\sigma}{q'}]$ and both can occur from $M$.   
	This is a contradiction to the assumption that $\sigma$ is deterministic. 
\end{proof}

 If $\sigma$ is deterministic, \refLemma{pToAdet1} shows that $\varrho_\sigma$ does not allow any \Lightning-actions.
We can hence neglect all \Lightning-actions and extend our proofs for bisimulation and winningness from $\cGame_\pGame$ to $\widehat{\cGame_\pGame}$.
We get that $\varrho_\sigma$ is a winning controller for $\widehat{\cGame_\pGame}$ (and also $\cGame_\pGame$) and, furthermore, bisimilar to $\sigma$.
This gives us the first half of our correctness statement:

\begin{proposition}
	If $\sigma$ is a winning strategy for $\pGame$ then $\varrho_\sigma$ is a winning controller for $\cGame_\pGame$ and bisimilar to $\sigma$.\\
	If $\sigma$ is a winning, deterministic strategy for $\pGame$ then $\varrho_\sigma$ is a winning controller for $\widehat{\cGame_\pGame}$ (and for $\cGame_\pGame$) and bisimilar to $\sigma$.
	\label{prop:PtoC_1}
\end{proposition}

\subsection{Translating Controllers to Strategies}
\label{sec:PGtoAA2}

In this section, we provide the formal translation of controllers to strategies. 
We first need to restrict the possible controllers for $\cGame_\pGame$: We only consider controllers that allow at most one commitment set (one $\tau$-action from each state).
This restriction is needed to allow for bisimilar strategies\footnote{If two commitment sets are chosen, two states that are indistinguishable by weak-bisimulation allow different behavior. A strategy must hence allow the behavior of both states from a single place. This is in general not possible.  }.  
Even though this constraint is not desirable, we can argue that it does not impose any relevant restriction on possible controllers:
Suppose controller $\varrho = \{f^\varrho_p\}_{p \in \processes}$ allows more than one commitment set. We can build a modified controller  $\varrho' = \{f^{\varrho'}_p\}_{p \in \processes}$ by
$$ f^{\varrho'}_p(u) = \{ \tau_{(q, \bigcup\limits_{i = 1, \cdots, n} A_i)} \} \text{ when } f^\varrho_p(u) = \{ \tau_{(q, A_1)}, \cdots, \tau_{(q, A_n)}  \}$$
Whenever $\varrho$ allows multiple commitment sets, $\varrho'$ chooses the union of all of them as the new (unique) commitment set.
$\varrho'$ admits the same observable sequences as $\varrho$. 
In particular, $\varrho'$  is winning if and only if $\varrho$ is winning. 
Allowing more commitment sets does not give an advantage to a controller\footnote{Instead of building the union-commitment set, it would be valid to simply choose one of the allowed commitment sets. An approach similar to this has been realized in \cite{muscholl2009look}. }. 
For convenience, we restrict controllers even further by enforcing \emph{exactly} one commitment set.
If a controller $\varrho$ chooses no commitment set we can instead choose the empty one\footnote{Unlike the restriction to at most one chosen commitment set, the further restriction to exactly one commitment set is not needed to maintain bisimilarity but purely for convenience.}. We call this restriction on controllers $\star$.

Assume now we are given a winning controller $\varrho$ for $\cGame_\pGame$ (or $\widehat{\cGame_\pGame}$) that satisfies $\star$. 
We need to construct a winning, bisimilar strategy $\sigma_\varrho$ for $\pGame$. 
Unlike controllers that are defined as functions evoked on an entire play, strategies for Petri games are defined as branching processes. 
We thus incrementally build a branching process for $\sigma_\varrho$.
In our incremental strategy construction, every system place needs to decide what transitions to allow from that place. This decision should be based on the causal past of that place and should be made in accordance with $\varrho$ in order to, in the end, obtain a bisimilar strategy. 
%The decision at place $q$ is based on his causal past.
We would therefore like to be able to translate the causal past to a play in $\cGame_\pGame$, give this play to controller $\varrho$, and enable exactly the transitions that the controller chose as a commitment set. 
The crucial step is the translation of the causal past of place $q$ to a play in $\cGame_\pGame$ that is compatible with $\varrho$.
When translating strategies to controllers in Appendix \ref{sec:PtoC_translateStrtoCont}, we had to translate a local view into the causal past of a place.
We could easily do so by ignoring all $\tau$-actions using~$\projtr{\cdot}$.
By contrast, in our present translation, we have to \emph{add} $\tau$-actions to obtain a play in $\cGame_\pGame$. 
For a place $q$ with causal past $\pastt{q}$, we thus want to compute a $\varrho$-compatible play $u$ that contains the same observable actions, i.e., where $\projtr{u} = \pastt{q}$.

\paragraph*{Play Reconstruction}

\begin{figure}[t!]
	\algblockdefx[io]{Input}{Output}
	[1]{\textbf{input:} #1}
	[1]{\textbf{return:} #1}
	\algblockdefx[Pro]{Proc}{ProcEnd}
	[1]{\textbf{function} #1}

	\begin{subfigure}[c]{1.0\textwidth}
		\begin{algorithmic}
			\Proc{$\mathit{extend}_\varrho$}
			\Input{$u \in \Plays(\cGame_\pGame)$}
			\ForAll{$p \in \processes$ \textbf{with} $\statep{p}{u} \in \places_\psys$}
			\State Compute $f^\varrho_p(\view_p(u)) = \{\tau_{(\statep{p}{u}, A)}\}$
			\State $u \gets u \, \tau_{(\statep{p}{u}, A)}$
			\EndFor
			\Output{$u$}
			\ProcEnd
		\end{algorithmic}
		\subcaption{}
	\end{subfigure}\\[0.5cm]
	\begin{subfigure}[c]{0.5\textwidth}
		\begin{algorithmic}
			\Proc{$\mathit{rec'}_\varrho$}
			\Input{$\kappa \gets \kappa_0, \cdots, \kappa_{n-1} \in \transitions^*$}
			\State{$u \gets \mathit{extend}_\varrho\,(\epsilon)$}
			\For{$i \gets 0 \textbf{ to } n-1$}
			\State $u \gets u \, \kappa_i$
			\State {\textbf{assert} $u \in \Plays(\cGame_\pGame)$ \quad \myItem{(A)}}
			\State{$u \gets \mathit{extend}_\varrho\,(u)$}
			\EndFor
			\Output{$u$}
			\ProcEnd
		\end{algorithmic}
		\subcaption{}
	\end{subfigure}
	\begin{subfigure}[c]{0.5\textwidth}
		\begin{algorithmic}
			\Proc{$\mathit{rec}_\varrho$}
			\Input{$\pastt{q}$}
			\State{Order $\pastt{q}$ totally into sequence}
			\State{$\kappa \gets \kappa_0, \cdots, \kappa_{n-1} \in \transitions^*$}
			\Output{$\mathit{rec}'_\varrho(\kappa)$}
			\ProcEnd
		\end{algorithmic}
		\subcaption{}
	\end{subfigure}
	
	\caption{Description of algorithm $\recv$ used to reconstruct a play in $\cGame_\pGame$ from the transitions in the causal past of a place. \vspace{-0.4cm}}
	\label{fig:recalg}
\end{figure}

Given the causal past of $q$, we need to add $\tau$-actions to the play. 
We pursue an incremental construction of that play: We begin with an empty play and add the transitions in the past of $q$ one at a time. In between, we need to play $\tau$-actions to allow all processes on system places to choose a commitment set. 
The incremental construction is done by a function $\recv$ that is depicted in \refFig{recalg}.
We add all transitions in the causal past of $q$ in some order compatible with $\leq$ (i.e., some linearization of the causal past) and in between allow every process to choose a commitment set.
Note that, as $\varrho$ satisfies $\star$, every process allows for exactly one commitment set. 
In the algorithm, we include an assertion \myItem{(A)} that requires the trace constructed so far to be a play in $\cGame_\pGame$.
We discuss this assertion later. For now, assume that it is always fulfilled.
It is easy to see that if $\recv$ does not trigger the assertion, the outputted play $u$ satisfies $\projtr{u} = \pastt{q}$.

The first step in $\recv$ consists of finding a total order of $\pastt{q}$. 
We can prove that the resulting trace does not depend on the concrete choice. 
So, $\recv$ is a deterministic procedure.

\begin{lemma}
	Let $q \in \places^{\pGame^\mathfrak{U}}$ be any place in the unfolding of $\pGame$. 
	If $\kappa_1$ and $\kappa_2$ are two totally ordered sequences of $\pastt{q}$ then $\mathit{rec}'_\varrho(\kappa_1) = \mathit{rec}'_\varrho(\kappa_2)$.
	\label{lem:inv_poset}
\end{lemma}
\begin{proof}
	We first show the statement for two totally ordered sequences $\kappa$ and $\kappa'$ of $\pastt{q}$ that only differ at exactly one location, i.e., two consecutive transitions $t_1, t_2$ have been swapped. 
	So $\kappa = \kappa^{r_1}, t_1, t_2, \kappa^{r_2}$ and $\kappa' = \kappa^{r_1}, t_2, t_1, \kappa^{r_2}$ for some sequences $\kappa^{r_1}$ and $\kappa^{r_2}$. 
	
	As $t_1$ and $t_2$ are unordered in $\pastt{q}$, we can conclude $(\pre{\pGame}{t_1} \cup \post{\pGame}{t_1}) \cap (\pre{\pGame}{t_2} \cup \post{\pGame}{t_2}) = \emptyset$. 
	As $t_1$ and $t_2$ can be fired exactly after each other, we can conclude that $t_1$ and $t_2$ involve different slices.
	By our construction of $\cGame_\pGame$, we know that $t_1 \, \ind \, t_2$. \\
	We can see that $\mathit{rec}'_\varrho(\kappa)$ and $\mathit{rec}'_\varrho(\kappa')$ have the following form for some plays $u'$ and $u''$: 
	$$\mathit{rec}'_\varrho(\kappa) = u' \, t_1 \, \underbrace{\tau \cdots \tau}_{\myItem{(1)}} \, t_2 \, \underbrace{\tau \cdots \tau}_{\myItem{(2)}} \, u'' \text{ and } \mathit{rec}'_\varrho(\kappa') = u' \, t_2 \, \underbrace{\tau \cdots \tau}_{\myItem{(3)}} \, t_1 \, \underbrace{\tau \cdots \tau}_{\myItem{(4)}} \, u''$$
	The $\tau$-actions played in \myItem{(1)} and \myItem{(4)} only involve processes from $\dom(t_1)$ and the ones in \myItem{(2)} and \myItem{(3)} from $\dom(t_2)$. Since $t_1 \, \ind \, t_2$ and all the $\tau$-actions are local to one process, both $\mathit{rec}'_\varrho(\kappa)$ and $\mathit{rec}'_\varrho(\kappa')$ describe identical traces. \\
	We have shown the claim for two totally ordered sequences that differ at exactly one location. The proof for general $\kappa_1$ and $\kappa_2$ follows by induction on the minimal number of swaps used to unify  $\kappa_1$ and $\kappa_2$ using the insight from above.
\end{proof}

Regarding the assertion \myItem{(A)}, it can happen that adding a transition from $\kappa$ results in a play that is not in $\Plays(\cGame_\pGame)$. 
The controller could have chosen its commitment sets such that the action that is added from $\kappa$ cannot be taken.
We can, however, show that \emph{if} $\recv$ did not violate the assertion, the obtained play is a valid play in $\cGame_\pGame$ and moreover compatible with $\varrho$.
We can, furthermore, observe that if there is some play in $\Plays(\cGame_\pGame, \varrho)$ that contains exactly the observable actions from the past of a place, then $\recv$ is guaranteed to find such a play without violating the assertion.

\begin{lemma}
	\label{lem:collection}
	Let $\pastt{q}$ be the causal past of some place $q \in \places^{\pGame^{\mathfrak{U}}}$ in the unfolding of $\pGame$.
	\begin{enumerate}
		\item[\myItem{1.}] Assume that $u = \recv(\pastt{q})$ and no assertion is violated. 
		Then $u \in \Plays(\cGame_\pGame, \varrho)$.
		%\item Assume that $u = \recv(E)$ and no assertion fired.
		%Let $p$ be the process $q$ belongs to, i.e., $q \in \inv{\lambda}[\places^{\pts{p}}]$
		%Then $\zeta(\statep{p}{u}) = \lambda(q)$
		\item[\myItem{2.}] If there is a play $u \in \Plays(\cGame_\pGame, \varrho)$ with $\projtr{u} = \pastt{q}$ then $\recv(\pastt{q})$ does not violate an assertion.
	\end{enumerate}
\end{lemma}
\begin{proof}
	The first statement follows from the definition of $\recv$ and the fact that all observable actions are uncontrollable. 
	%The second statement follows from \refLemma{?}.
	The second claim follows since by $\star$ every process has chosen at most one commitment set and simulation is therefore unique. 
\end{proof}

\paragraph*{Construction of $\sigma_\varrho$}
Using $\recv$, we can finally define the construction of $\sigma_\varrho$. It is depicted in \refFig{PtoC_translation2}. 
We incrementally build up a branching process by iterating over every reachable marking $M$ in the partially constructed strategy. 
Every place $q$ in a marking $M$ needs to decide which transitions to enable.
This decision is stored in a set $\mathds{A}_q$.
Since an environment place cannot be restricted by a strategy all outgoing transitions are allowed ($\mathds{A}_q = \post{\pGame}{\lambda(q)}$).
For each system place $q$, we consider its causal past and convert it to a play in $\cGame_\pGame$ using $\recv$.
In this play, the process that corresponds to $q$ has chosen a commitment set., i.e., is in a state of the form $(\_, \_)$. 
We define $\mathds{A}_q$ to be the set of transitions that are in the current commitment set of that process. Hence, $q$ copies the decision made by the corresponding process on the reconstructed play.
Once we have computed $\mathds{A}$ for every place in the marking we add all transitions where all places agree on, i.e., compute $\Delta_M$. Since $\pGame$ is sliceable and therefore safe we can uniquely tell which places need to agree on a transition:
$\Delta_M$ is the set of all transitions that are enabled in $\lambda[M]$ and where all places $q$ in the precondition of $t$ ($\lambda(q) \in \pre{\pGame}{t}$) have agreed on $t$ ($t \in \mathds{A}_q$).
For now, we impose an assertion \myItem{(B)} in the construction. We later see that the assertion can be neglected, i.e., the causal past of any place in the partially constructed strategy can always be converted to a play using $\recv$. 

\begin{figure}[t!]
	\begin{tcolorbox}[colback=white, colframe=myYellow, arc=3mm, boxrule=1mm]
		Start by creating an initial marking $\init^{\sigma_\varrho}$ and extend $\lambda$ s.t.\ $\lambda(\init^{\sigma_\varrho}) = \init^\pGame$.
		
		Iterate over every unprocessed reachable marking $M$ in $\pNet^{\sigma_\varrho}$:
		Consider every $q \in M$:
		
		\begin{itemize}
			\item If $q$ is a system place, i.e., $q \in \inv{\lambda}[\places_\psys]$:\\
			$q$ belongs to a process $p_q \in \processes$, i.e., $q \in M \cap \inv{\lambda}[\places^{\pts{p_q}}]$.
			Compute $u = \recv(\pastt{q})$. Assume that no assertion is violated and that $\zeta(\statep{p_q}{u}) = \lambda(q)$ \myItem{(B)}.
			Because of \myItem{(B)} and the fact that in $u$ every process has chosen a commitment set, it holds that $\statep{p_q}{\view_{p_q}(u)} = (\lambda(q), B)$.
			Define $\mathds{A}_q = B \subseteq \transitions^\pGame$.
			
			\item If $q$ is an environment place, i.e., $q \in \inv{\lambda}[\places_\penv]$:\\
			Define $\mathds{A}_q = \post{\pGame}{\lambda(q)}$
			
		\end{itemize}
		
		Define $$\Delta_M = \{ t \in \transitions^\pGame \mid  \pre{\pGame}{t} \subseteq \lambda[M] \land \forall q \in M: \, \lambda(q) \in \pre{\pGame}{t} \Rightarrow t \in \mathds{A}_q  \}$$
		These are all transitions that can occur and on which all places have agreed.
		We want to add exactly the transitions from $\Delta_M$ from $M$:
		For every $t \in \Delta_M$:
		Check if there already exists a transition $t'$ with $\pre{\pNet^{\sigma_\varrho}}{t'} \subseteq M$ and $\lambda(t') = t$:
		\begin{itemize}
			\item If it already exists, do not add anything.
			\item If it does not exist: Create a new transition $t'$ and extend the flow s.t.\ $\pre{\pNet^{\sigma_\varrho}}{t'} = \{q \in M \mid \lambda(q) \in \pre{\pGame}{t}\}$ and extend $\lambda$ with $\lambda(t') = t$. 
			Add a new place $q'$ for every $q \in \post{\pGame}{t}$ with $\lambda(q') = q$ and extend the flow s.t.\ $\pre{\pNet^{\sigma_\varrho}}{q'} = \{t'\}$.
		\end{itemize}
		Mark $M$ as processed and continue with another, unprocessed marking.
	\end{tcolorbox}
	
	\caption{Construction of strategy $\sigma_\varrho$ for $\pGame$ from a given controller $\varrho$ for $\cGame_\pGame$}
	\label{fig:PtoC_translation2}
\end{figure}

\begin{example}
	
	As an example, consider a possible winning controller $\dot{\varrho}$ for $\dot{\cGame_\pGame}$ from \refFig{firstRed} where $p_2$ allows the following:
	Whenever in state $C$, it chooses the commitment set including $i$ and hence allows a move to $D$. If in state $D$ for the first time, $p_2$ moves to the commitment set containing $b$, i.e., restricts communication to $b$. 
	After executing $b$, it can deduce whether the environment played $e_1$ or $e_2$. 
	In case of $e_1$, it allows $b$ for one more time and subsequentially terminates. 
	In case of $e_2$, it allows communication on $a$, afterwards moves to state $D$ and terminates. 
	The relevant plays and the decision of $\dot{\varrho}$ are depicted in \refFig{exampleController}.
	
	\begin{figure}[t]
		\centering
		\small
		\def\arraystretch{1.2}
		\begin{tabular}{|l|c|} 
			\hline
			\multicolumn{1}{ |c| }{$u \in \Plays(\dot{\cGame_\pGame}, \dot{\varrho}) \cup \Plays_p(\dot{\cGame_\pGame})$} & $f^{\dot{\varrho}}_{p_2}(u)$  \\ 
			\hline 
			\hline
			$\epsilon$ & $\{ \tau_{(C, \{i\})}\}$  \\ 
			\hline 
			$\tau_{(C, \{i\})}$ & $\emptyset$ \\ 
			\hline 
			$\tau_{(C, \{i\})}, i$& $\{\tau_{(D, \{b\})}\}$ \\ 
			\hline 
			$\tau_{(C, \{i\})}, i, \tau_{(D, \{b\})}$ & $\emptyset$ \\ 
			\hline 
			$\tau_{(C, \{i\})}, i, \tau_{(D, \{b\})}, e_1, b$ & $\{ \tau_{(D, \{b\})}\}$ \\ 
			\hline 
			$\tau_{(C, \{i\})}, i, \tau_{(D, \{b\})}, e_2, b$ & $\{ \tau_{(D, \{a\})}\}$ \\ 
			\hline 
			$\tau_{(C, \{i\})}, i, \tau_{(D, \{b\})}, e_1, b, \tau_{(D, \{b\})}$ & $\emptyset$ \\ 
			\hline 
			$\tau_{(C, \{i\})}, i, \tau_{(D, \{b\})}, e_e, b, \tau_{(D, \{a\})}$ & $\emptyset$ \\ 
			\hline 
			$\tau_{(C, \{i\})}, i, \tau_{(D, \{b\})}, e_1, b, \tau_{(D, \{b\})}, b$ & $\{\tau_{(D, \emptyset)} \}$ \\ 
			\hline 
			$\tau_{(C, \{i\})}, i, \tau_{(D, \{b\})}, e_e, b, \tau_{(D, \{a\})}, a$ & $\{ \tau_{(D, \{i\})}\}$ \\ 
			\hline 
			$\tau_{(C, \{i\})}, i, \tau_{(D, \{b\})}, e_e, b, \tau_{(D, \{a\})}, a, \tau_{(C, \{i\})}$ & $\emptyset$ \\ 
			\hline
			$\tau_{(C, \{i\})}, i, \tau_{(D, \{b\})}, e_e, b, \tau_{(D, \{a\})}, a, \tau_{(D, \{i\})}, i$ & $\{\tau_{(D, \emptyset)}\}$ \\ 
			\hline
		\end{tabular} 
		
		\caption{Example winning controller $\dot{\varrho}$ for $\dot{\cGame_\pGame}$ in \refFig{firstRed}.
			The controller is depicted by listing plays and the decision of  $f^{\dot{\varrho}}_{p_2}$ on them. } 
		\label{fig:exampleController}
	\end{figure}

	If we apply our construction we end up with the strategy $\dot{\sigma}_\varrho$ depicted in \refFig{exam_starte_trans} (e).
	Note that $\dot{\sigma}_\varrho$ allows the same behavior as $\dot{\varrho}$: After moving to $D$, the system player allows communication only on $b$. 
	Depending on whether the environment chose $e_1$ or $e_2$, $\dot{\sigma}_\varrho$ either allows $b$ once more or allows $a$ and afterwards moves to $D$ using $i$.
	Apart from showing the final strategy, \refFig{exam_starte_trans} also depicts possible indeterminate steps in the strategy construction.
	Next to each place, the set $\mathds{A}$ as computed in the construction is given in red. 
	The gray label is the one given by $\lambda$. Places $q_1, \cdots, q_5$ are named explicitly in blue. The causal past of them is surrounded in blue. \\
	We retrace the construction depicted in \refFig{exam_starte_trans}:
	The construction begins with an initial marking~(a). 
	For every system place in that marking, we compute the transitions in the causal past and reconstruct a play using $\rec_{\dot{\varrho}}$. For the system place $q_1$, $\rec_{\dot{\varrho}}$ applied to the empty causal past gives us the play $[\tau_{(C, \{i\})}]_\ind$. 
	After playing $[\tau_{(C, \{i\})}]_\ind$, the process that corresponds to $q_1$ (process $p_2$) is in state $(C, \{i\})$. So, we derive that $\mathds{A}_{q_1} = \{i\}$. 
	For the environment place, we define $\mathds{A}$ as the set of all outgoing transitions, i.e., $\{e_1, e_2\}$. After having computed the sets $\mathds{A}$ for all places in the initial marking, we add all transitions that are allowed by all involved places and corresponding places for the postcondition. We end up with the branching process in (b). 
	We repeat the same procedure:
	For both new environment places, we define $\mathds{A}$ as the set of outgoing transitions, in this case $\{a, b\}$. 
	For the system place $q_2$, we compute $\rec_{\dot{\varrho}}$ on the causal past (surrounded in blue) which gives us the play $[\tau_{(C, \{i\})}, i, \tau_{(D, \{b\})}]_\ind$ explaining why $\mathds{A}_{q_2} = \{b\}$. 
	As $q_2$ restricted its set $\mathds{A}$ to $b$, only transitions labelled with $b$ are added from that place. 
	We proceed this way and add more and more places and transitions. 
	The construction terminates with the strategy in (e). 
	At this point, all $\mathds{A}$-sets are such that no more transitions can be added and our construction terminates.
	The reader is encouraged to convince herself of this construction and, in particular, to comprehend how every $\mathds{A}$-set is chosen.
	
	\begin{figure}[!t]
		\begin{subfigure}[c]{0.15\textwidth}
			\centering
			\begin{tikzpicture}[scale=1.0, every label/.append style={font=\scriptsize}, label distance=-1mm]
			\node[envplace,label=west:$\color{myDGray}A$,label=east:{$\color{red}\{e_1, e_2\}$}] at (0,0.75) (p1){};
			
			\node[sysplace,label=west:$\color{myDGray}C$,label=east:{$\color{red}\{i\}$},label=south:{$\color{blue}q_1$}] at (0,-0) (p4){};
			
			\node[envplace,label=south:$ $,opacity=0.0] at (0,-3) (){};
			
			\node[token] at (0,0.75) (){};
			\node[token] at (0,0) (){};
			
			\draw[-,blue,thick,dashed] {[rounded corners=5pt] 
				(0,0.375) -- (0.7,0.375) -- (0.7,-0.5) -- (-0.7,-0.5) -- (-0.7,0.375) -- (0,0.375)};
			
			\end{tikzpicture}
			\subcaption{}
		\end{subfigure}
		\begin{subfigure}[c]{0.4\textwidth}
			\centering
			\begin{tikzpicture}[scale=1.0, every label/.append style={font=\scriptsize}, label distance=-1mm]
			\node[envplace,label=east:$\color{myDGray}A$] at (0,0.75) (p1){};
			\node[transition,label=north:$\color{myDGray}e_1$] at (-2,-0.75) (t1){};
			\node[envplace,label=east:$\color{myDGray}B$,label=south:{$\color{red}\{a,b\}$}] at (-2,-1.5) (p2){};
			\node[transition,label=north:$\color{myDGray}e_2$] at (2,-0.75) (t2){};
			\node[envplace,label=west:$\color{myDGray}B$,label=south:{$\color{red}\{a,b\}$}] at (2,-1.5) (p3){};
			
			\node[sysplace,label=east:$\color{myDGray}C$] at (0,-0) (p4){};
			\node[transition,label=east:$\color{myDGray}i$] at (0,-0.75) (t3){};
			\node[sysplace,label=south:$\color{myDGray}D$,label=east:{$\color{red}\{b\}$},label=west:$\color{blue}q_2$] at (0,-1.5) (p5){};
			
			\node[envplace,label=south:$ $,opacity=0.0] at (0,-3) (){};
			
			\node[token] at (0,0.75) (){};
			\node[token] at (0,0) (){};
			
			\draw[arrow] (p1) to (t1);
			\draw[arrow] (p1) to (t2);
			\draw[arrow] (t1) to (p2);
			\draw[arrow] (t2) to (p3);
			
			\draw[arrow] (p4) to (t3);
			\draw[arrow] (t3) to (p5);
			
			\draw[-,blue,thick,dashed] {[rounded corners=5pt] (0,-2.1) -- (0.8,-2.1) -- (0.8, -1.3) -- (0.55,-1.1) --(0.55,0.375) -- (-0.55,0.375) -- (-0.55,-1.1) -- (-0.8,-1.3) -- (-0.8, -2.1) -- (0,-2.1)};
			
			\end{tikzpicture}
			\subcaption{}
		\end{subfigure}
		\begin{subfigure}[c]{0.45\textwidth}
			\centering
			\begin{tikzpicture}[scale=1.0, every label/.append style={font=\scriptsize}, label distance=-1mm]
			\node[envplace,label=east:$\color{myDGray}A$] at (0,0.75) (p1){};
			\node[transition,label=north:$\color{myDGray}e_1$] at (-2,-0.75) (t1){};
			\node[envplace,label=east:$\color{myDGray}B$] at (-2,-1.5) (p2){};
			\node[transition,label=north:$\color{myDGray}e_2$] at (2,-0.75) (t2){};
			\node[envplace,label=west:$\color{myDGray}B$] at (2,-1.5) (p3){};
			
			\node[sysplace,label=east:$\color{myDGray}C$] at (0,-0) (p4){};
			\node[transition,label=east:$\color{myDGray}i$] at (0,-0.75) (t3){};
			\node[sysplace,label=south:$\color{myDGray}D$] at (0,-1.5) (p5){};
			
			\node[transition,label=west:$\color{myDGray}b$] at (-2,-2.25) (t4){};
			\node[transition,label=east:$\color{myDGray}b$] at (2,-2.25) (t5){};
			
			\node[envplace,label=south:$\color{myDGray}B$,label=north:{$\color{red}\{a,b\}$}] at (-2.75,-3) (p6){};
			\node[sysplace,label=south:$\color{myDGray}D$,label=east:{$\color{red}\{b\}$}] at (-1.25,-3) (p7){};
			
			\node[sysplace,label=south:$\color{myDGray}D$,label=west:{$\color{red}\{a\}$},label=north:$\color{blue}q_3$] at (1.25,-3) (p8){};
			\node[envplace,label=south:$\color{myDGray}B$,label=west:{$\color{red}\{a,b\}$}] at (2.75,-3) (p9){};

			\node[token] at (0,0.75) (){};
			\node[token] at (0,0) (){};
			
			\draw[arrow] (p1) to (t1);
			\draw[arrow] (p1) to (t2);
			\draw[arrow] (t1) to (p2);
			\draw[arrow] (t2) to (p3);
			
			\draw[arrow] (p4) to (t3);
			\draw[arrow] (t3) to (p5);
			
			\draw[arrow] (p2) to (t4);
			\draw[arrow] (p3) to (t5);
			
			\draw[arrow] (p5) to (t4);
			\draw[arrow] (p5) to (t5);
			
			%%%%%%%
			
			\draw[arrow] (t4) to (p6);
			\draw[arrow] (t4) to (p7);
			
			\draw[arrow] (t5) to (p8);
			\draw[arrow] (t5) to (p9);
			
			\draw[-,blue, thick,dashed] {[rounded corners=5pt] 
				(1.25,-3.6) -- (1.7, -3.6) -- (1.7, -2.6) -- (2.5,-2.6) --(2.5,1.25) -- (-0.5,1.25) -- (-0.5,-2.25) -- (0.5,-2.25) -- (0.5, -3.6) -- (1.25,-3.6)};
			
			%%%
			\end{tikzpicture}
			\subcaption{}
		\end{subfigure}
		
		\begin{subfigure}[c]{0.5\textwidth}
			\centering
			\begin{tikzpicture}[scale=1.0, every label/.append style={font=\scriptsize}, label distance=-1mm]
			\node[envplace,label=east:$\color{myDGray}A$] at (0,0.75) (p1){};
			\node[transition,label=north:$\color{myDGray}e_1$] at (-2,-0.75) (t1){};
			\node[envplace,label=east:$\color{myDGray}B$] at (-2,-1.5) (p2){};
			\node[transition,label=north:$\color{myDGray}e_2$] at (2,-0.75) (t2){};
			\node[envplace,label=west:$\color{myDGray}B$] at (2,-1.5) (p3){};
			
			\node[sysplace,label=east:$\color{myDGray}C$] at (0,-0) (p4){};
			\node[transition,label=east:$\color{myDGray}i$] at (0,-0.75) (t3){};
			\node[sysplace,label=south:$\color{myDGray}D$] at (0,-1.5) (p5){};
			
			\node[transition,label=west:$\color{myDGray}b$] at (-2,-2.25) (t4){};
			\node[transition,label=east:$\color{myDGray}b$] at (2,-2.25) (t5){};
			
			\node[envplace,label=south:$\color{myDGray}B$] at (-2.75,-3) (p6){};
			\node[sysplace,label=south:$\color{myDGray}D$] at (-1.25,-3) (p7){};
			
			\node[sysplace,label=south:$\color{myDGray}D$] at (1.25,-3) (p8){};
			\node[envplace,label=south:$\color{myDGray}B$] at (2.75,-3) (p9){};
			
			%%%%%
			
			\node[transition,label=west:$\color{myDGray}a$] at (2,-3.75) (t6){};

			\node[envplace,label=east:$\color{myDGray}A$,label=west:{$\color{red}\{e_1, e_2\}$}] at (2,-4.5) (p10){};
			\node[sysplace,label=west:$\color{myDGray}C$,label=south:{$\color{red}\{i\}$},label=320:{$\color{blue}q_4$}] at (2,-5.25) (p11){};
			
			%%

			%%%%%%%%%%%%%%%%%%%%%%
			\node[transition,label=north:$\color{myDGray}b$] at (-2,-3.75) (t10){};
			\node[envplace,label=south:$\color{myDGray}B$,label=north:{$\color{red}\{a,b\}$}] at (-2.75,-4.5) (p15){};
			\node[sysplace,label=south:$\color{myDGray}D$,label=north:{$\color{red}\emptyset$}] at (-1.25,-4.5) (p16){};
			
			%%%%%%%%%%%
			
			%%%%%%%%%%%
			\node[] at (0,-7.5) (){};

			\node[token] at (0,0.75) (){};
			\node[token] at (0,0) (){};
			
			\draw[arrow] (p1) to (t1);
			\draw[arrow] (p1) to (t2);
			\draw[arrow] (t1) to (p2);
			\draw[arrow] (t2) to (p3);
			
			\draw[arrow] (p4) to (t3);
			\draw[arrow] (t3) to (p5);
			
			\draw[arrow] (p2) to (t4);
			\draw[arrow] (p3) to (t5);
			
			\draw[arrow] (p5) to (t4);
			\draw[arrow] (p5) to (t5);
			
			%%%%%%%
			
			\draw[arrow] (t4) to (p6);
			\draw[arrow] (t4) to (p7);
			
			\draw[arrow] (t5) to (p8);
			\draw[arrow] (t5) to (p9);
			
			%%%
			
			\draw[arrow] (p8) to (t6);
			\draw[arrow] (p9) to (t6);
			
			\draw[arrow] (t6) to (p10);
			\draw[arrow] (t6) to[out=0,in=0] (p11);
			
			\draw[arrow] (p6) to (t10);
			\draw[arrow] (p7) to (t10);
			\draw[arrow] (t10) to (p15);
			\draw[arrow] (t10) to (p16);
			
			\draw[-,blue,thick,dashed] {[rounded corners=5pt] 
				(2,-5.9) -- (3.25,-5.9) --(3.25,1.25) -- (-0.5,1.25) -- (-0.5,-2.25) -- (0.75,-2.25) -- (0.75, -4.1) -- (2.6,-4.1) -- (2.6, -4.875) -- (1.4, -4.875) -- (1.4, -5.9) -- (2, -5.9)};

			\end{tikzpicture}
			\subcaption{}
		\end{subfigure}
		\begin{subfigure}[c]{0.5\textwidth}
			\centering
			\begin{tikzpicture}[scale=1.0, every label/.append style={font=\scriptsize}, label distance=-1mm]
			\node[envplace,label=east:$\color{myDGray}A$] at (0,0.75) (p1){};
			\node[transition,label=north:$\color{myDGray}e_1$] at (-2,-0.75) (t1){};
			\node[envplace,label=east:$\color{myDGray}B$] at (-2,-1.5) (p2){};
			\node[transition,label=north:$\color{myDGray}e_2$] at (2,-0.75) (t2){};
			\node[envplace,label=west:$\color{myDGray}B$] at (2,-1.5) (p3){};
			
			\node[sysplace,label=east:$\color{myDGray}C$] at (0,-0) (p4){};
			\node[transition,label=east:$\color{myDGray}i$] at (0,-0.75) (t3){};
			\node[sysplace,label=south:$\color{myDGray}D$] at (0,-1.5) (p5){};
			
			\node[transition,label=west:$\color{myDGray}b$] at (-2,-2.25) (t4){};
			\node[transition,label=east:$\color{myDGray}b$] at (2,-2.25) (t5){};
			
			\node[envplace,label=south:$\color{myDGray}B$] at (-2.75,-3) (p6){};
			\node[sysplace,label=south:$\color{myDGray}D$] at (-1.25,-3) (p7){};
			
			\node[sysplace,label=south:$\color{myDGray}D$] at (1.25,-3) (p8){};
			\node[envplace,label=south:$\color{myDGray}B$] at (2.75,-3) (p9){};
			
			%%%%%
			
			\node[transition,label=west:$\color{myDGray}a$] at (2,-3.75) (t6){};

			\node[envplace,label=west:$\color{myDGray}A$] at (2,-4.5) (p10){};
			\node[sysplace,label=west:$\color{myDGray}C$] at (2,-5.25) (p11){};

			\node[transition,label=north:$\color{myDGray}e_1$] at (0.75,-6) (t7){};
			\node[transition,label=west:$\color{myDGray}i$] at (2,-6) (t8){};
			\node[transition,label=north:$\color{myDGray}e_2$] at (3.25,-6) (t9){};

			\node[envplace,label=west:$\color{myDGray}B$,label=south:{$\color{red}\{a,b\}$}] at (0.75,-6.75) (p12){};
			\node[sysplace,label=south:$\color{myDGray}D$,label=east:{$\color{red}\emptyset$},label=west:$\color{blue}q_5$] at (2,-6.75) (p13){};
			\node[envplace,label=east:$\color{myDGray}B$,label=south:{$\color{red}\{a,b\}$}] at (3.25,-6.75) (p14){};

			%%%%%%%%%%%%%%%%%%%%%%
			\node[transition,label=north:$\color{myDGray}b$] at (-2,-3.75) (t10){};
			\node[envplace,label=south:$\color{myDGray}B$] at (-2.75,-4.5) (p15){};
			\node[sysplace,label=south:$\color{myDGray}D$] at (-1.25,-4.5) (p16){};
			
			%%%%%%%%%%%

			%%%%%%%%%%%
			\node[] at (0,-7.5) (){};
			
			\node[token] at (0,0.75) (){};
			\node[token] at (0,0) (){};
			
			\draw[arrow] (p1) to (t1);
			\draw[arrow] (p1) to (t2);
			\draw[arrow] (t1) to (p2);
			\draw[arrow] (t2) to (p3);
			
			\draw[arrow] (p4) to (t3);
			\draw[arrow] (t3) to (p5);
			
			\draw[arrow] (p2) to (t4);
			\draw[arrow] (p3) to (t5);
			
			\draw[arrow] (p5) to (t4);
			\draw[arrow] (p5) to (t5);
			
			%%%%%%%
			
			\draw[arrow] (t4) to (p6);
			\draw[arrow] (t4) to (p7);
			
			\draw[arrow] (t5) to (p8);
			\draw[arrow] (t5) to (p9);
			
			%%%
			
			\draw[arrow] (p8) to (t6);
			\draw[arrow] (p9) to (t6);
			
			\draw[arrow] (t6) to (p10);
			\draw[arrow] (t6) to[out=0,in=0] (p11);
			
			%%%
			
			\draw[arrow] (p11) to (t8);
			\draw[arrow] (t8) to (p13);
			
			\draw[arrow] (p10) to (t7);
			\draw[arrow] (p10) to (t9);
			\draw[arrow] (t7) to (p12);
			\draw[arrow] (t9) to (p14);

			%%%%%%%%%%%%%%%%%%%%%
			
			\draw[arrow] (p6) to (t10);
			\draw[arrow] (p7) to (t10);
			\draw[arrow] (t10) to (p15);
			\draw[arrow] (t10) to (p16);
			
			\draw[-,blue,thick,dashed] {[rounded corners=5pt] (2,-7.3) -- (2.75,-7.3) -- (2.75,-3.75) -- (3.25,-3.75) --(3.25,-2.25) -- (2.75,-2.25) -- (2.75,1.25) -- (-0.5,1.25) -- (-0.5,-4.125) -- (2.5,-4.125) -- (2.5, -4.875) -- (1.25, -4.875) -- (1.25, -7.3) -- (2,-7.3)};
			
			\end{tikzpicture}
			\subcaption{}
		\end{subfigure}
		
		\caption{Intermediate Steps in the construction of strategy $\dot{\sigma}_\varrho$ (for $\dot{\cGame_\pGame}$ from \refFig{firstRed}) from controller $\dot{\varrho}$ (cf.\ \refFig{exampleController}). The gray label is given by $\lambda$. The red labels are the transitions that should be enabled, i.e., the sets $\mathds{A}$ computed in the construction. 
			The causal past of places $q_1$ to $q_n$ is surrounded in blue.
			(e) illustrates the final strategy.}
		\label{fig:exam_starte_trans}
	\end{figure}

\end{example}

Coming back to our general translation, we can show that the construction does indeed yield a strategy. 
The observation is that each place in $\sigma_\varrho$ decides which transitions to enable (i.e., chooses $\mathds{A}$) based on its causal past only. 
The decision is therefore based solely on the place and not on the current marking. 
\begin{lemma}
	$\pNet^{\sigma_\varrho}$ is a strategy for $\pGame$.
\end{lemma}
\begin{proof}
	It is easy to verify that the constructed net $\pNet^{\sigma_\varrho}$ is a branching process of $\pGame$.
	We need to prove \textit{justified refusal}: \\
	Suppose there is a reachable marking $M$ in $\sigma_\varrho$ and transition $t$ in $\pGame$ s.t.\ $t$ is enabled in $\lambda[M]$ (i.e., $\pre{\pGame}{t} \subseteq\lambda[M]$) but there is no $t'$ with $\lambda(t') = t$ enabled in $M$. \\
	Since no such $t'$ has been added to $\sigma_\varrho$ we conclude that $t \not\in \Delta_M$. 
	
	Since we know that $\pre{\pGame}{t} \subseteq\lambda[M]$, the definition of $\Delta_M$ gives us that there is a $q \in M$ with $\lambda(q) \in \pre{\pGame}{t}$ but $t \not \in \mathds{A}_q$.
	By construction of $\mathds{A}_q$, we can conclude that $q$ is a system place. 
	We, furthermore, know that $\mathds{A}_q$ solely depends on the causal past of $M$. For every marking $M'$ that contains $q$ we always have that $t \not\in \mathds{A}_q$ and therefore $t \not\in \Delta_M$. 
	It hence holds that $t \not\in \lambda[\post{\pNet^{\sigma_\varrho}}{q}]$.
\end{proof}

\paragraph*{Strategy-Equivalence}

We can now prove that $\varrho$ and $\sigma_\varrho$ are bisimilar. 	
As relation $\relation$, we use the same one we used before and restrict it to the reachable markings in $\pNet^{\sigma_\varrho}$ and plays in $\Plays(\cGame_\pGame, \varrho)$.
We begin with a consequence of \refLemma{PtoC_sameCausalPast}.
\begin{lemma}
	If $M \relation u$, $p \in \processes$, and $q \in M \,\cap\, \inv{\lambda}[\places^{\pts{p}}]$, then computing $u' = \recv(\pastt{q})$ does not violate any assertions.
	If $u$ is maximal w.r.t.\ $\tau$-actions, i.e., there is no $\tau$ s.t.\ $u \, \tau \in \Plays(\cGame_\pGame, \varrho)$, it holds that
	$$\view_p(u') = \view_p(u)$$
	\label{lem:rec}
\end{lemma}
\begin{proof}
	We first show that computing $u' = \recv(\pastt{q})$ does not violate any assertions:
	By definition from $\relation$, it holds that $M = \fireSeq{\pNet^{\sigma_\varrho}}{\projtr{u}}$. 
	By \refLemma{PtoC_sameCausalPast}, we get that 
	\begin{align*}
	\pastt{q} = \projtr{\view_p(u)} \tag*{\myItem{(1)}}
	\end{align*}
	Since $u \in \Plays(\cGame_\pGame, \varrho)$, we know that $\view_p(u) \in \Plays(\cGame_\pGame, \varrho)$. The claim that no assertion is violated follows from \refLemma{collection}. \\
	We can now show that $\view_p(u') = \view_p(u)$. 
	By definition of $\recv$, it holds that $\projtr{u'} = \pastt{q}$.
	When using together with \myItem{(1)}, we conclude that
	$$\projtr{u'} = \projtr{\view_p(u)}$$
	It now follows that 
	\begin{align*}
	\view_p(\projtr{u'}) &= \view_p(\projtr{\view_p(u)})\\
	&= \view_p(\view_p(\projtr{u}))\\
	&= \view_p(\projtr{u})\\
	&= \projtr{\view_p(u)}
	\end{align*}
	since $\view_p(\cdot)$ is idempotent and the $\tau$-actions removed by $\projtr{\cdot}$ are local, i.e., $\projtr{\view_p(u)}) = \view_p(\projtr{u})$.
	We, furthermore, know that $u$ and $u'$ are both maximal w.r.t.\ $\tau$-actions (by assumption and from definition of $\recv$).
	Because of $\star$, every process chooses exactly one commitment set. The $\tau$-actions in both $\view_p(u')$ and $\view_p(u)$ are hence unique and we get $\view_p(u') = \view_p(u)$.
\end{proof}

The definition of $\sigma_\varrho$ is completely independent to the definition of $\relation$. \refLemma{rec}, however, characterizes a connection between both. 
In our construction of $\sigma_\varrho$, each place computes its decision (the set $\mathds{A}$) by applying $\recv$ to the transitions in its causal past. 
In $\relation$-related situations, this results, according to \refLemma{rec}, in the local view of one of the processes.
This observation allows us to show that $\varrho$ and $\sigma_\varrho$ are bisimilar. 
We can reason in both direction:
\begin{itemize}
	
	\item If $t$ is enabled in $M$ then by construction of $\sigma_\varrho$ every involved system place $q$ has allowed it, i.e., $t \in \mathds{A}_q$. 
	The set $\mathds{A}_q$ was chosen by computing the causal past of that place and convert it to a play using $\recv$. By \refLemma{rec}, each place therefore computes the local view of one of the processes on $u$ and copies the decision. Since $t$ is allowed by all involved system places, we can conclude that all involved processes must have chosen commitment sets where $t$ is included.
	Hence, $u$ can be extended by $t$ (after playing sufficiently many $\tau$-actions to choose a commitment set).
	
	\item If $u$ can be extended with $t$ by $\varrho$ then all involved processes enable $t$.
	So, every process $p \in \dom(t)$ either resides on an environment place where it has no control or it is on a system place where it must have chosen a commitment set that includes $t$. 
	Each system place in $M$ evokes $\recv$ on its causal past and, by \refLemma{rec}, therefore computes the local view of on process on $u$.
	The place then copies the decision made on that local play, i.e., copies the chosen commitment set.
	Since $t$ is in the commitment of every involved process every place $q$ involved in $t$ will allow $t$ (i.e., choose $\mathds{A}_q$ such that $t \in \mathds{A}_q$). So together the system places allow $t$ from $M$.
\end{itemize}

 We can now prove this formally. 
Since $\varrho$ is, by assumption, winning we can neglect all \Lightning-actions. 

\begin{lemma}
	If $M \relation u$ and $\fireTranTo{M}{t}{M'}$ for some $M' \in \reach(\pNet^{\sigma_\varrho})$ there exists $u' = u \,\tau^* \,t \in \Plays(\cGame_\pGame, \varrho)$ and $M' \relation u'$.
	\label{lem:PtoAl4}
\end{lemma}
\begin{proof}
	From $M \relation u$, we get $\lambda[M] = \zeta(\state{u})$ \myItem{(1)} by \refLemma{PtoC_sameLabel}.
	Since $\fireTranTo{M}{t}{M'}$, all places in $M$ that are involved in $t$ allow it, i.e., for every $q \in M$ with $\lambda(q) \in \pre{\pGame}{t}$, it holds that $t \in \lambda[\post{\pNet^{\sigma_\varrho}}{q}]$ . We hence conclude that $t \in \Delta_M$ \myItem{(2)}.\\
	Let $u_\tau$ be the trace obtained from $u$ by playing as many $\tau$-actions as possible s.t.\ there are no $\tau$-actions enabled after $u_\tau$.
	It holds that $M \relation u_\tau$. 
	By assumption $\star$, every process, that can, chooses a commitment set. For every $p$ with $\zeta(\statep{p}{u_\tau}) \in \places_\psys$, we therefore know that $\statep{p}{u_\tau} = (\_, \_)$.\\
	Assume for contradiction that $u_\tau \, t \not\in \Plays(\cGame_\pGame, \varrho)$. 
	Because of \myItem{(1)}, we know that $t$ would be possible after $u_\tau$ if the commitment sets are chosen appropriately.  
	There hence is a process ${p} \in \dom(t)$ that has chosen a commitment set that does not include $t$, i.e., $\statep{{p}}{u_\tau} = (\lambda(q_{p}), B)$ where $t \not\in B$.\\
	Let $q_{p} \in M \cap \inv{\lambda}[\places^{\pts{p}}]$ be the place that corresponds to $p$ in $M$. 
	By \refLemma{rec} and as $u_\tau$ is by assumption maximal, we now know that 
	$$\view_{p}(\recv(\pastt{q_{p}})) = \view_{p}(u_\tau)$$
	From \myItem{(1)} and as $p \in \dom(t)$, we can conclude that $\lambda(q_{p}) \in \pre{\pGame}{t}$. So, since $t \in \Delta_M$ \myItem{(2)}, we get that $t \in \mathds{A}_{q_{p}}$.	\\
	We can now analyze the construction of $\sigma_\varrho$ to observe how $\mathds{A}_{q_{p}}$ is derived. It is computed by matching 
	$$\statep{{p}}{\view_{p}(\recv(\pastt{q_{p}}))} = (\lambda(q_{p}), \mathds{A}_{q_{p}})$$
	But now 
	\begin{align*}
	(\lambda(q_{p}), B) &= \statep{{p}}{u_\tau}\\
	&= \statep{{p}}{\view_{p}(u_\tau)}\\
	&= \statep{\bar{a}}{\view_{p}(\recv(\pastt{q_{p}}))}\\
	&= (\lambda(q_{p}), \mathds{A}_{q_{p}})
	\end{align*}
	So $B = \mathds{A}_{q_{p}}$, i.e., the transitions allowed by $q_{p}$ are exactly the transitions that $p$ has chosen as a commitment set. 
	This is a contradiction since $t \in \mathds{A}_{q_{p}}$ \myItem{(2)} but by assumption $t \not\in B$.
\end{proof}

 We use the previous lemma to justify our assumption \myItem{(A)} made in the construction of $\sigma_\varrho$.
\begin{corollary}
	For any place $q \in \places^{\pNet^{\sigma_\varrho}}$ with $q \in M \cap \inv{\lambda}[\places^{\pts{p_q}}]$, computing $u = \recv(\pastt{q})$ does not violate an assertion and $\zeta(\statep{p_q}{u}) = \lambda(q)$.
\end{corollary}
\begin{proof}
	Place $q$ is part of some reachable marking $M$. 
	Using \refLemma{PtoAl4}, we get that there is some $u' \in \Plays(\cGame_\pGame, \varrho)$ with $M \relation u'$. 
	By \refLemma{rec}, computing $\recv(\pastt{q})$ does not violate any assertions.
	For the second part, we know (from \refLemma{rec}) that $\view_p(u') = \view_p(u)$.
	Now
	\begin{align*}
	\zeta(\statep{p_q}{u}) &= \zeta(\statep{p_q}{\view_{p_q}(u)})\\
	&= \zeta(\statep{p_q}{\view_{p_q}(u')})\\
	&= \zeta(\statep{p_q}{u'})\\
	&= \lambda(q)
	\end{align*}
	where the second equality follows from $\view_p(u') = \view_p(u)$ and the third from \refLemma{PtoC_sameLabel} since $M \relation u'$. 
\end{proof}

\begin{lemma}
	If $M \relation u$ and $u' = u\,t \in \Plays(\cGame_\pGame, \varrho)$ then there exists $M' \in \reach(\pNet^{\sigma_\varrho})$ with $\fireTranTo{M}{t}{M'}$ and $M' \relation u'$.
	\label{lem:PtoAl5}
\end{lemma}
\begin{proof}
	From \refLemma{PtoC_sameLabel}, we know that $\lambda[M] = \zeta(\state{u})$. So, $t$ is by construction enabled from $\lambda[M]$. \\
	Assume for contradiction that $t$ is not enabled in $M$, i.e., forbidden by the strategy. 
	Then $t \not\in \Delta_M$. Since $t$ is enabled from $\lambda[M]$, by construction of $\Delta_M$, there must be a \emph{system} place $q \in M$ with $\lambda(q) \in \pre{\pGame}{t}$ but $t \not\in \mathds{A}_{q}$ \myItem{(1)}, i.e., there is at least one place that hindered $t$ from being added to the strategy. \\
	Let $p_{q}$ be the process to which $q$ belongs, i.e., $q \in M \cap \inv{\lambda}[\places^{\pts{p_{q}}}]$. 
	Since $q$ is involved in $t$ we get that $p_{q} \in \dom(t)$. 
	Since $u' = u \, t \in \Plays(\cGame_\pGame, \varrho)$ and $q$ is a system place we get that $\statep{p_{q}}{u} = (\lambda(q), B)$ with $t \in B$ \myItem{(2)}, i.e., $p_{q}$ has chosen a commitment set that contains $t$.\\
	We know that $\mathds{A}_{q}$ for place $q$ is computed by matching $$\statep{p_{q}}{\view_{p_{q}}(\recv(\pastt{q}))} = (\lambda(q), \mathds{A}_{q})$$
	Let $u_\tau$ be $u$ extended with as many $\tau$-actions as possible (only necessary to fulfill the assumptions of \refLemma{rec}). It holds that $M \relation u_\tau$. 
	By \refLemma{rec}, we get
	$$\view_{p_{q}}(\recv(\pastt{q})) = \view_{p_{q}}(u_\tau)$$
	It now holds that
	\begin{align*}
	(\lambda(q), \mathds{A}_{q}) &= \statep{p_{q}}{\view_{p_{q}}(\recv(\pastt{q}))} \\
	&= \statep{p_{q}}{\view_{p_{q}}(u_\tau)}\\
	&= \statep{p_{q}}{u_\tau}\\
	&= \statep{p_{q}}{u}\\
	&= (\lambda(q), B)
	\end{align*} 
	Where the fourth equality holds since $p_{q}$ has already chosen a commitment set after $u$, i.e., adding more $\tau$-actions to get from $u$ to $u_\tau$ does not affect $p_{q}$.\\
	So $\mathds{A}_{q} = B$.
	This is a contradiction to $t \not\in \mathds{A}_{q}$ \myItem{(1)} and $t \in B$ \myItem{(2)}.
\end{proof}

\begin{lemma}
	If $M \relation u$ and $u' = u\, \tau \in \Plays(\cGame_\pGame, \varrho)$ then $M \relation u'$.
	\label{lem:PtoAl6}
\end{lemma}
\begin{proof}
	Obvious consequence from the definition of $\relation$.
\end{proof}

\begin{corollary}
	$\varrho$ and $\sigma_\varrho$ are bisimilar.
\end{corollary}
\begin{proof}
	By definition of $\relation$, it holds that $\init^{\pNet^{\sigma_\varrho}} \relation \epsilon$.
	Since there are no $\tau$-transitions in $\pGame$ the claim follows from \refLemma{PtoAl4}, \refLemma{PtoAl5}, and \refLemma{PtoAl6}.
\end{proof}

 We show next that a winning $\varrho$ results in a winning $\sigma_\varrho$. Since a winning controller for $\widehat{\cGame_\pGame}$ avoids all \Lightning-actions, neglecting them in our bisimulation proofs is justified. 

\begin{lemma}
	If $\varrho$ is a winning controller for $\cGame_\pGame$ or $\widehat{\cGame_\pGame}$ then $\sigma_\varrho$ is a winning strategy for $\pGame$.
	\label{lem:PtC_winning}
\end{lemma}
\begin{proof}
	We first show that $\pNet^{\sigma_\varrho}$ is finite: 
	Assume for contradiction that it is infinite. Koenig's lemma and the fact that  $\pNet^{\sigma_\varrho}$ is an occurrence net allow us to conclude that there is an infinite sequence of consecutive markings.
	By bisimilarity, any infinite sequence of markings in $\pNet^{\sigma_\varrho}$ results in an infinite $\varrho$-compatible play. 
	A contradiction since $\varrho$ is winning. \\
	Now, suppose that $M$ is a reachable final marking in $\pNet^{\sigma_\varrho}$, i.e., there are no further transitions enabled. 
	There is a $\varrho$-compatible play $u$ with $M \relation u$ and this play is maximal (up to $\tau$-actions). Since $\varrho$ is winning, $\state{u}$ must be winning (playing further $\tau$-actions does not move into winning states).
	It holds that $\zeta(\state{u}) = \lambda[M]$ (by \refLemma{PtoC_sameLabel}). So, by construction of $\cGame_\pGame$, $\lambda[M]$ is winning.
\end{proof}

\paragraph*{Deterministic Strategies}

By \refLemma{PtC_winning}, any winning controller $\varrho$ for either $\widehat{\cGame_\pGame}$ or $\cGame_\pGame$ results in a winning strategy $\sigma_\varrho$ for $\pGame$. 
If $\varrho$ is winning for $\widehat{\cGame_\pGame}$ it must additionally avoid all \Lightning-actions.
We now show that such a controller results in a deterministic $\sigma_\varrho$:
The \Lightning-action are designed such that they can occur if and only if a commitment set is chosen and two distinct actions from this set can occur. 
A winning controller for $\widehat{\cGame_\pGame}$ must avoid every \Lightning-action and therefore has to choose commitment sets where at most one transitions from every set is possible. 
In $\sigma_\varrho$, every place decides what to enable in accordance with the commitment sets chosen by $\varrho$. 
If in $\varrho$ there is at most one action from each commitment set enabled, there is at most one transition enabled from every system place in $\sigma_\varrho$. 

\begin{lemma}
	If $\varrho$ is a controller for $\widehat{\cGame_\pGame}$ such that no play in $\Plays(\widehat{\cGame_\pGame}, \varrho)$ contains a \Lightning-action, then $\sigma_\varrho$ is deterministic.
	\label{lem:PtoC_det2}
\end{lemma}
\begin{proof}
	We assume for contradiction that $\sigma_\varrho$ is not deterministic, i.e., there exists a reachable marking $M$ in $\pNet^{\sigma_\varrho}$ and a system place $q \in M$ from which two transitions $t_1, t_2 \in \post{\pNet^{\sigma_\varrho}}{q}$ are enabled.\\
	By our previous bisimulation result, there is a $u \in \Plays(\widehat{\cGame_\pGame}, \varrho)$ with $M \relation u$. Choose this $u$ such that there are no more $\tau$-actions possible.
	Because of assumption $\star$, every process on a system place has chosen a commitment set. 
	By bisimulation, we know that $u \, t_1$ and $u \, t_2$ are both in $\Plays(\widehat{\cGame_\pGame}, \varrho)$. \\
	Let $p$ be the process that $q$ belongs to, i.e., $q \in \inv{\lambda}[\places^{\pts{p}}]$. Since $q$ is in the precondition of $t_1$ and $t_2$ we have $t_1, t_2 \in \transitions^{\pts{p}}$ and, so, $p \in \dom(t_1)$ and $p \in \dom(t_2)$. 
	Since $q$ is a system place we can conclude that $\zeta(\statep{p}{u}) \in \places_\psys$ (by \refLemma{PtoC_sameLabel}) and, since in $u$ every process, that can, has chosen a commitment set, $\statep{p}{u} = (\lambda(q), B)$.
	Since $t_1$ and $t_2$ are both enabled we derive $t_1, t_2 \in B$.\\
	Now $t_1, t_2$ are both enabled from the same commitment set. 
	By construction of the \Lightning-actions, it is easy to see that $u \, \text{\Lightning}^{(\lambda(q), B)}_{[t_1, t_2]}$ is a play in $\Plays(\widehat{\cGame_\pGame})$ and, since all \Lightning-actions are uncontrollable, in $\Plays(\widehat{\cGame_\pGame}, \varrho)$.
	A contradiction.
\end{proof}

\refLemma{PtC_winning} together with \refLemma{PtoC_det2} gives us the second half of our correctness proof:
\begin{proposition}
	If $\varrho$ is a winning controller for $\cGame_\pGame$, then $\sigma_\varrho$ is a winning strategy for $\pGame$ and bisimilar to $\varrho$.
	If $\varrho$ is a winning controller for $\widehat{\cGame_\pGame}$, then $\sigma_\varrho$ is a winning, deterministic strategy for $\pGame$ and bisimilar to $\varrho$.
	\label{prop:PtoC_2}
\end{proposition}

 Combining \refProp{PtoC_1} and \refProp{PtoC_2}, we conclude \refTheo{theor1}.

\section{Singular Net Distributions}
\label{sec:appSND}

In this section, we introduce introduce a new mechanism to distribute a game.
Thereby, we generalize our translation to all \emph{concurrency-preserving} games and obtain a proof of \refTheo{theor2}.
We can observe that the notion of slices is too strict for our purposes: 
Our translation requires to distribute the global movement of the Petri game into local behavior.
A \emph{partitioning} of the places (as prescribed by slice distributions) is not necessarily needed. 
Requiring such a partitioning is what enables proofs as the one above and hence limit the applicability of slices. 
%The limiting factor is that even in concurrency-preserving and safe nets, it might not be possible to tell in advance what token will reside on a place. 
%As an example in \refFig{nonDet}, there is at all times at most one token on each place, but depending on the execution it might be different ones. 
%In a slice distribution we would hence need to assign the same place to two slice, something not possible. 

\subparagraph{Singular Net Distribution}
We introduce the new concepts of \emph{singular nets} (SN) and \emph{singular net distributions} (SND). 
We later see how our translation can be modified to work with SNs instead of slices.

Before giving a formal description, we consider the example in \refFig{PtoC_snDist}.
The Petri net in (a) comprises three tokens of which two reside on the same place. 
As the net is not safe it is not sliceable. 
In (b) and (c), two possible singular net distributions of (a) are given. The black label (annotated with a hat) is the name of the node, whereas the gray label is the one given by $\pi$. 
The singular nets share transitions. If we, e.g., consider the SND in (b), the labelling of the initial marking agrees with the initial marking of (a) and both transitions $a$ and $b$ can be matched by some copy ($\hat{a_1}, \hat{a_2}$, $\hat{b_1}$).
By observing the SNDs in both (b) and (c), it becomes clear that both are valid distributions of the behavior in (a). 

\begin{figure}[t]
	\begin{subfigure}[c]{0.4\textwidth}
		%\vspace{-1cm}
		\begin{center}
			\begin{tikzpicture}[scale=1.0, every label/.append style={font=\scriptsize}, label distance=-1mm]
			\path 	(0,0) node[envplace,label=west:$A$](p1) {}
			(1.5,0) node[envplace,label=east:$B$](p2) {}
			(0,-1.5) node[envplace,label=west:$C$](p3) {}
			(1.5,-1.5) node[envplace,label=east:$D$](p4) {}
			
			(0.75,-0.75) node[transition,label=west:$a$](t1) {}
			(1.5,-0.75) node[transition,label=east:$b$](t2) {};

			\path	(-0.1,0) node[token]() {};
			\path	(0.1,0) node[token]() {};
			\path	(1.5,0) node[token]() {};

			\draw[arrow] (p1) -- (t1);
			\draw[arrow] (p2) -- (t1);
			\draw[arrow] (t1) -- (p3);
			\draw[arrow] (t1) -- (p4);
			\draw[arrow] (p4) -- (t2);
			\draw[arrow] (t2) -- (p2);
			
			\end{tikzpicture}
		\end{center}
		\subcaption{}
	\end{subfigure}
	\begin{subfigure}[c]{0.6\textwidth}
		\begin{center}
			\begin{tikzpicture}[scale=1.0, every label/.append style={font=\scriptsize}, label distance=-1mm]
			\path 	(0,0) node[envplace,label=west:\textcolor{myLGray}{$A$},label=east:$\hat{A}$](p11) {}
			(0,-1.5) node[envplace,label=west:\textcolor{myLGray}{$C$},label=east:$\hat{B}$](p12) {}
			(0,-0.75) node[transition,label=west:\textcolor{myLGray}{$a$},label=east:$\hat{a_1}$](t11) {}
			
			(2,0) node[envplace,label=west:\textcolor{myLGray}{$A$},label=east:$\hat{C}$](p21) {}
			(2,-1.5) node[envplace,label=west:\textcolor{myLGray}{$C$},label=east:$\hat{D}$](p22) {}
			(2,-0.75) node[transition,label=west:\textcolor{myLGray}{$a$},label=east:$\hat{a_2}$](t21) {}
			
			(4,0) node[envplace,label=west:\textcolor{myLGray}{$B$},label=east:$\hat{E}$](p31) {}
			(4,-1.5) node[envplace,label=west:\textcolor{myLGray}{$D$},label=east:$\hat{F}$](p32) {}
			(3.5,-0.75) node[transition,label=north:\textcolor{myLGray}{$a$},label=south:$\hat{a_1}$](t31) {}
			(4.5,-0.75) node[transition,label=west:\textcolor{myLGray}{$a$},label=east:$\hat{a_2}$](t32) {}
			(5.5,-0.75) node[transition,label=north:\textcolor{myLGray}{$b$},label=south:$\hat{b_1}$](t33) {};

			\path	(0,0) node[token]() {};
			\path	(2,0) node[token]() {};
			\path	(4,0) node[token]() {};

			\draw[arrow] (p11) -- (t11);
			\draw[arrow] (t11) -- (p12);
			
			\draw[arrow] (p21) -- (t21);
			\draw[arrow] (t21) -- (p22);
			
			\draw[arrow] (p31) -- (t31);
			\draw[arrow] (p31) -- (t32);
			\draw[arrow] (t31) -- (p32);
			\draw[arrow] (t32) -- (p32);
			\draw[arrow] (p32) -- (t33);
			\draw[arrow] (t33) -- (p31);
			\end{tikzpicture}
		\end{center}
		\subcaption{}
	\end{subfigure}\\[0.3cm]
	\begin{subfigure}[c]{0.6\textwidth}
		\begin{center}
			\begin{tikzpicture}[scale=1.0, every label/.append style={font=\scriptsize}, label distance=-1mm]
			\path 	(0,0) node[envplace,label=west:\textcolor{myLGray}{$A$},label=east:$\hat{A}$](p11) {}
			(0,-1.5) node[envplace,label=west:\textcolor{myLGray}{$D$},label=east:$\hat{a}$](p12) {}
			(0,-3) node[envplace,label=west:\textcolor{myLGray}{$B$},label=east:$\hat{C}$](p13) {}
			(0,-0.75) node[transition,label=west:\textcolor{myLGray}{$a$},label=east:$\hat{a_1}$](t11) {}
			(0,-2.25) node[transition,label=west:\textcolor{myLGray}{$b$},label=east:$\hat{b_2}$](t12) {}
			(1,-2.25) node[transition,label=north:\textcolor{myLGray}{$a$},label=east:$\hat{a_3}$](t13) {}
			
			(2.5,0) node[envplace,label=west:\textcolor{myLGray}{$A$},label=east:$\hat{D}$](p21) {}
			(2.5,-1.5) node[envplace,label=west:\textcolor{myLGray}{$C$},label=east:$\hat{E}$](p22) {}
			(2,-0.75) node[transition,label=north:\textcolor{myLGray}{$a$},label=west:$\hat{a_2}$](t21) {}
			(3,-0.75) node[transition,label=west:\textcolor{myLGray}{$a$},label=east:$\hat{a_3}$](t22) {}
			
			(5,0) node[envplace,label=west:\textcolor{myLGray}{$B$},label=east:$\hat{F}$](p31) {}
			(4.5,-1.5) node[envplace,label=west:\textcolor{myLGray}{$C$},label=south:$\hat{G}$](p32) {}
			(5.5,-1.5) node[envplace,label=west:\textcolor{myLGray}{$D$},label=south:$\hat{H}$](p33) {}
			(4.5,-0.75) node[transition,label=north:\textcolor{myLGray}{$a$},label=west:$\hat{a_1}$](t31) {}
			(5.5,-0.75) node[transition,label=west:\textcolor{myLGray}{$a$},label=east:$\hat{a_2}$](t32) {}
			(6.5,-0.75) node[transition,label=north:\textcolor{myLGray}{$b$},label=south:$\hat{b_1}$](t33) {};

			\path	(0,0) node[token]() {};
			\path	(2.5,0) node[token]() {};
			\path	(5,0) node[token]() {};

			\draw[arrow] (p11) -- (t11);
			\draw[arrow] (t11) -- (p12);
			\draw[arrow] (p12) -- (t12);
			\draw[arrow] (t12) -- (p13);
			\draw[arrow] (p13) -- (t13);
			\draw[arrow] (t13) -- (p12);
			
			\draw[arrow] (p21) -- (t21);
			\draw[arrow] (t21) -- (p22);
			\draw[arrow] (p21) -- (t22);
			\draw[arrow] (t22) -- (p22);
			
			\draw[arrow] (p31) -- (t31);
			\draw[arrow] (p31) -- (t32);
			\draw[arrow] (t31) -- (p32);
			\draw[arrow] (t32) -- (p33);
			\draw[arrow] (p33) -- (t33);
			\draw[arrow] (t33) -- (p31);
			
			\end{tikzpicture}
		\end{center}
		\subcaption{}
	\end{subfigure}
	\begin{subfigure}[c]{0.4\textwidth}
		\begin{center}
			\begin{tikzpicture}[scale=1.0, every label/.append style={font=\scriptsize}, label distance=-1mm]
			\path 	(-1.75,0) node[envplace,label=west:\textcolor{myLGray}{$A$},label=east:$\hat{A}$](p11) {}
			(-1.75,-1.5) node[envplace,label=west:\textcolor{myLGray}{$C$},label=east:$\hat{B}$](p12) {}
			%(2,-0.75) node[transition,label=west:\textcolor{myLGray}{$a$},label=east:$a_1$](t11) {}
			
			(1.75,0) node[envplace,label=west:\textcolor{myLGray}{$A$},label=east:$\hat{C}$](p21) {}
			(1.75,-1.5) node[envplace,label=west:\textcolor{myLGray}{$C$},label=east:$\hat{D}$](p22) {}
			%(6,-0.75) node[transition,label=west:\textcolor{myLGray}{$a$},label=east:$a_2$](t21) {}
			
			(0,0) node[envplace,label=west:\textcolor{myLGray}{$B$},label=east:$\hat{E}$](p31) {}
			(0,-1.5) node[envplace,label=west:\textcolor{myLGray}{$D$},label=east:$\hat{F}$](p32) {}
			(-1,-0.75) node[transition,label=north:\textcolor{myLGray}{$a$},label=south:$\hat{a_1}$](t31) {}
			(1,-0.75) node[transition,label=north:\textcolor{myLGray}{$a$},label=south:$\hat{a_2}$](t32) {}
			(0,-0.75) node[transition,label=west:\textcolor{myLGray}{$b$},label=east:$\hat{b_1}$](t33) {};
			
			\node[] at (0,-3) (){};
			
			\path	(0,0) node[token]() {};
			\path	(-1.75,0) node[token]() {};
			\path	(1.75,0) node[token]() {};

			\draw[arrow] (p11) -- (t31);
			\draw[arrow] (t31) -- (p12);
			
			\draw[arrow] (p21) -- (t32);
			\draw[arrow] (t32) -- (p22);
			
			\draw[arrow] (p31) -- (t31);
			\draw[arrow] (p31) -- (t32);
			\draw[arrow] (t31) -- (p32);
			\draw[arrow] (t32) -- (p32);
			\draw[arrow] (p32) -- (t33);
			\draw[arrow] (t33) -- (p31);
			\end{tikzpicture}
		\end{center}
		\subcaption{}
	\end{subfigure}
	
	\caption{A Petri game (a) and two possible distributions in singular nets (b) and (c). 
		In (b) and (c), labels of the SND ($\pi$) are given in gray.
		The label of each node in an SND is annotated with a hat and labelled in black. 
		Note that transitions can be shared between SNs.
		In (d) the composition of the SND in (b) is depicted.
	}
	\label{fig:PtoC_snDist}
\end{figure}
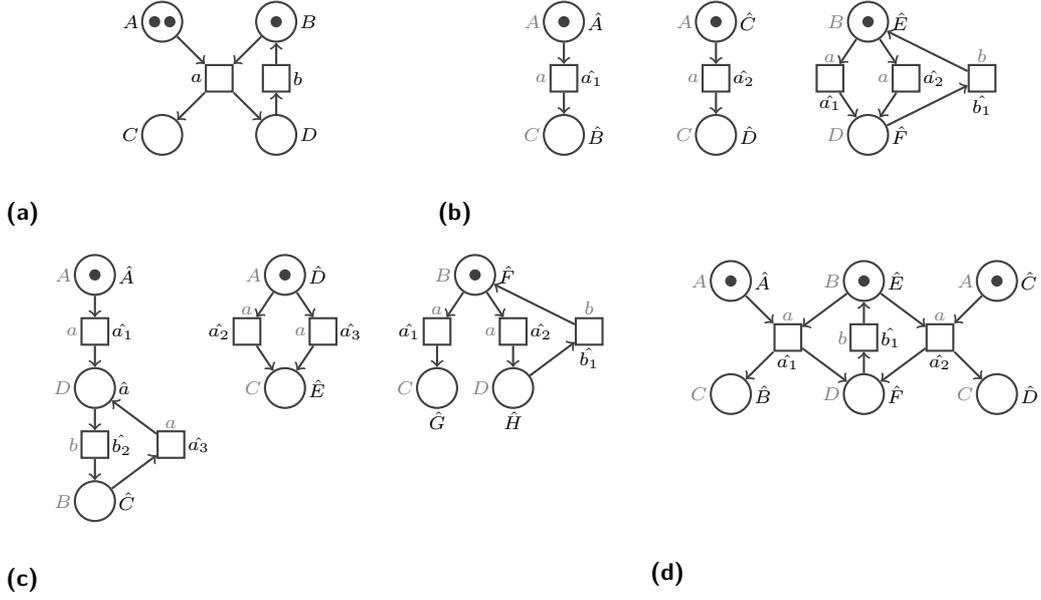

Throughout this section let $\pNet$ be a finite, concurrency-preserving Petri net.
We now proceed and give a formal description of both SN and SNDs.

\begin{definition}
	A \emph{singular net} (SN) of $\pNet$ is a pair $(\slice, \pi)$ where $\slice = (\places^\slice, \transitions^\slice, \flow^\slice, \init^\slice)$ is a Petri net satisfying
	\begin{center}
		$|\init^\slice| = 1$ and $\forall t \in \transitions^\slice: \, |\pre{}{t}| = |\post{}{t}| = 1$
	\end{center}
	and $\pi : \places^\slice \cup \transitions^\slice \to \places^\pNet \cup \transitions^\pNet$ is a mapping with the following properties:
	
	\def\arraystretch{1.6}
	\begin{tabular}{ll}
		\myItem{(1)} \; $\pi(\places^\slice) \subseteq \places^\pNet$ and $\pi(\transitions^\slice) \subseteq \transitions^\pNet$ & \myItem{(2)} \; $\forall q_1, q_2 \in \places^\slice : \, \pi(q_1) \neq \pi(q_2)$ \\ 
		\myItem{(3)} \; $\pi(\init^\slice) \subseteq \init^\pNet$ &  \myItem{(4)} \; $\forall q \in \places^\slice : \, \post{\pNet}{\pi(q)} \subseteq \pi(\transitions^\slice)$\\ 
		\multicolumn{2}{l}{\myItem{(5)} \; $\forall x, y \in \places^\slice \cup \transitions^\slice: \, (x, y) \in \flow^\slice \Leftrightarrow (\pi(x), \pi(y)) \in \flow^\pNet$} 
	\end{tabular} 
	
\end{definition}

\noindent A singular net can be thought of as a generalized slice.
The underlying net describes the movement of a single token.
Instead of viewing it as a subnet of $\pNet$ (as we have done for slices), we label it using $\pi$.
This labelling should satisfy five properties, most of which correspond to properties lifted from the definition of slices:
$\pi$ must respect the node type \myItem{(1)} and copy each place at most once \myItem{(2)}.
Singular nets of finite nets are hence finite.
The initial marking must be labelled within the initial marking of $\pNet$ \myItem{(3)}. 
%In a slice all transitions leaving a place in the original net are added to the slice.
Similar to the definition of slices, we require that all transitions leaving the label of some place are represented by at least one copy \myItem{(4)}.
Lastly, the flow relation adds a flow between two nodes if and only if there is a flow between the labels of the nodes in $\pNet$ \myItem{(5)}. 
%As we have done for branching processes we view a singular net $(\slice, \pi)$ as the underlying net $\slice$ and assume that $\pi$ is implicitly given.

Singular nets are, similar to slices, defined as nets describing behavior of individual tokens. 
To model global behavior in the end, we want to compose multiple singular nets to obtain a description of a system involving more than one player.

\begin{definition}
	If $\pNet$ is a Petri net and $\slices = \{ (\slice_i, \pi_i)\}_{i \in \indx}$ with $\slice_i = (\places^i, \transitions^i, \flow^i, \init^i)$ is a finite family of singular nets for $\pNet$, 
	we call $\slices$ \emph{compatible} if
	$$\places^{\slice_i} \cap \places^{\slice_j} = \emptyset \text{ for all } i, j \in \indx \text{ with } i \neq j$$
	and 
	$$\forall t : \, t \in \transitions^{\slice_i} \cap \transitions^{\slice_j} \Rightarrow \pi_{\slice_i}(t) = \pi_{\slice_j}(t)$$

	If $\slices$ is compatible, we define the \emph{composition} of $\slices$ as the pair $(\parcomp, \pi_{\parcomp})$ where 
	$$\parcomp = (\places^{\parcomp}, \transitions^{\parcomp}, \flow^{\parcomp}, \init^{\parcomp})$$
	with $\places^{\parcomp} = \biguplus_{i \in \indx} \places^i$, \, $\transitions^{\parcomp} = \bigcup_{i \in \indx} \transitions^i$, \, $\flow^{\parcomp} = \biguplus_{i \in \indx} \flow^i$, $\init^{\parcomp} = \biguplus_{i \in \indx} \init^i$, 
	and
	$$\pi_{\parcomp} = \bigcup_{i \in \indx} \pi_i$$
\end{definition}

\noindent  As for slices, we require a family of singular nets to contain disjoint sets of places. As each SN is furthermore labelled with $\pi$, we require that shared transitions are labelled equally among all singular nets.
Then, $\parcomp$ is defined as the Petri net obtained by taking the union of places, transitions, flows and initial markings. Since only transitions can be shared all unions except for them are disjoint. 
As $\pNet$ is compatible we know that for all transitions the label agrees in all SNs. We can hence label the nodes in $\parcomp$ with nodes in $\pNet$, i.e., design $\pi_{\parcomp}$ as the union of all individual labelling functions. 
Note that unlike for slices the composition $\parcomp$ in general differs from $\pNet$. In \refFig{PtoC_snDist} (d), the composition of the SN-family in (b) is depicted. The $\pi$-label of the composition is given in gray. 

The labelling of an SND allows us to split up places and transitions. 
We want to distribute a Petri net into a family of singular nets, that together show the same behavior as the Petri net.
We can hence define what a family of singular nets should suffice to be a valid \emph{distribution} of a net:
\begin{definition}
	A \emph{singular net distribution} (SND) for Petri net $\pNet$ is a compatible family~$\slices$ of singular nets for $\pNet$ where the composition $(\parcomp, \pi_{\parcomp})$ fulfills:
	\begin{enumerate}
		\item[\myItem{(1)}] $\pi(\init^{\parcomp}) = \init^\pNet$

		\item[\myItem{(2)}] For every transition $t \in \transitions^{\parcomp}$, $\pi(\pre{\parcomp}{t}) = \pre{\pNet}{\pi(t)}$ and $\pi(\post{\parcomp}{t})) = \post{\pNet}{\pi(t)}$.
		
		\item[\myItem{(3)}] For every $t_1, t_2 \in \transitions^{\parcomp}$ with $\pre{\parcomp}{t_1} = \pre{\parcomp}{t_2}$ and $\pi(t_1) = \pi(t_2)$, it holds that $t_1 = t_2$.
		
		\item[\myItem{(4)}] For reachable markings $M \in \reach(\parcomp)$ and subsets $C \subseteq M$ with $\pi(C) = \pre{\pNet}{t}$ for some $t \in \transitions^\pNet$, there exists a transition $t' \in \transitions^{\parcomp}$ with $\pi(t') = t$ and $\pre{\parcomp}{t'} = A$. 
	\end{enumerate}
\end{definition}

\noindent
A singular net distribution is a compatible family of singular nets, i.e., a family with disjoint places and equally labelled shared transitions.
The additional restrictions guarantee that the composition of the SNs shows the same behavior as the original net. They are reminiscent of the definition of a branching process and unfolding.
Restriction \myItem{(1)} requires the initial marking of $\parcomp$ to be labelled within the initial marking of $\pNet$, whereas \myItem{(2)} requires the composition to preserve the structure on transitions.
Together, \myItem{(1)} and \myItem{(2)} state that $\pi_{\parcomp}$ is an initial homomorphism from $\parcomp$ to $\pNet$.
%To avoid splits that are useless, i.e., add multiple, equally labelled transitions with the same precondition, we require $\pi$ to be injective on such transitions (5).
As for a branching process, \myItem{(3)} requires $\pi_{\parcomp}$ to be injective on transitions with the same precondition: Equally labelled transitions must occur from distinct situations. 
An SND is almost identical to a branching process with the exception of not requiring an underlying occurrence net and, furthermore, being described in terms of local token movements.
%An SND can hence be thought of as a branching process that is defined in terms of local SNs. %Every transitions that is added, describes the same behavior as the equally labelled transition in $\pNet$.
Lastly, requirement \myItem{(4)} is similar to the one found in the definition of an unfolding. It is a maximality criterion requiring that, for every situation where there are tokens on places in $C$, every transition possible from $\pi_{\parcomp}(C)$ is matched by some copy. 
While we can split up places in an SND, \myItem{(4)} requires us to still add transitions from every possible combination of the new copies. 
Both families of singular nets in \refFig{PtoC_snDist}~(b) and (c) form singular net distributions of the net in (a).

%Consider for instance transition $a$ in \refFig{PtoC_snDist}. And the SND in (b). If the first SN is in place $A$ and the third in place $E$ then there should be transition labelled $a$ possible (in this case $a_1$). Similar for the places $C$ and $E$. 
%While splitting up places, an SND must hence still ensure that every transition that would be possible must have been added. 

Note that our notion of a singular net distribution agrees with slice distributions if we enforce to have only one copy of each place. In this case, we can choose $\pi$ as the identity. 
Every sliceable net has an SND.

\subparagraph{Properties of SNDs}

An SND is defined as a family of singular nets such that their composition is both structure-preserving \myItem{(2)} and at the same time captures all behavior \myItem{(4)}.
It is easy to see that an SND describes the exact behavior of a Petri net. 

\begin{corollary}
	$\reach(\pNet) = \pi(\reach(\parcomp))$. 
\end{corollary}

\noindent %an SND hence describes exactly he same behavior as $\pNet$ with the only exception that is splits up nodes by equally labelled copies. 
Our main motivation for defining SNs and SNDs is to generalize our previous translation to allow for a broader class of games.
We already remarked that every sliceable net has an~SND.
In \refFig{PtoC_snDist}, we saw that even some non-sliceable nets have an SND.
Nets with SND are thus a \textit{strict superset} of sliceable nets. 
The next theorem shows that the class of Petri nets, that can be distributed into singular nets, can be characterized precisely: It is exactly the class of concurrency-preserving Petri nets. 
\begin{proposition}
	Every finite, concurrency-preserving Petri net has an SND.
	\label{prop:allSND}
\end{proposition}
\begin{proof}
	We present a constructive proof. Given a Petri net $\pNet = (\places^\pNet, \transitions^\pNet, \flow^\pNet, \init^\pNet)$, we build $\abs{\init^\pNet}$ many singular nets. 
	Each of these nets initially consists of a copy of the places in~$\pNet$, i.e., $\abs{\init^\pNet}$-many copies of $\pNet$ without any transitions. 
	So, $\slices = \{(\slice_1, \pi_1), \cdots, (\slice_{\abs{\init^\pNet}}, \pi_{\abs{\init^\pNet}})\}$ where $\slice_i = (\places^i, \transitions^i, \flow^i, \init^i)$ with $\places^i = \{q_i \mid q \in \places^\pNet\}$. 
	Define $\pi_i(q_i) = q$.
	From each $\slice_i$, we select a single place and add it to a set $D$ s.t.\ $\pi(D) = \init^\pNet$ (this is always possible). We put one token on each of these selected places, resulting in one token in the initial marking of each SN.\\
	Now, define $(\parcomp, \pi_{\parcomp})$ as the composition of the singular nets.
	We incrementally add transitions to the SNs:
	We iterate over every reachable marking $M$ in $\parcomp$ and consider every set $C \subseteq M$ where $\pi_{\parcomp}(C) = \pre{\pNet}{t}$ for a transition $t$ and there is no $t'$ with $\pi_{\parcomp}(t') = t$ and $\pre{\parcomp}{t'} = C$.
	The set of SNs involved in $C$ is $\nabla_C = \{i \mid \places^i \cap C \neq \emptyset\}$. 
	We create a new transition $t'$, define $\pi_{\parcomp}(t') = t$, and add it to all SNs with places contained in $C$.
	$$\transitions^i = \begin{cases}
	\begin{aligned}[t]
	&\transitions^i \quad &\text{ if } i \not\in \nabla_C\\
	&\transitions^i \cup \{t'\} \quad &\text{ if } i \in \nabla_C
	\end{aligned}
	\end{cases}$$
	We extend the flow of every SN in $\nabla_C$ s.t.\ $\pre{\slice_i}{t'} = C \cap \places^i$.
	We pick a set of places $C'$ s.t.\ $\nabla_{C'} = \nabla_{C}$ and $\pi_{\parcomp}(C') = \post{\pNet}{t}$.
	We hence assign for each involved SN a place such that the label of $C'$ agrees with $\post{\pNet}{t}$. We note that there might be many such combinations but there is at least one. We extend the flow of every SN in $\nabla_C$ ($\nabla_{C'}$) such that $\post{\slice_i}{t'} = C' \cap \places^i$. 
	Afterwards, we recompute the composition $(\parcomp, \pi_{\parcomp})$ with the newly added transitions and repeat until no more transitions can be added. \\
	We iterate this and thereby add more and more transitions. 
	Since we deal with a finite number of places and transitions the construction terminates. 
	Since we add exactly the transitions required in an SND it can easily be checked that each net is a singular net and the resulting family is a singular net distribution. 
	%The construction hence terminates. It is easy to check that the obtained family of singular nets is a singular net distribution of $\pNet$ with bounds
	%$$\abs{\bigcup_{\slice \in \slices} \places^\slice} \leq \abs{\init^\pNet} \cdot \abs{\places^\pNet}$$
	%and
	%$$\abs{\bigcup_{\slice \in \slices} \transitions^\slice} \leq \abs{\transitions^\pNet} \cdot
	%\mathit{max}_{1 \leq k \leq \abs{\places^\pNet}} {{\abs{\init^\pNet}}\choose{k}} = \abs{\transitions^\pNet} \cdot {{\abs{\init^\pNet}}\choose{0.5 \cdot \abs{\init^\pNet}}}$$ These bounds are tight. TODO: Fix, faculty
\end{proof}

\paragraph*{Branching processes of SNDs}

In the long run, we want to extend SNDs to Petri games and use them for our translation.  
Since strategies are defined in terms of branching processes we begin by comparing branching processes for an SND with ones for the original net. 
Assume $\pNet$ is a Petri net,  $\slices$ is an SND for~$\pNet$, and $(\parcomp, \pi)$ the composition of $\slices$. 
We analyze and compare possible branching processes for both $\pNet$ and $\parcomp$.

Let $\iota = (\pNet^\iota, \lambda)$ be a branching process for $\pNet$ and $\iota_{\slices} = (\pNet_{\slices}^\iota, \lambda_{\slices})$ a branching process for~$\parcomp$. 
$\lambda$ labels the nodes from $\pNet^\iota$ with nodes in $\pNet$ whereas $\lambda_{\slices}$ labels the nodes of $\pNet_{\slices}^\iota$ in $\parcomp$. All nodes in $\parcomp$ are themself, by $\pi$, labelled in $\pNet$.
A branching process of $\parcomp$ hence has a finer label; instead of being labelled in nodes from $\pNet$ directly, it is labelled in an intermediate entity, namely $\parcomp$, that is itself labelled in $\pNet$.
%\looparrowright
We define 
$$\iota \leftrightsquigarrow \iota_{\slices} \; \Leftrightarrow \; \pNet^\iota = \pNet_{\slices}^\iota \; \land \; \lambda = \pi \circ \lambda_{\slices}$$
$\leftrightsquigarrow$ relates a branching processes for $\parcomp$ and $\pNet$ iff the underlying occurrence net is identical and the labelling of $\iota_{\slices}$ is finer than that of $\iota$, i.e., agrees when made coarser by applying~$\pi$. 
It is intuitive that $\leftrightsquigarrow$-related branching processes describe equivalent restrictions of the Petri net.

%We can state this observation for general Petri nets:
%\begin{corollary}
%	%Let $\pNet$ be a Petri net, $\{\slice\}_{\slice \in {\slices}}$ an SND of $\pNet$. 
%	If $\mathfrak{U} = (\pNet^\mathfrak{U}, \lambda)$ is the unfolding of $\pNet$ and $\mathfrak{U}_{\slices} = (\pNet_{\slices}^\mathfrak{U}, \lambda_{\slices})$ the unfolding of the $\parcomp$, then $\mathfrak{U} \leftrightsquigarrow \mathfrak{U}_{\slices}$.
%\end{corollary}

%Strategies for Petri games are branching processes of the underlying net fulfilling justified refusal. By the definition of an SND it describes the behavior of the Petri net. 

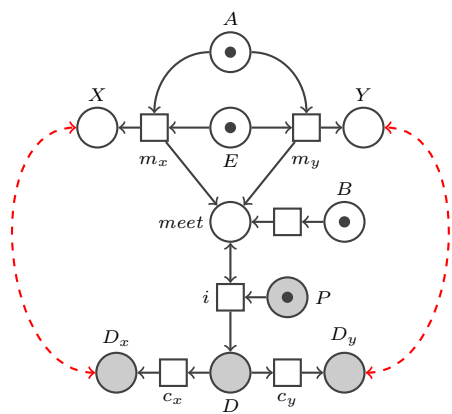
\begin{figure}[t]
	\begin{center}
		\begin{tikzpicture}[scale=1.0, every label/.append style={font=\scriptsize}, label distance=-0.5mm]
		\node[envplace,label=south:$E$] at (0,-0.75) (p1) {};
		\node[envplace,label=north:$X$] at (-1.75,-0.75) (p2) {};
		\node[envplace,label=north:$Y$] at (1.75,-0.75) (p3) {};
		
		\node[envplace,label=north:$A$] at (0,0.25) (p4) {};
		
		\node[envplace,label=north:$B$] at (1.5,-2) (p5) {};
		
		\node[envplace,label=west:$\mathit{meet}$] at (0,-2) (p6) {};
		
		%Dec
		\node[sysplace,label=south:$D$] at (0,-4) (p7) {};
		\node[sysplace,label=east:$P$] at (0.75,-3) (p8) {};
		
		\node[sysplace,label={[]north:$D_x$}] at (-1.5,-4) (p9) {};
		\node[sysplace,label={[]north:$D_y$}] at (1.5,-4) (p10) {};
		
		%%%%%%%
		
		\node[transition,label={[yshift=-1mm]south:$m_x$}] at (-1, -0.75) (t1) {};
		\node[transition,label={[yshift=-1mm]south:$m_y$}] at (1, -0.75) (t2) {};
		
		\node[transition] at (0.75, -2) (t6) {};
		
		\node[transition,label=west:$i$] at (0, -3) (t7) {};
		
		\node[transition,label=south:$c_x$] at (-0.75, -4) (t8) {};
		\node[transition,label=south:$c_y$] at (0.75, -4) (t9) {};
		
		\node[token] at (0,0.25) (){};
		\node[token] at (0,-0.75) (){};
		\node[token] at (1.5,-2) (){};
		\node[token] at (0.75,-3) (){};

		\draw[arrow] (p1) -- (t1);
		\draw[arrow] (p1) -- (t2);
		\draw[arrow] (t1) -- (p2);
		\draw[arrow] (t2) -- (p3);
		
		\draw[arrow] (p4) to[out=180,in=90] (t1);
		\draw[arrow] (p4) to[out=0,in=90] (t2);
		\draw[arrow] (t1) -- (p6);
		\draw[arrow] (t2) -- (p6);
		
		\draw[arrow] (p5) -- (t6);
		\draw[arrow] (t6) -- (p6);
		
		\draw[darrow] (p6) -- (t7);
		
		\draw[arrow] (p8) -- (t7);
		\draw[arrow] (t7) -- (p7);
		
		\draw[arrow] (p7) -- (t8);
		\draw[arrow] (p7) -- (t9);
		\draw[arrow] (t8) -- (p9);
		\draw[arrow] (t9) -- (p10);

		\draw[darrow, red, dashed] (p9) to[out=180, in=180] (p2);
		\draw[darrow, red, dashed] (p10) to[out=0, in=0] (p3);

		\end{tikzpicture}
	\end{center}
	
	\caption{Petri Game without a winning strategy. The player starting in $P$ should copy the decision made by the player in $E$, indicated with the red-arrows. 
	}
	\label{fig:PtoC_commitment}
\end{figure}

We can show the following, as an SND preserves both the structure and every possible transition is added:
\begin{corollary} For every branching process $\iota_{\slices}$ of $\parcomp$, there exists a branching process $\iota$ for $\pNet$ with $\iota \leftrightsquigarrow \iota_{\slices}$\\
	For every branching process $\iota$ of $\pNet$,  there exists a branching process $\iota_{\slices}$ for $\parcomp$ with $\iota \leftrightsquigarrow \iota_{\slices}$
\end{corollary}
\noindent For every branching process of either $\pNet$ or $\parcomp$, there hence exists an equivalent one for $\parcomp$  or $\pNet$, i.e., one with a finer or coarser labelling.

\paragraph*{Translating Games using SND}

We can now adopt the previous concepts to Petri games, i.e., mark the places in an SN as either system and environment, and require that $\pi$ respects this distribution. 
Let $\pGame$ be a concurrency-preserving Petri game, $\slices$ an SND for $\pGame$, and $(\parcomp, \pi)$ the composition of $\slices$. 
While for every \emph{branching process} for $\pGame$ there exists an equivalent (defined by $\leftrightsquigarrow$) branching process for $\parcomp$ and vice versa,  this does not hold for \emph{strategies}. 
It still holds that for every strategy of $\pGame$ there exists an equivalent one for $\parcomp$, but the reverse does not hold in general: 
A strategy for $\parcomp$ can distinguish between copies of transition even though they have the same $\pi$-label (i.e., belong to the same transition in $\pGame$). 
Since an SND splits up transitions a strategy for the composition can be more restrictive without violating justified refusal. 
There even exist games where the composition of an SND has a winning strategy even though the original game has not.

As an example, we consider the Petri game in \refFig{PtoC_commitment}. 
It comprises four players: An environment player that generates inputs starting in $E$, two dummy players starting in $A$ and $B$ as well as a system player starting in $P$. 
The player starting in $E$ can use transition $m_x$ or $m_y$ and thereby move to $X$ or $Y$ and synchronize with the dummy player in $A$.
Upon synchronization, $A$ hence moves to the place $\mathit{meet}$ whereas the dummy player in $B$ can move there directly. 
The system player that is initially in place $P$ can synchronize with a token on $\mathit{meet}$ on $i$ and afterwards use $c_x$ or $c_y$ to move to $D_x$ or $D_y$.
To win the game, the system player should copy the decision of the environment, i.e., move to $D_x$ iff the $E$ moves to $X$. This wining criterion can be expressed in either reachability and safety games. 
This game has \emph{no} winning strategy:
Both dummy players do not possess the same information since only the one starting in $A$ knows the decision that needs to be copied.
To copy the player from $E$ reliably, the system player in $P$ needs to share transition $i$ with the player starting in $A$ since this is the only source of the much needed information. Communication with the player from $B$ does not provide any relevant information. 
Justified refusal, however, prohibits strategies that can guarantee communication with the token starting in $A$ and not with the one from $B$.
%Note that the game in \refFig{PtoC_commitment} is not sliceable.

A possible singular net distribution of the Petri game in \refFig{PtoC_commitment} is depicted in \refFig{ex_snDist}~(a). The gray label is the one given by $\pi$. 
The name of the node is depicted in black where each name is equipped with a hat to aid readability. 
The place $\mathit{meet}$ is split up into two places $\hat{\mathit{meet}_1}$ and $\hat{\mathit{meet}_2}$. The transition $i$ is split up into $\hat{i_1}$ and $\hat{i_2}$. 
\refFig{ex_snDist}~(b) delineates the composition of the SND in (a). The $\pi$-label is omitted to aid readability. 
Unlike the initial game from \refFig{PtoC_commitment}, the composition (b) has a winning strategy, since a strategy could forbid $\hat{i_2}$ while allowing $\hat{i_1}$ without violating justified refusal. 
When applying the coarser label to this winning strategy (i.e., applying $\pi$ pointwise), the resulting branching process is no strategy for the original game.

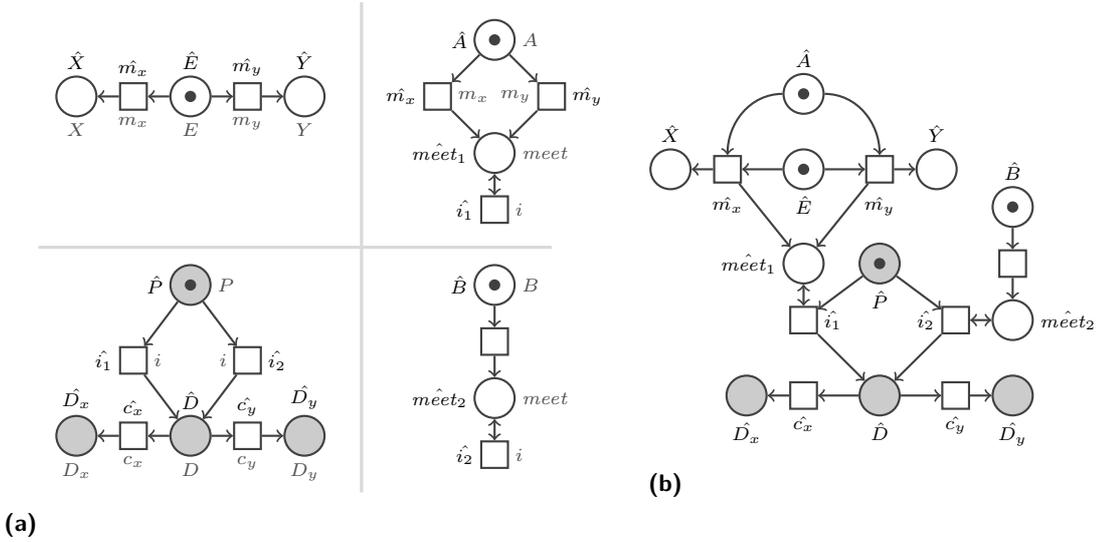
\begin{figure}[t!]
	
	\begin{subfigure}[c]{0.6\textwidth}
		\centering
		\begin{tikzpicture}[scale=1.0, every label/.append style={font=\scriptsize}, label distance=-0.5mm]
		\node[envplace,label=south:$\color{myDGray}E$,label=north:$\hat{E}$] at (0,0) (p11) {};
		\node[envplace,label=south:$\color{myDGray}X$,label=north:$\hat{X}$] at (-1.5,0) (p12) {};
		\node[envplace,label=south:$\color{myDGray}Y$,label=north:$\hat{Y}$] at (1.5,0) (p13) {};

		\node[transition,label=south:$\color{myDGray}m_x$,label=north:$\hat{m_x}$] at (-0.75, 0) (t11) {};
		\node[transition,label=south:$\color{myDGray}m_y$,label=north:$\hat{m_y}$] at (0.75, 0) (t12) {};
		
		\node[token] at (0,0) (){};
		
		\draw[arrow] (p11) -- (t11);
		\draw[arrow] (p11) -- (t12);
		\draw[arrow] (t11) -- (p12);
		\draw[arrow] (t12) -- (p13);

		%%%%%%%%%%%%%%%%%%%%%%%%%%%%%%%%%%%%%%%%%%%%%%%%%%%%%%%%%
		
		\node[envplace,label=east:$\color{myDGray}A$,label=west:$\hat{A}$] at (4,0.75) (p24) {};
		
		\node[transition,label=east:$\color{myDGray}m_x$,label=west:$\hat{m_x}$] at (3.25, 0) (t21) {};
		\node[transition,label=west:$\color{myDGray}m_y$,label=east:$\hat{m_y}$] at (4.75, 0) (t22) {};
		
		\node[envplace,label=east:$\color{myDGray}\mathit{meet}$,label=west:$\hat{\mathit{meet}_1}$] at (4,-0.75) (p26) {};
		
		\node[token] at (4,0.75) (){};
		
		\draw[arrow] (p24) to (t21);
		\draw[arrow] (p24) to (t22);
		\draw[arrow] (t21) -- (p26);
		\draw[arrow] (t22) -- (p26);
		
		\node[transition,label=east:$\color{myDGray}i$,label=west:$\hat{i_1}$] at (4, -1.5) (t27) {};
		
		\draw[darrow] (p26) -- (t27);

		%%%%%%%%%%%%%%%%%%%%%%%%%%%%%%%%%%%%%%%%%%%%%%%%%%%%%%%%%%%%%%%
		
		\node[envplace,label=east:$\color{myDGray}B$,label=west:$\hat{B}$] at (4,-2.5) (p35) {};
		
		\node[envplace,label=east:$\color{myDGray}\mathit{meet}$,label=west:$\hat{\mathit{meet}_2}$] at (4,-4) (p36) {};

		\node[transition] at (4, -3.25) (t36) {};
		
		\node[transition,label=east:$\color{myDGray}i$,label=west:$\hat{i_2}$] at (4, -4.75) (t37) {};
		
		\node[token] at (4,-2.5) (){};
		
		\draw[arrow] (p35) -- (t36);
		\draw[arrow] (t36) -- (p36);
		
		\draw[darrow] (p36) -- (t37);

		%%%%%%%%%%%%%%%%%%%%%%%%%%%%%%%%%%%
		
		\node[sysplace,label=south:$\color{myDGray}D$,label=north:$\hat{D}$] at (0,-4.5) (p47) {};
		\node[sysplace,label=east:$\color{myDGray}P$,label=west:$\hat{P}$] at (0,-2.5) (p48) {};
		
		\node[sysplace,label=south:$\color{myDGray}D_x$,label=north:$\hat{D_x}$] at (-1.5,-4.5) (p49) {};
		\node[sysplace,label=south:$\color{myDGray}D_y$,label=north:$\hat{D_y}$] at (1.5,-4.5) (p410) {};

		\node[transition,label=east:$\color{myDGray}i$,label=west:$\hat{i_1}$] at (-0.75, -3.5) (t47) {};
		\node[transition,label=west:$\color{myDGray}i$,label=east:$\hat{i_2}$] at (0.75, -3.5) (t47i) {};
		
		\node[transition,label=south:$\color{myDGray}c_x$,label=north:$\hat{c_x}$] at (-0.75, -4.5) (t48) {};
		\node[transition,label=south:$\color{myDGray}c_y$,label=north:$\hat{c_y}$] at (0.75, -4.5) (t49) {};
		
		\node[token] at (0,-2.5) (){};

		\draw[arrow] (p48) -- (t47);
		\draw[arrow] (t47) -- (p47);
		\draw[arrow] (p48) -- (t47i);
		\draw[arrow] (t47i) -- (p47);
		
		\draw[arrow] (p47) -- (t48);
		\draw[arrow] (p47) -- (t49);
		\draw[arrow] (t48) -- (p49);
		\draw[arrow] (t49) -- (p410);
		
		%%%%%%%

		\draw[-, black!15, very thick] (-2,-2) to (4.75,-2);
		\draw[-, black!15, very thick] (2.25,1.25) to (2.25,-5.25);
		
		\end{tikzpicture}
		
		\subcaption{}
	\end{subfigure}
	\begin{subfigure}[c]{0.4\textwidth}
		\centering
		\begin{tikzpicture}[scale=1.0, every label/.append style={font=\scriptsize}, label distance=-0.5mm]
		\node[envplace,label=south:$\hat{E}$] at (-0.75,-0.75) (p1) {};
		\node[envplace,label=north:$\hat{X}$] at (-2.5,-0.75) (p2) {};
		\node[envplace,label=north:$\hat{Y}$] at (1,-0.75) (p3) {};
		
		\node[envplace,label=north:$\hat{A}$] at (-0.75,0.25) (p4) {};
		
		\node[envplace,label=north:$\hat{B}$] at (2,-1.25) (p5) {};
		
		\node[envplace,label=west:$\hat{\mathit{meet}_1}$] at (-0.75,-2) (p6) {};
		
		\node[envplace,label=east:$\hat{\mathit{meet}_2}$] at (2,-2.75) (p6i) {};
		
		%Dec
		\node[sysplace,label=south:$\hat{D}$] at (0.25,-3.75) (p7) {};
		\node[sysplace,label=south:$\hat{P}$] at (0.25,-2) (p8) {};
		
		\node[sysplace,label=south:$\hat{D_x}$] at (-1.5,-3.75) (p9) {};
		\node[sysplace,label=south:$\hat{D_y}$] at (2,-3.75) (p10) {};
		
		%%%%%%%
		
		\node[transition,label={[yshift=-1mm]south:$\hat{m_x}$}] at (-1.75, -0.75) (t1) {};
		\node[transition,label={[yshift=-1mm]south:$\hat{m_y}$}] at (0.25, -0.75) (t2) {};
		
		\node[transition] at (2, -2) (t6) {};
		
		\node[transition,label=east:$\hat{i_1}$] at (-0.75, -2.75) (t7) {};
		\node[transition,label=west:$\hat{i_2}$] at (1.25, -2.75) (t7i) {};
		
		\node[transition,label=south:$\hat{c_x}$] at (-0.75, -3.75) (t8) {};
		\node[transition,label=south:$\hat{c_y}$] at (1.25, -3.75) (t9) {};
		
		\node[token] at (-0.75,0.25) (){};
		\node[token] at (-0.75,-0.75) (){};
		\node[token] at (2,-1.25) (){};
		\node[token] at (0.25,-2) (){};

		\draw[arrow] (p1) -- (t1);
		\draw[arrow] (p1) -- (t2);
		\draw[arrow] (t1) -- (p2);
		\draw[arrow] (t2) -- (p3);
		
		\draw[arrow] (p4) to[out=180,in=90] (t1);
		\draw[arrow] (p4) to[out=0,in=90] (t2);
		\draw[arrow] (t1) -- (p6);
		\draw[arrow] (t2) -- (p6);
		
		\draw[arrow] (p5) -- (t6);
		\draw[arrow] (t6) -- (p6i);
		
		\draw[darrow] (p6) -- (t7);
		\draw[darrow] (p6i) -- (t7i);
		
		\draw[arrow] (p8) -- (t7);
		\draw[arrow] (t7) -- (p7);
		\draw[arrow] (p8) -- (t7i);
		\draw[arrow] (t7i) -- (p7);
		
		\draw[arrow] (p7) -- (t8);
		\draw[arrow] (p7) -- (t9);
		\draw[arrow] (t8) -- (p9);
		\draw[arrow] (t9) -- (p10);

		%\draw[darrow, red, dashed] (p9) to[out=180, in=180] (p2);
		%\draw[darrow, red, dashed] (p10) to[out=0, in=0] (p3);

		\end{tikzpicture}
		\subcaption{}
	\end{subfigure}

	\caption{A singular net distribution for the Petri game in \refFig{PtoC_commitment} (a) and the composition of the distribution in (b). The black label is the name of the node whereas the gray label is the one given by $\pi$. To aid readability, the name of nodes in the SND are annotated with a hat and the $\pi$-label is omitted in (b).}
	\label{fig:ex_snDist}
\end{figure}

While we cannot find equivalent strategies between $\pGame$ and $\parcomp$, we can, however, find equivalent strategies, if strategies for $\parcomp$ do not distinguish between equally $\pi$-labelled transitions. This motivates the following definition:
\begin{definition}
	A strategy $\sigma = (\pNet^\sigma, \lambda)$ for $\parcomp$ is \emph{$\pi$-insensitive}, if for any pairwise concurrent set of places $C$ with $\lambda[C] = \pre{\pGame}{t}$ for some transition $t$ there either is a transition $t'$ with $\lambda(t') = t$ and $\pre{\pNet^\sigma}{t'} = C$ or there is a system place $q \in C \cap \inv{\lambda}[\places_\psys]$ with $\pi(t) \not\in \pi(\lambda[\post{\sigma}{q}])$.
\end{definition}

\noindent The definition is almost identical to the one of a strategy. 
The original justified refusal requires that every transition, that is not added to the strategy, must be uniformly forbidden by a system place. By contrast, in a $\pi$-insensitive strategy, there must be a system place that uniformly forbids \emph{all} transitions with the same $\pi$-label. 
Even though a transition is duplicated, a $\pi$-insensitive strategy considers all transitions with the same label as identical. 
The interested reader is advised to check that the composition in \refFig{ex_snDist} (b) has no winning strategy that is $\pi$-insensitive.

We can show that if a strategy for $\parcomp$ is $\pi$-insensitive and we apply the coarse label, the resulting branching process fulfills justified refusal, i.e., is a strategy:
\begin{corollary}
	If $\sigma = (\pNet^\sigma, \lambda)$ is a $\pi$-insensitive strategy for $\parcomp$ then $\sigma' = (\pNet^\sigma, \pi \circ \lambda)$ is a strategy for $\pGame$. 
\end{corollary}

\subparagraph{Extending the Translation}
We can do a similar translation as before but work with singular net distributions instead of slice distributions by treating SNs as slices, i.e., ignoring the $\pi$ label.
In \refFig{exma_comm_translation} (a), the translated automaton for the SND from \refFig{ex_snDist} is depicted. 
We already saw that a composition of an SND might have a winning strategy even though the original game has not (cf.\ \refFig{ex_snDist} (b)). If we build our translation from an SND we run into the same problem:
We give potential controllers to much power, by allowing them to distinguish equally $\pi$-labelled transitions (using their commitment sets) and therefore restrict the behavior in a way that the strategy of the Petri game cannot.
For example, the control game in \refFig{exma_comm_translation} (a) has a winning controller: As in composition of the SND, a controller can distinguish between $\hat{a_1}$ and $\hat{a_2}$ and therefore enforce communication with the player that possesses the information needed to win the game.

We fix this by modifying our translation slightly: We restrict the commitment sets for each process to transitions in the original game instead of the copies in the SND\footnote{For a place $q$ in the SND, we do not allow all commitment sets $A \subseteq \post{}{q}$ but $A \subseteq \pi(\post{}{q})$.}.
From such a commitment set, all copies of a transition in the set are allowed.
With the coarser commitment sets, a controller can no longer distinguish equally labelled transitions and has to allow either all copies of a transition or none. 
If we translate the singular net distribution from \refFig{ex_snDist} with the modified translation, the fourth singular net yields the process in \refFig{exma_comm_translation} (b).
If we substitute this process into the overall control game in (a) the resulting control game has no longer a winning controller, as $\hat{i_1}$ and $\hat{i_2}$ can no longer be distinguished.

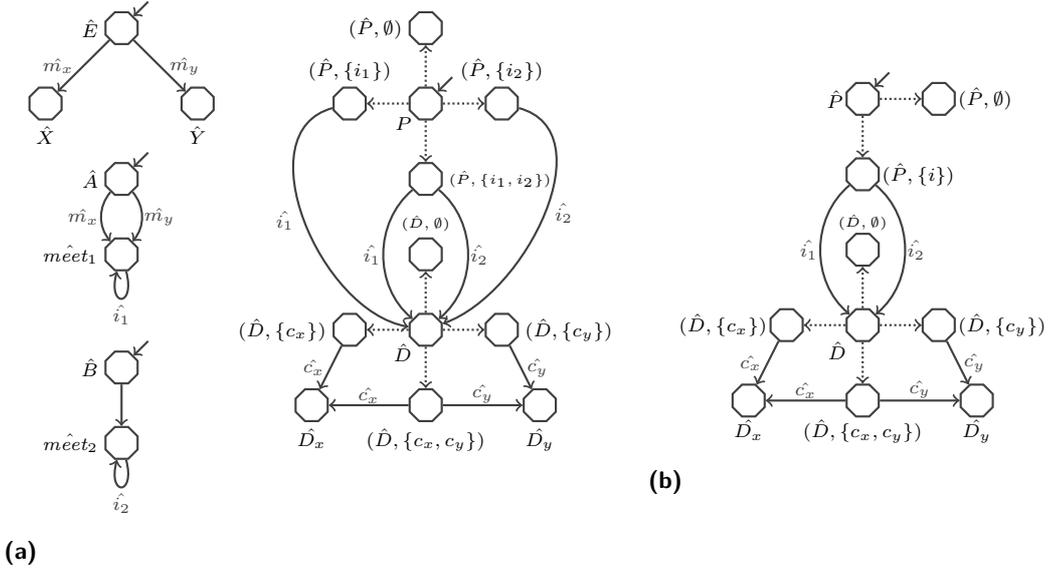
\begin{figure}[t!]
	\begin{subfigure}[c]{0.6\textwidth}
		\centering
		\begin{tikzpicture}[scale=1.0, every label/.append style={font=\scriptsize}, label distance=-0.5mm]
		
		\node[aastate,label=west:{$\hat{E}$}] at (0,0) (s11){}; 
		\node[aastate,label=south:{$\hat{X}$}] at (-1,-1) (s12){}; 
		\node[aastate,label=south:{$\hat{Y}$}] at (1,-1) (s13){}; 
		
		\draw[arrow] (s11)+(0.35, 0.35) to (s11);
		
		\draw[arrow] (s11) to node[left] {\scriptsize$\hat{m_x}$} (s12);
		\draw[arrow] (s11) to node[right] {\scriptsize$\hat{m_y}$} (s13);

		%%%%%%%%%%%%%%%%%%%%%%%%%%%%%%%%%%%%%%%%%%%%%%%%%%%%%%%%
		
		\node[aastate,label=west:{$\hat{A}$}] at (0,-2) (s21){}; 
		\node[aastate,label=west:{$\hat{\mathit{meet}_1}$}] at (0,-3) (s22){}; 
		
		\draw[arrow] (s21)+(0.35, 0.35) to (s21);
		
		\draw[arrow] (s21) to[out=230,in=130] node[left=-1mm] {\scriptsize$\hat{m_x}$} (s22);
		\draw[arrow] (s21) to[out=310,in=50] node[right=-1mm] {\scriptsize$\hat{m_y}$} (s22);
		
		\draw[arrow,loop below] (s22) to node[below=-0.5mm] {\scriptsize$\hat{i_1}$} (s22);

		%%%%%%%%%%%%%%%%%%%%%%%%%%%%%%%%%%%%%%%%%%%%%%%%%%%%%%%%%%%%%%%%%%%%%%%
		\node[aastate,label=west:{$\hat{B}$}] at (0,-4.5) (s31){}; 
		\node[aastate,label=west:{$\hat{\mathit{meet}_2}$}] at (0,-5.5) (s32){}; 
		
		\draw[arrow] (s31)+(0.35, 0.35) to (s31);
		
		\draw[arrow] (s31) to (s32);
		
		\draw[arrow,loop below] (s32) to node[below=-0.5mm] {\scriptsize$\hat{i_2}$} (s32);
		
		%%%%%%%%%%%%%%%%%%%%%%%%%%%%%%%%%%%%%%%%%%%%%%%%%%%%%%%%%%%%%%%%%%%%%%
		
		\node[aastate,label={[label distance=-1.5mm]225:{$P$}}] at (4,-1) (s41){}; 
		
		\node[aastate,label=west:{$(\hat{P}, \emptyset)$}] at (4,0) (s42){}; 
		\node[aastate,label=north:{$(\hat{P}, \{i_1\})$}] at (3,-1) (s43){}; 
		\node[aastate,label=north:{$(\hat{P}, \{i_2\})$}] at (5,-1) (s44){}; 
		\node[aastate,label=east:{\tiny$(\hat{P}, \{i_1, i_2\})$}] at (4,-2) (s45){}; 
		
		\node[aastate,label={[label distance=-1.5mm]225:{$\hat{D}$}}] at (4,-4) (s46){}; 
		
		\node[aastate,label=north:{\tiny$(\hat{D}, \emptyset)$}] at (4,-3) (s47){}; 
		\node[aastate,label=west:{$(\hat{D}, \{c_x\})$}] at (3,-4) (s48){}; 
		\node[aastate,label=east:{$(\hat{D}, \{c_y\})$}] at (5,-4) (s49){}; 
		\node[aastate,label=south:{$(\hat{D}, \{c_x, c_y\})$}] at (4,-5) (s410){}; 
		
		\node[aastate,label=south:{$\hat{D_x}$}] at (2.5,-5) (s411){}; 
		\node[aastate,label=south:{$\hat{D_y}$}] at (5.5,-5) (s412){}; 
		
		\draw[arrow] (s41)+(0.35, 0.35) to (s41);
		
		\draw[arrow, densely dotted] (s41) to (s42);
		\draw[arrow, densely dotted] (s41) to (s43);
		\draw[arrow, densely dotted] (s41) to (s44);
		\draw[arrow, densely dotted] (s41) to (s45);

		\draw[arrow] (s43) to[out=195,in=170] node[left] {\scriptsize$\hat{i_1}$} (s46);
		\draw[arrow] (s44) to[out=345,in=20] node[right] {\scriptsize$\hat{i_2}$} (s46);
		\draw[arrow] (s45) to[out=220,in=140]  node[left=-1mm] {\scriptsize$\hat{i_1}$} (s46);
		\draw[arrow] (s45) to[out=320,in=40] node[right=-1mm] {\scriptsize$\hat{i_2}$} (s46);

		\draw[arrow, densely dotted] (s46) to  (s47);
		\draw[arrow, densely dotted] (s46) to (s48);
		\draw[arrow, densely dotted] (s46) to (s49);
		\draw[arrow, densely dotted] (s46) to (s410);

		\draw[arrow] (s48) to node[left=-0.5mm] {\scriptsize$\hat{c_x}$} (s411);
		\draw[arrow] (s49) to node[right=-0.5mm] {\scriptsize$\hat{c_y}$} (s412);
		\draw[arrow] (s410) to node[above=-1mm] {\scriptsize$\hat{c_x}$} (s411);
		\draw[arrow] (s410) to node[above=-1mm] {\scriptsize$\hat{c_y}$} (s412);

		\end{tikzpicture}
		\subcaption{}
	\end{subfigure}
	\begin{subfigure}[c]{0.4\textwidth}
		\centering
		\begin{tikzpicture}[scale=1.0, every label/.append style={font=\scriptsize}, label distance=-1mm]
		\node[aastate,label={west:{$\hat{P}$}}] at (2,-3) (s41){}; 
		
		\node[aastate,label=east:{$(\hat{P}, \emptyset)$}] at (3,-3) (s44){}; 
		\node[aastate,label=east:{$(\hat{P}, \{i\})$}] at (2,-4) (s45){}; 
		
		\node[aastate,label=225:{$\hat{D}$}] at (2,-6) (s46){}; 
		
		\node[aastate,label=north:{\tiny$(\hat{D}, \emptyset)$}] at (2,-5) (s47){}; 
		\node[aastate,label=west:{$(\hat{D}, \{c_x\})$}] at (1,-6) (s48){}; 
		\node[aastate,label=east:{$(\hat{D}, \{c_y\})$}] at (3,-6) (s49){}; 
		\node[aastate,label=south:{$(\hat{D}, \{c_x, c_y\})$}] at (2,-7) (s410){}; 
		
		\node[aastate,label=south:{$\hat{D_x}$}] at (0.5,-7) (s411){}; 
		\node[aastate,label=south:{$\hat{D_y}$}] at (3.5,-7) (s412){}; 
		
		\draw[arrow] (s41)+(0.35, 0.35) to (s41);
		
		\draw[arrow, densely dotted] (s41) to (s44);
		\draw[arrow, densely dotted] (s41) to (s45);

		\draw[arrow] (s45) to[out=220,in=140]  node[left=-1mm] {\scriptsize$\hat{i_1}$} (s46);
		\draw[arrow] (s45) to[out=320,in=40] node[right=-1mm] {\scriptsize$\hat{i_2}$} (s46);

		\draw[arrow, densely dotted] (s46) to  (s47);
		\draw[arrow, densely dotted] (s46) to (s48);
		\draw[arrow, densely dotted] (s46) to (s49);
		\draw[arrow, densely dotted] (s46) to (s410);

		\draw[arrow] (s48) to node[left=-0.5mm] {\scriptsize$\hat{c_x}$} (s411);
		\draw[arrow] (s49) to node[right=-0.5mm] {\scriptsize$\hat{c_y}$} (s412);
		\draw[arrow] (s410) to node[above=-1mm] {\scriptsize$\hat{c_x}$} (s411);
		\draw[arrow] (s410) to node[above=-1mm] {\scriptsize$\hat{c_y}$} (s412);
		\end{tikzpicture}
		
		\subcaption{}
	\end{subfigure}
	
	\caption{Translation of the Petri game from \refFig{PtoC_commitment} using the singular net distribution from \refFig{ex_snDist} (a). The local automaton of the first three singular nets is depicted in (a). If we do our unmodified translation the fourth SN is translated to the automaton in (b). The modified translation designed for SNDs yields the automaton in (c). }
	\label{fig:exma_comm_translation}
\end{figure}

%We can show that we can transfer strategies and controller between $\pGame$ and the modified translation $\cGame_\pGame$. 

\subparagraph{Translating Strategies to Controllers}

Given a winning strategy $\sigma$ for $\pGame$, we outline that there exists a bisimilar winning controller for the modified $\cGame_\pGame$. We can refine the labels of $\sigma$ to obtain a strategy $\sigma_{\slices}$ for $\parcomp$. It holds that $\sigma \leftrightsquigarrow \sigma_{\slices}$. Since $\sigma$ satisfies justified refusal we get that $\sigma_{\slices}$ is $\pi$-insensitive, i.e., if a place forbids a transition it forbids all transitions with the same $\pi$-label. 
We can now do the same controller construction as in Appendix~\ref{sec:PtoC_translateStrtoCont} on the strategy $\sigma_{\slices}$.
Since $\sigma_{\slices}$ is $\pi$-insensitive every transition that is not added must be forbidden together with all equally labelled transitions. The controller can hence choose an appropriate commitment set, even though the selection of sets does not allow to distinguish equally labelled transitions. 
As $\sigma \leftrightsquigarrow \sigma_{\slices}$, it is easy to see that the obtained controller and $\sigma$ are bisimilar\footnote{For the bisimulation, we identify every transition in the control game with its $\pi$-label.}.

\subparagraph{Translating Controllers to Strategies}
Given a controller for the modified $\cGame_\pGame$, we can construct a bisimilar strategy for $\pGame$. We first build a strategy $\sigma_{\slices}$ for $\parcomp$ using the construction from Appendix~\ref{sec:PGtoAA2}. As the commitment sets of the control game range over original transitions rather than copies, we observe that the resulting $\sigma_{\slices}$ is $\pi$-insensitive, i.e., equally labelled transitions are not distinguished. 
When taking the coarser label, we obtain a branching process $\sigma$ for $\pGame$. Since $\sigma_{\slices}$ is $\pi$-insensitive, $\sigma$ fulfills justified refusal, i.e., is a strategy for $\pGame$. 
Bisimilar behavior follows since $\sigma \leftrightsquigarrow \sigma_{\slices}$.

\paragraph*{The General Result}
We generalized the translation to work with SNDs instead of slice distributions. 
In \refProp{allSND}, we showed that all concurrency-preserving nets (and hence games) have an SND.
Therefore, we can derive the generalization of our initial result stated in \refTheo{theor2}.

\section{Lower Bound}
\label{sec:appLB1}

In this section, we prove that our translation is asymptomatically optimal in size if we require strategy-equivalence by giving an exponential lower bound. 
While this does not answer the question whether there is a sub-exponential translation, it highlights that such a translation would inevitably destroy the structure of the game.

Before we can proceed with a concrete analyses, we need to agree on the parameters used to define the size of a control game and of a Petri game. 
Two natural parameters that are well suited for measuring the size of a control game are the number of local states $\abs{\bigcup_{p \in \processes} S_p}$ as well as the number of actions $\abs{\Sigma}$.
Conversely, for a Petri game, the number of places $\abs{\places}$ and transitions $\abs{\transitions}$ are good candidates\footnote{We remark at this point, that the more concise communication scheme of control games over Petri games allows us to hide additional complexity that is made explicit in Petri games.
In a concurrency-preserving Petri game, the size of the flow relation $\flow$ is always polynomial in the number of places and transitions. 
In a control game, the transition function $\{\delta_a\}_{a \in \Sigma}$ can be of exponential size in the number of local states.
Our resulting game $\cGame_\pGame$ can, however, be described as the parallel composition of local automata. The local transition relation is therefore again polynomial in $\abs{\bigcup_{p \in \processes} S_p}$ and $\abs{\Sigma}$.
For our analysis, we restrict us to the number of local states and actions in the alphabet.}. 

\paragraph*{Lower Bound}

\begin{figure}[!t]
	\centering
	\begin{tikzpicture}[scale=1.0, every label/.append style={font=\scriptsize}, label distance=-1mm]
	\node[envplace,label=east:$A$] at (0,0) (p11){};
	\node[envplace,label=east:$B$] at (0,-1.5) (p12){};
	\node[envplace, specialEnv,label=east:$C$] at (0,-3) (p13){};
	
	\node[transition,label=west:$a$] at (-0.5,-0.75) (t11){};
	\node[transition,label=east:$b$] at (0.5,-0.75) (t12){};
	
	\node[transition,label=south:$t_1$] at (-1,-2.25) (t13){};
	\node[] at (-0.5,-2.25) (){\tiny$\cdots$};
	\node[transition,draw=black!50] at (0,-2.25) (t14){};
	\node[] at (0.5,-2.25) (){\tiny$\cdots$};
	\node[transition,label=south:$t_n$] at (1,-2.25) (t15){};
	
	\node[token] at (0,0)(){};
	
	\draw[arrow] (p11) -- (t11);
	\draw[arrow] (p11) -- (t12);
	\draw[arrow] (t11) -- (p12);
	\draw[arrow] (t12) -- (p12);
	
	\draw[arrow] (p12) -- (t13);
	\draw[arrow,black!50] (p12) -- (t14);
	\draw[arrow] (p12) -- (t15);
	
	\draw[arrow] (t13) -- (p13);
	\draw[arrow,black!50] (t14) -- (p13);
	\draw[arrow] (t15) -- (p13);

	%%%%%%%%%%%
	\node[sysplace, specialSys,label=north:$D$] at (3,-1.5) (p21){};
	
	\node[transition,label=south:$t_1$] at (2,-2.25) (t21){};
	\node[] at (2.5,-2.25) (){\tiny $\cdots$};
	\node[transition,draw=black!50] at (3,-2.25) (t22){};
	\node[] at (3.5,-2.25) (){\tiny $\cdots$};
	\node[transition,label=south:$t_n$] at (4,-2.25) (t23){};
	
	\node[token] at (3, -1.5)(){};
	
	\draw[darrow] (p21) -- (t21);
	\draw[darrow,black!50] (p21) -- (t22);
	\draw[darrow] (p21) -- (t23);
	
	\end{tikzpicture}
	
	\caption{Slices of a concurrency-preserving (reachability) Petri game family $\{\pGame_n\}_{n \in \mathbb{N}}$ where every strategy-equivalent control game is of exponential size.}
	\label{fig:PtoC_lowerBound}
\end{figure}
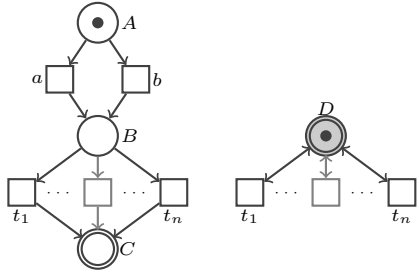

Consider the Petri game family $\{\pGame_n\}_{n \in \mathbb{N}}$ obtained as the composition of the slices in \refFig{PtoC_lowerBound}. We fix any $n$ and refer to $\pGame_n$ as $\pGame$. 
In the initial marking of $\pGame$, both $a$ and $b$ can fire resulting in a  marking $M = \{B, D\}$.
From here, all transitions $t_1, \cdots, t_n$ are enabled. 
$\pGame$~is played between two players.
Both of which possess different information, i.e., the first player (starting in $A$) knows whether $a$ or $b$ occurred, while the second (starting in $D$) does not. 
Only the second player can decide which of the transitions $t_1, \cdots, t_n$ should be possible. 
The decision of which of the transitions in $t_1, \cdots, t_n$  to allow can hence not be based on the occurrence of $a$ or $b$. 
Since a winning strategy for $\pGame$ can restrict any combinations of $t_i$-transitions, every strategy-equivalent control game must admit controllers that can do the same. At the same time, the decision, which of the $t_i$s to enable, cannot be based on the occurrence of $a$ or $b$, since this would imply a strategy for $\pGame$ that can do the same. 
Unlike Petri games that can naturally express that the second player can restrict transitions $t_1, \cdots, t_n$, while the first one cannot (using system and environment places), control games are limited to controllable or uncontrollable actions.
We now show that this already results in exponentially many global states. 

For a strategy $\sigma$ for $\pGame$, we write $\Seq(\sigma)$ for the set of sequences admitted by $\sigma$.
For every~$\sigma$, it holds that $\Seq(\sigma) = \{ \epsilon, a, b \} \cup \{ a \, t, b \, t \mid t \in B \}$ for some $B \subseteq \{t_1, \cdots, t_n\}$.
Conversely, for every such $B$, there exists a $\sigma$ s.t.\ $\Seq(\sigma)$ has exactly this form.
In particular, a strategy cannot base the decision of what $t_i$s to enable on the occurrence of $a$ or $b$.

Consider any control game $\cGame$ that is strategy-equivalent to $\pGame$. 
Due to $\cGame$ being a translation of $\pGame$, we assume $\transitions \subseteq \Sigma$.

\begin{lemma}
	$a$ and $b$ are uncontrollable. 
	\label{lem:unncessary}
\end{lemma}
\begin{proof}
	Choose the (winning) strategy $\sigma$ for $\pGame$ as the one that allows everything, i.e., $\pNet^\sigma = \pGame^\mathfrak{U}$.
	Let $\varrho_\sigma = \{ f^{\varrho_\sigma}_p\}_{p \in \processes}$ be a bisimilar controller for $\cGame$.
	There exists a relation $\relation$ with $\init^{\pNet^\sigma} \relation \epsilon$. 
	
	Since $\varrho_\sigma$ is winning every play must be finite.
	Now, we consider every play $u$ in $\Plays(\cGame, \varrho_\sigma)$ that only consists of $\tau$-actions and is maximal w.r.t.\ $\tau$-actions, i.e., $u$ cannot be extended by another $\tau$.
	By bisimulation, we know that $\init^{\pNet^\sigma} \relation u$. Since both $a$ and $b$ are possible from~$\init^{\pNet^\sigma}$, i.e., there is a marking $M$ with $\fireTranTo{\init^{\pNet^\sigma}}{a}{M}$ and $\fireTranTo{\init^{\pNet^\sigma}}{b}{M}$, we know that $a$ and~$b$ must be the only extensions of such a play $u$.
	So, $u \, a \in \Plays(\cGame, \varrho_\sigma)$ and $u \, b \in \Plays(\cGame, \varrho_\sigma)$. 
	This holds for \emph{every} play $u$ that solely consists of $\tau$-actions and is maximal w.r.t.\ them \myItem{(1)}. 
	Now, assume for contradiction and w.l.o.g.\ that $a$ is controllable. 
	We build a slightly modified controller $\varrho' = \{f^{\varrho'}_p\}_{p \in \processes}$ as follows:
	$$f^{\varrho'}_p(u) = f^{\varrho_\sigma}_p(u) - \{a\}$$
	$\varrho'$ behaves like $\varrho_\sigma$ but always forbids $a$.

	As strategy-equivalence only considers winning strategies and controllers, our main objective is to show that $\varrho'$ is winning.
	It holds that $\Plays(\cGame, \varrho') \subseteq \Plays(\cGame, \varrho_\sigma)$.
	This alone does not allow us to conclude that $\varrho'$ is winning. It could happen that $\varrho'$ blocks action $a$ and therefore blocks itself from reaching a winning configuration. 
	Suppose $u \in \Plays(\cGame, \varrho') \subseteq  \Plays(\cGame, \varrho_\sigma)$ is any sequence that consists only of $\tau$-actions and is maximal w.r.t.\ $\tau$-actions, i.e., there is no $\tau$ with $u \, \tau \in \Plays(\cGame, \varrho')$. 
	Outside from always rejecting $a$, $\varrho'$ behaves like~$\varrho_\sigma$. 
	As $us$ is maximal w.r.t.\ $\tau$-actions in  $\varrho'$, it is hence maximal w.r.t.\ $\tau$-actions for $\varrho_\sigma$, i.e., there is no $\tau$ with $u \, \tau \in \Plays(\cGame, \varrho_\sigma)$. 
	By \myItem{(1)}, we get that $u \, a \in \Plays(\cGame, \varrho_\sigma)$ and $u \, b \in \Plays(\cGame, \varrho_\sigma)$. While $\varrho'$ blocks $a$, we still conclude that  $u \, b \in \Plays(\cGame, \varrho')$. 
	Starting in $u\,b$, $\varrho'$ again behaves like $\varrho_\sigma$, since after $u \, b$ there cannot be an $a$-action  (by bisimulation). Therefore, there is no maximal play in $\Plays(\cGame, \varrho')$ that does not reach a winning configuration.
	%By blocking $a$ a winning state is never blocked since in any situation where $a$ is needed to make progress $b$ is also possible. 
	$\varrho'$ is winning.
	
	Since $\varrho'$ is winning, there exists a strategy $\sigma_{\varrho'}$ that is bisimilar to $\varrho'$. 
	We can derive an easy contradiction: Let $u \in \Plays(\cGame, \varrho')$ be any play that only consists of $\tau$-actions and is maximal w.r.t.\ them. It holds that $\init^{\pNet^\sigma} \relation u$. 
	By construction of $\varrho'$, we know that $u \, b$ is the only extension of $u$, i.e., no $\tau$ or $a$ is possible. 
	This is a contradiction since $\fireTranTo{\init^{\pNet^\sigma}}{a}{M}$ for some $M$ as this is possible for any strategy for $\pGame$. 
\end{proof}

\noindent We can show next that all $t_i$-actions are uncontrollable. The insight is that if some of them were controllable, \emph{all} processes can control them. There must be at least one process that can, from its local view, deduce whether $a$ or $b$ happened. This process can then base its decision of which $t_i$-actions to allow on the occurrence of $a$ or $b$. Such behavior cannot be achieved by a strategy for $\pGame$.  

\begin{lemma}
	$t_1, \cdots, t_n$ are uncontrollable.
	\label{lem:PtC_lb_uncon}
\end{lemma}
\begin{proof}
	Choose the strategy $\sigma$ for $\pGame$ as the one that allows everything, i.e., $\pNet^\sigma = \pGame^\mathfrak{U}$. It holds that $\Seq(\sigma) = \{ \epsilon, a, b \} \cup \{ a \, t , b \, t \mid t \in \{t_1, \cdots, t_n \} \}$.
	Let $\varrho_\sigma = \{f^{\varrho_\sigma}_p\}_{p \in \processes}$ be a bisimilar winning controller for $\cGame$ that exists by assumption.
	Assume for contradiction and w.l.o.g.\ that $t_1$ is controllable.\\
	We make an important observation: Any sequence of transitions in $\pGame$ or $\pNet^\sigma$ always begins with an $a$ or a $b$ followed by one of the $t_i$-actions. 
	Since $\varrho_\sigma$ and $\sigma$ are bisimilar we conclude that for any play in $\Plays(\cGame, \varrho_\sigma)$ the $t_i$-action must always follow the $a$ or $b$-action. 
	To put it differently: Assume there is a play $u = u' \, t_1 \in \Plays(\cGame, \varrho_\sigma)$, then there is an action $a$ or $b$ in $u'$. 
	At every point where $t_1$ is executed, there hence is some process in $\dom(t_1)$ that can derive the occurrence of $a$ or $b$ from its local view. The process might not be the same on every execution, but there is at least one at all times.\\
	We modify $\varrho_\sigma$ into a new controller $\varrho' = \{f^{\varrho'}_p\}_{p \in \processes}$ as follows:
	$$f^{\varrho'}_p (u) = \begin{cases}
	\begin{aligned}
	&f^{\varrho_\sigma}_p(u) \quad \quad &\text{if } a \not\in u\\
	&f^{\varrho_\sigma}_p(u) - \{t_1\} \quad &\text{if } a \in u
	\end{aligned}
	\end{cases}$$
	Here, $a \in u$ denotes that $a$ is an action in $u$. $\varrho'$ behaves just like $\varrho_\sigma$ with one difference: Whenever any process can deduce $a$ in its causal past, it forbids $t_1$. Since there is at least one process that can deduce $a$ or $b$ we can conclude that $\varrho'$ never admits a sequence where $t_1$ is played after a previous $a$.  \\
	Similar to \refLemma{unncessary}, the key is to argue that $\varrho'$ is a winning controller, i.e., blocking~$t_1$ never results in a state from which no winning configuration can be reached.
	It holds that $\Plays(\cGame, \varrho') \subseteq \Plays(\cGame, \varrho_\sigma)$.
	By the same reasoning as for the previous lemma, any maximal play $u$ where no $t_i$-actions have been played can be extended by all $t_i$-actions, i.e., $(u\, t_1), \cdots, (u\, t_n) \in \Plays(\cGame, \varrho_\sigma)$ (since $\sigma$ allows all $t_i$-transitions). 
	Since $\varrho'$ only forbids $t_1$ we conclude that $(u\, t_2), \cdots, (u\, t_n) \in \Plays(\cGame, \varrho')$.
	Blocking $t_1$ never blocks a winning state. Therefore, $\varrho'$ is winning. \\
	Since $\cGame$ and $\pGame$ are strategy-equivalent there is a strategy $\sigma_{\varrho'}$ for $\pGame$ that is bisimilar to $\varrho'$. Because of the bisimilarity it is easy to see that 
	$$\Seq(\sigma_{\varrho'}) = \{\epsilon, a, b\} \cup \{  a \, t_i \mid t_i \in \{t_2, \cdots, t_n\} \} \cup \{ b \, t_i \mid t_i \in \{t_1, \cdots, t_n\} \}$$
	Justified refusal forbids a strategy achieving this behavior, a contradiction. 
\end{proof}

\noindent We can now show that there are exponentially many global states needed in $\cGame$.
The idea is to simulate maximal $\tau$-sequences in the controller. The resulting \textit{global} state should allow exactly all the actions in any subset of the $t_i$-actions. Since we know that the $t_i$-actions are uncontrollable the fact that exactly certain actions are enabled from a global state must be a ``property'' of the global state, i.e., it cannot be the result of a controller simply forbidding some $t_i$-actions.

\begin{lemma}
	For every $\emptyset \neq B \subseteq \{t_1, \cdots, t_n\}$, there is a \emph{global} state $s_B = \{s_p\}_{p \in \processes}$ s.t.\ for the set of actions $E$ that can fire from $s_B$ ($E= \{a \in \Sigma \mid \{s_p\}_{p \in \dom(a)} \in \domain(\delta_a)  \}$), it holds that $E \cap \transitions = B$.
\end{lemma}
\begin{proof}
	Consider the (winning) strategy $\sigma$  s.t.\ $\Seq(\sigma) = \{ \epsilon, a, b \} \cup \{ a\,t , b\,t \mid t \in B \}$ and the bisimilar (winning) controller $\varrho_\sigma$. 
	We know that $a$ is a $\sigma$-compatible sequence of transitions ($\fireTranTo{\init^{\pNet^\sigma}}{a}{M}$ for some $M$). So by assumption, there is a $\varrho_\sigma$-compatible play $u$ with $u = \tau^* a$ and $M \relation u$.
	We now extend the play $u$ as long as possible with $\tau$-actions.
	We obtain a play $u' \in \Plays(\cGame_\pGame, \varrho_\sigma)$ s.t.\ $u'$ cannot be extended by another $\tau$-action. 
	We can guarantee the existence of such a $u'$ since $\varrho_\sigma$ is winning and therefore does not admit infinite plays. \\
	It holds that $M \relation u'$.
	Since $\sigma$ does allow from $M$ exactly the transitions in $B$, $M \relation u'$ holds, and $u'$ is maximal w.r.t.\ $\tau$-actions, we know that the possible extensions of $u'$ are exactly $B$. So 
	$$u' \, t_i \in \Plays(\cGame_\pGame, \varrho_\sigma) \Leftrightarrow t_i \in B$$ 
	Since all $t_i$-actions are uncontrollable (\refLemma{PtC_lb_uncon}), we get that $u' \, t_i \in \Plays(\cGame_\pGame, \varrho_\sigma)$ if and only if $u' \, t_i \in \Plays(\cGame_\pGame)$.
	So 
	$$u' \, t_i \in \Plays(\cGame_\pGame) \Leftrightarrow t_i \in B$$ 
	The global state $s_B = \state{u'}$ hence allows exactly the $t_i$-actions that are in $B$.
\end{proof}

\noindent Consequently, there must be exponentially many \emph{global} states. 
For every strategy-equivalent control game $\cGame$ with a constant number of players, there must hence be exponentially many local states. 

\begin{theorem}
	There is a family of Petri games $\{\pGame_n\}_{n \in \mathbb{N}}$ with $\abs{\transitions^{\pGame_n}} = n$ s.t.\ every strategy-equivalent control game (with an equal number of players) must have at least $\Omega(d^n)$ local states for $d > 1$.
	\label{theo:LB_1}
\end{theorem}
\begin{proof}
	There are $\Omega(2^n)$ many sets $B \subseteq \{t_1, \cdots, t_n\}$.
	By the previous lemma, any strategy-equivalent control game must hence have $\Omega(2^n)$ many global states.
	For any control game (asynchronous automaton) with two processes $p_1, p_2$, there are at most $\abs{S_{p_1}} \cdot \abs{S_{p_2}}$ many global states. 
	Hence, one of the two processes must have $\Omega((\sqrt{2})^n)$ many local states.
\end{proof}

This is the proof of \refTheo{theor3}.

\section{Artificial Deadlocks}
\label{sec:appAD}
We want to hinder a strategy from terminating early by using the newly introduced commitment sets. 
Recall our use of \Lightning-actions in \refSection{PGtoAA}:
We used uncontrollable actions to prohibit certain global configurations, namely non-deterministic ones, by moving the processes to a locked state and causing the system to lose. 
We pursue a similar approach by using (losing) transitions to prohibit configurations where the commitment sets are used to terminate early. 
The gray parts in \refFig{secReduction} describe this formalism. 
In contrast to the \Lightning-actions that were purely optional when considering deterministic strategies, the deadlock-detection mechanisms is inevitable to even allow for winning-equivalent translations.

We begin by formally defining artificial deadlocks:
Every place in $\pGame_\cGame$ (ignoring the $\badw_\DL^p$-places being introduced in this section) corresponds to a state in $\cGame$.
The correspondence is formalized by $\zeta$ as:
\begin{align*}
\zeta (s) &= s\\
\zeta ( \,(s, A)\, ) &= s
\end{align*}
We extend this definition to markings by defining for each marking $M$ a corresponding global state in the asynchronous automaton by: $\zeta (M) = \bigcup_{q \in M} \{\zeta(q) \}$.

A global state in a control game is called \emph{final}, if no further actions are possible once in that state. 
An \emph{artificial deadlock} now comprises a situation where $M$ is final even though the corresponding state $\zeta(M)$ could still act, i.e., is not final.
As $M$ is final, a strategy would be allowed to terminate in that marking even though the controller would still be required to keep playing. 
We want to hinder and penalize a strategy that reaches such a situation. 
To this extent, we equip $\pGame_\cGame$ with additional $\badw_{\DL}^p$-places that are marked as losing. 
We define the set of all artificial deadlocks by
$$\mathfrak{D}_\DL = \{ M \in \reach(\pGame_\cGame) \mid M \subseteq \places_\penv \; \land \; M \text{ is final} \; \land \; \zeta(M) \text{ is not final}\}$$
Note that every marking in $\mathfrak{D}_\DL$ contains only environment places, i.e., only places corresponding to chosen commitment sets\footnote{We remark that the definition of $\mathfrak{D}_\DL$ depends on the reachable markings in the very game we are just defining. Conceptually, we first construct the game without the deadlock-mechanism and afterwards add the gray parts corresponding to the deadlock-detection.  }. 
For every $M \in \mathfrak{D}_\DL$, we add a transition $t_{\DL}^M$ \myItem{(7)} that fires exactly from $M$ and moves every token to a losing place $\badw_{\DL}^p$ \myItem{(8)}.
It holds that 
$$ \pre{}{t_{\DL}^M} = M \;\quad \text{ and } \; \quad \post{}{t_{\DL}^M} = \{\badw_{\DL}^p \mid p \in \processes \}$$
Since all $\badw_{\DL}^p$-places are losing, a winning strategy has to guarantee that none of the $t_{\DL}^M$-transitions are enabled and therefore has to avoid all artificial deadlocks. 
We remark that this relies on the fact that all places in any marking $M \in \mathfrak{D}_\DL$ belong to the environment and are therefore unrestrictable by a strategy.
With the added deadlock detection a strategy cannot maneuver into a situation where it can terminate early, but is only allowed to end the game if the corresponding global state in $\cGame$ is final as well.

\section{Translating Control Games to Petri Games}
\label{sec:appsecondDir}

In this section, we show that our translated Petri game $\pGame_\cGame$ and the control game $\cGame$ are strategy-equivalent, i.e., prove \refTheo{theor4}.
As for the first translation, we can give an even stronger result by defining our bisimulation $\relation$ not on a concrete strategy and a concrete controller but relate markings in the unfolding with plays in the control game.
For our proofs, we assume that any winning strategy for $\pGame_\cGame$ always commits.

For our bisimulation, we pursue a similar approach as we did in Appendix~\ref{sec:firstDir}. 
There, we related a marking and a play if both describe the same situation, i.e., result from the same observable actions/transitions. 
To translate strategies and controllers, we showed that in related situations the local information of each player are identical to its counterpart in the other game and they are therefore able to copy the decisions of one another.
For our present translation, we also relate a marking $M$ and play $u$ if they describe the ``same situation''. 
More concretely, we relate $M$ and $u$ if the observable transitions in the causal past of $M$ (when identifying transitions with the corresponding actions) agree with $u$, i.e., if the poset structure of the transitions in the causal past of $M$ agree with the poset representation of~$u$. 
This captures the idea that in equivalent situations a strategy and a controller should act equivalently.

Having fixed $\relation$, we need to show that we can translate strategies and controller such that they allow bisimilar behavior from related situations.
%By construction of $\relatin$ the undeling autmataon and Petri net allow for equal moves. If $M \relation u$ and $u$ can be extended by some action than $M$ enables the same action if apprioate commitment sets are chosen.
As our commitment set constructions allows equal control possibilities, it remains to argue that the local informations of each player are preserved in related situations. 
To this extend, it is helpful to think of each token in the Petri game as a player moving along one slice and consider the game as a composition of slices. 
The game can be sliced such that every token resides only on places of exactly one process. 
A token hence takes part in precisely the transitions that correspond to actions in which the corresponding process takes part in.

As we did in our first translation, we overload notation and do not distinguish between transitions in a branching process and transitions in the original net. We are hence able to simulate sequences of original transitions in a branching process. Since our obtained Petri game is safe we again obtain a unique simulation.
We can note that for any reachable marking in the unfolding there is a one-to-one correspondence between places in the marking and processes. 
For marking $M$ and $p$, we define $\markpro{M}{p}$ as the unique place in $M$ with $\zeta(\lambda(\markpro{M}{p})) \in S_p$. 

\subparagraph{On the relation $\relation$}

We use the function $\zeta$ defined in the context of deadlock detection to map places or markings in $\pGame_\cGame$ to states or global states in $\cGame$.
Recall that $\pastt{C} = \{y \in \transitions \mid \exists x \in C, y \leq x \}$ are the transitions in the causal past of a set of places $C$.
The transitions in the causal past can either be $\tau$-transitions or of the form $(a, B, \{A_s\}_{s \in B})$. 
Since $\tau$-transitions are only added in our translation and $(a, B, \{A_s\}_{s \in B})$-transitions were only needed because of the restrictive synchronization primitives of Petri nets, we define a pointwise operation $\square$ that deletes $\tau$-transitions and maps transitions to the corresponding actions. 
\begin{align*}
\square \big(\tau_{(s, A)} \big) &= \epsilon\\
\square\big( (a, B, \{A_s\}_{s \tuplein B})\big) &=  a
\end{align*}
$\square$ can be seen as projection on the observable actions followed by a mapping to the underlying action. Here, $\epsilon$ denotes the deletion of an element. 

We can now express our informal ideas on $\relation$ properly:
\begin{tcolorbox}[colback=white]
	\centering
	$M \relation u$ iff $\square(\pastt{M}) = u$
\end{tcolorbox}

\noindent Note that we can express an equality between the causal past of a marking and a play by comparing the underlying poset representation in terms of isomorphisms. $\square(\pastt{M}) = u$ hence means that the poset of $\square(\pastt{M})$ and $u$ is equal. Note that both are labelled with~$\Sigma$.  
$\relation$ agrees with what we argued informally. Given some marking $M$ in the unfolding, $\square(\pastt{M})$ is the partially ordered set that describes the observable transitions in the causal past of $M$.
If this agrees with some play $u$ then $M$ and $u$ result from the same situation, i.e., they are reached on the same observable actions/transitions.

In our translation, we represent each local state as a place and add a transition $(a, B, \{A_s\}_{s \in B})$ exactly from preconditions that correspond to configurations from which $a$ can occur.
A transition in $\pGame_\cGame$ hence moves the tokens exactly as the corresponding actions would move the processes in $\cGame$.
We can hence see that related marking and play result in equally labelled configurations. 
\begin{lemma}
	If $M \relation u$ then $\zeta(\lambda[M]) = \state{u}$.
	\label{lem:AtoPsameState}
\end{lemma}
\begin{proof}
	By induction on the length of a totally ordered sequence of $\pastt{M}$ using the following three facts:
	$$\post{\pGame_\cGame}{(a, B, \{A_s\}_{s \tuplein B})} = \delta_a(B) \in \mathit{image}(\delta_a)$$ 
	$$\zeta(\pre{\pGame_\cGame}{(a, B, \{A_s\}_{s \tuplein B})}) = \zeta(\{(s, A_s) \mid s \in B\}) = B \in \domain(\delta_a)$$ 
	$$~\hspace{4.1cm}\zeta(\pre{\pGame_\cGame}{\tau_{(s, A)}} = \zeta(\post{\pGame_\cGame}{\tau_{(s, A)}}\hspace{4.1cm}\qedhere$$
\end{proof}

\subparagraph{Causal Information Flow}

The translated Petri game $\pGame_\cGame$ describes the global behavior of all players. 
It is nerveless helpful to view $\pGame_\cGame$ in terms of slices where one slice comprises all places added from the local states of one process. 
A token hence moves along one slice and thereby along the states of one process.
Each transition $(a, B, \{A_s\}_{s \in B})$ added in $\pGame_\cGame$ involves exactly the tokens that correspond to the processes that take part in $a$.

Now, we consider a marking $M$ and play $u$ s.t.\ $M \relation u$. By our definition, the observable transitions in the past of $M$ agree with $u$. 
We can observe that in both $M$ and $u$ the local information of a player can be seen as the minimal downward closed set that contains all transitions/actions where the player is involved in directly. 
This goes well with the idea that both game types rely on causal information. 
The partial order organizes the events in time. As each communication transmits everything, the entire previous execution is transmitted comprising all causally preceding events. 
The local view of a player hence comprises all transitions/actions it is involved in directly as well as all causally preceding ones resulting in a downward closed set.
%Recall that $\markpro{M}{p}$ is the unique place in $M \cap \inv{\lambda}[\inv{\zeta}[S_p]]$.
We can state:

\begin{lemma}
	\label{lem:identicalView}
	If $M \relation u$ and $p \in \processes$ then 
	$$\view_p(u) = \square(\pastt{\markpro{M}{p}}) = \view_p(\square(\pastt{\markpro{M}{p}}))$$
	%For any reachable marking $M$ in the unfolding of $\pGame_\cGame$ and place $q \in M$ that belongs to process $p$, i.e., $\zeta(\lambda(q)) \in S_p$, it holds that $\view_p(u) = \view_p(\square(\pastt{M})) = \view_p(\square(\pastt{q})) = \square(\pastt{q})$.
\end{lemma}
\begin{proof}
	Form $M \relation u$, it follows that $\square(\pastt{M}) = u$ where both are considered as partially ordered sets. \\
	We first consider $\view_p(u)$:	We already argued that in the poset representation, $\view_p(u)$ is the smallest downward closed subset of $u$ that contains all actions from $\Sigma_p$.\\
	Now consider $\pastt{\markpro{M}{p}}$:
	As $\pGame^\mathfrak{U}$ is an occurrence net, there is a unique transition $t \in \pre{}{\markpro{M}{p}}$. It holds that $\pastt{\markpro{M}{p}} = \pastt{t}$.
	The token on place $\markpro{M}{p}$ moves along the places that correspond to process $p$ and thereby takes part in all transitions that correspond to an action in $\Sigma_p$. 
	All transitions in $\pastt{\markpro{M}{p}}$ that correspond to an action in $\Sigma_p$ are hence causally related to $\markpro{M}{p}$. Since transition $t$ also corresponds to an action in $\Sigma_p$ and $\pastt{\markpro{M}{p}} = \pastt{t}$ we can characterize $\pastt{\markpro{M}{p}}$ as the smallest downward closed subset of $\pastt{M}$ that contains all transitions corresponding to actions from $\Sigma_p$. \\
	The minimal downward closed subset is unique. As both $\view_p(u)$ and $\pastt{M}$ describe the minimal downward closed subset containing all actions in $\Sigma_p$ and all transitions corresponding to actions in $\Sigma_p$, respectively, it holds that 
	$$ \view_p(u) = \square(\pastt{\markpro{M}{p}})$$
	We can then conclude $ \view_p(u) = \view_p(\square(\pastt{\markpro{M}{p}}))$ as $\view_p(\cdot)$ is idempotent. 
\end{proof}

\refLemma{identicalView} states that our definition of $\relation$ does not only capture the global configurations of both game types (as stated in \refLemma{AtoPsameState}) but preserves causal information. 
In by $\relation$-related situations $M$ and $u$, the causal past of the place in $M$ that belongs to some process $p$ ($\markpro{M}{p}$) agrees with the local view of $p$ on $u$ (after applying $\square$).  
Note that in general $\pastt{\markpro{M}{p}} \neq \pastt{M}$.

\begin{figure}[t!]
	\begin{tcolorbox}[colback=white, colframe=myYellow, arc=3mm, boxrule=1mm]
		We create places for the initial marking $\init^{\sigma_\varrho}$ and extend $\lambda$ s.t.\ $\lambda[\init^{\sigma_\varrho}] = \init^\pNet$.
		
		We then iterate:
		\begin{itemize}
			\item For every system place $q$ with no outgoing transitions: 
			$q$ belongs to a process, i.e., $\lambda(q) \in S_p$ for process $p$. \\
			Compute $u = \square(\pastt{q})$ in the already constructed strategy.
			
			We assume that $u$ is a $\varrho$-compatible play and $\statep{p}{u} = \lambda(q)$ \myItem{(A)}.
			Now define $\mathds{E} = f^\varrho_p(\view_p(u)) \subseteq \Sigma^{\sys}_p$.
			Add a transition $t'$ with $\lambda(t') = \tau_{(\statep{p}{u}, \mathds{E})}$, a place $q'$ with $\lambda(q') = (\statep{p}{u}, \mathds{E})$, and the flow s.t.\ $q \in \pre{\pNet^{\sigma_\varrho}}{t'}$ and $q' \in \post{\sigma_\varrho}{t'}$.\\
			Afterwards, continue with a new unprocessed marking.
			
			\item For every set of concurrent environment places $C$ with $\lambda[C] = \pre{\pGame_\cGame}{t}$ for some~$t$, we add a new copy of $t$ and places for the postcondition (if these nodes did not already exist).
		\end{itemize}
	\end{tcolorbox}
	
	\caption{Construction of strategy $\sigma_\varrho$ for $\pGame_\cGame$  that is build from controller $\varrho$ for $\cGame$}
	\label{fig:CtoP_translation1}
\end{figure}

\subsection{Translating Controllers to Strategies}

In this section, we provide a formal translation of strategies for $\pGame_\cGame$ to controllers for $\cGame$. Given a winning controller $\varrho$ for $\cGame$, we define a strategy $\sigma_\varrho$ for $\pGame_\cGame$. The strategy construction is depicted in \refFig{CtoP_translation1}.

$\sigma_\varrho$ does what we sketched informally. It is build incrementally. 
We start by creating a branching process that only contains the initial marking and incrementally add more and more places and transitions.
For every system place $q$ in the partially constructed strategy, we need to add an environment place that represents a commitment set. 
To decide which to choose, we apply $\square$ to the causal past of $q$, i.e., transform the transitions in the past to a trace of actions. 
There is a process $p$ that corresponds to $q$. 
The play obtained from the causal past is then given to the local controller of this process which decides for a set of controllable actions $\mathds{E}$.
Then, $q$ copies this decision by adding the commitment set that contains exactly these actions. 
As soon as our construction adds a new system place, we can hence choose a commitment set for that place. 
Apart from the $\tau$-transitions used to choose commitment sets, no observable transitions involve any system place.
To add them, we hence consider every set of pairwise concurrent places $C$ and add all transitions leaving from there. 
The construction proceeds by choosing commitment sets for every system place and afterwards adding all transitions possible from these commitment sets. 

\begin{lemma}
	$\sigma_\varrho$ is a deterministic, deadlock-avoiding strategy that always commits. 
	\label{lem:PtC_dlI}
\end{lemma}
\begin{proof}
	It is easy to see that the constructed net is a branching process of $\pGame_\cGame$.
	The only transitions that might not be added are local $\tau$-transitions. Since they are local and leave a system place we can refuse to add them without violating justified refusal.
	From any commitment set, the construction adds all transitions possible from this set, i.e., it does not restrict any transitions that can occur from the commitment sets. 
	The constructed $\sigma_\varrho$ is hence a strategy. 
	For every system place, we add exactly one commitment set. $\sigma_\varrho$ is therefore deterministic and always chooses a commitment set. 
	As $\sigma_\varrho$ always commits, it is also deadlock-avoiding. 
\end{proof}

\paragraph*{Strategy-Equivalence}

Having constructed $\sigma_\varrho$, we can prove it strategy-equivalent to $\varrho$.
For our bisimulation, we use the previously defined $\relation$ and restrict it to the reachable markings in $\pNet^{\sigma_\varrho}$ and plays compatible with $\varrho$. 
The previous lemmas (\refLemma{AtoPsameState} and \refLemma{identicalView}), established for the unfolding, extend to the restricted version.

The definition of $\sigma_\varrho$ and $\relation$ are, on their own, completely independent. 
\refLemma{identicalView}, however, established an important connection between both:
Assume that $M \relation u$. 
By definition, it holds that $\square(\pastt{M}) = u$. 
In our construction of the strategy, every system place $q$ in $M$ decides what commitment set to choose by constructing a play from its causal past and copying the decision of $\varrho$. 
According to \refLemma{identicalView}, the computation of $\view_p(\square(\pastt{q}))$ (as done in the definition of a strategy) agrees with $\view_p(u)$. 
We hence conclude that in $\relation$-related situations the places in $M$ copy the decision made by $\varrho$ on $u$. 
We can argue in both directions:
\begin{itemize}
	\item Suppose that there is a transition $(a, B, \{A_s\}_{s \in B})$ enabled in $M$. Transition $a$ is either controllable or uncontrollable. If uncontrollable, $u$ can be extended by $a$ as it is independent of the controller and the state reached on $u$ agrees with $\lambda[M]$ (\refLemma{AtoPsameState}). If controllable, all places $q$ involved in $(a, B, \{A_s\}_{s \in B})$ represent commitment sets that contain $a$. The commitment set of a place $q$ was chosen by computing $\view_p(\square(\pastt{q}))$ for the respective process $p$ and defining the commitment set as the set of actions allowed by the controller. By \refLemma{identicalView}, this is however identical to $\view_p(u)$.
	Since $a$ is included in the commitment set of all involved places we can deduce that $a$ must have been allowed by the controller of each involved process. We get that $u \, a \in \Plays(\cGame, \varrho)$.
	
	\item Suppose $u \, a \in \Plays(\cGame, \varrho)$. Then $a$ is either uncontrollable or controllable. We first move all tokens from $M$ to a commitment set to be able to execute observable transitions. We call this marking $M'$.
	In case of $a$ being uncontrollable, we can deduce that a transition corresponding to $a$ is possible from $M'$. 
	If $a$ is controllable we observe that all processes involved in $a$ allowed $a$. By \refLemma{AtoPsameState}, every place in $M'$ involved in $a$ has chosen its commitment set in accordance with the controller's decision on $u$. As all involved processes  allowed $a$, it is included in all commitment sets. Hence, there is a transition corresponding to $a$ enabled in $M'$. 
\end{itemize}

\noindent 
We can now give formal proofs: 
We begin by showing that it suffices to show that in $\relation$-related situations the same actions/transitions are possible. That is, if $M \relation u$ and we extend $M$ and $u$ by the same action/transition we obtain markings and plays that are again related\footnote{For the first translation, we did not need such a results as we defined $\relation$-directly in terms of firing the actions in the branching process of a strategy. Extending a related marking and play with the same action/transition hence automatically resulted in related situation. }. 

\begin{lemma}
	If $M \relation u$ and $\fireTranTo{M}{\tau^* (a, B, \{A_s\}_{s \in B}) \tau^*}{M'}$ for some $M'$ and $u' = u \, a$ then $M' \relation u'$.
	\label{lem:progress}
\end{lemma}
\begin{proof}
	From $M \relation u$, we conclude that $\square(\pastt{M}) = u$.
	All $\tau$-transitions are local, i.e., involve only one place and hence do not add any dependencies in the poset of $\pastt{M'}$. An action $a$ adds dependency, i.e., a causal relation, to all actions from processes in $\dom(a)$. As by construction, $(a, B, \{A_s\}_{s \in B})$ involves places of tokens that correspond to $\dom(a)$ it induces a dependency to the transitions that belong to actions where processes in $\dom(a)$ are involved in. It hence holds that $\square(\pastt{M'}) = u'$.
\end{proof}

\noindent We can now formally prove bisimilarity.
$\pGame_\cGame$ comprises additional $t_{\DL}^M$-transitions used to detect artificial deadlocks. 
We ignore them for our bisimulation proofs.
We later show that they are indeed never enabled if we construct $\sigma_\varrho$ from a deadlock-avoiding controller $\varrho$.

\begin{lemma}
	\label{lem:AtoPl1}
	If $M \relation u$ and $\fireTranTo{M}{ (a, \_, \_) }{M'}$ for some $M' \in \reach(\pNet^{\sigma_\varrho})$ then $u' = u \, a \in \Plays(\cGame, \varrho)$ and $M' \relation u'$.
\end{lemma}
\begin{proof}
	If $M \relation u$ then $\zeta(\lambda[M]) = \state{u}$ \myItem{(1)} (by \refLemma{AtoPsameState}).
	We first remark that $a$ is a possible extension of $u$ in the underlying game arena, i.e, $u \, a \in \Plays(\cGame)$.
	This follows from 
	$$\pre{\pNet}{(a, B, \{A_s\}_{s \in B})} \subseteq \lambda[M]$$
	$$\zeta(\pre{\pNet}{(a, B, \{A_s\}_{s \tuplein B})}) = B \in \domain(\delta_a)$$ 
	and \myItem{(1)}.
	We now show that $u' = u\, a \in \Plays(\cGame, \varrho)$. $M' \relation u'$ follows from \refLemma{progress}.
	We distinguish two cases:
	\begin{itemize}
		\item If $a \in \Sigma^{\env}$:
		Since $a$ is enabled and uncontrollable, and $u\, a \in \Plays(\cGame)$ it follows that $u\, a \in \Plays(\cGame, \varrho)$, i.e, $u \,a $ is a $\varrho$-compatible play.
		
		\item If $a \in \Sigma^{\sys}$:
		We know that $a$ is enabled in $\state{u}$.
		By our construction of $\pGame_\cGame$, because $a \in \Sigma^{\sys}$, and since $\fireTranTo{M}{ (a, B, \{A_s\}_{s \tuplein B}) }{M'}$, we can furthermore conclude that $a \in A_s$ holds for all $s \in B$ \myItem{(2)}.
		
		We assume for contradiction that $u\, a \not\in \Plays(\cGame, \varrho)$. Then there is a process $p \in \dom(a)$ s.t.\ $a \not\in f^\varrho_p(\view_p(u))$.
		We derive the contradiction by showing that the set of allowed actions agrees with one of the commitment sets in $M$ which, by \myItem{(2)}, contains $a$. 
		
		Since $(a, B, \{A_s\}_{s \in B})$ is enabled in $M$ and because of \myItem{(1)}, we conclude that for the place $q = \markpro{M}{p}$ it holds that $\lambda(q) = (\statep{p}{u}, A_{\statep{p}{u}})$. Since $p \in \dom(a)$, we conclude that $\statep{p}{u} \in B$ and from \myItem{(2)} we get that $a \in A_{\statep{p}{u}}$. The place that belongs to process~$p$ has chosen a commitment set that includes $a$.
		
		We can now observe how this commitment set was chosen. Let $q'$ be the predecessor (system) place of $q$, i.e., the place in the strategy from which we added $q$ as a committent set. When considering the construction of $\sigma_\varrho$, we observe that $q'$ computed what commitment set to choose by applying $\square$ to its causal past. It follows by definition: 
		\begin{align*}
		A_{\statep{p}{u}} &= f^\varrho_p(\view_p(\square(\pastt{q'})))\\
		&= f^\varrho_p(\view_p(\square(\pastt{q})))
		\end{align*}
		where the last equality holds since the causal past of $q$ ad $q'$ only differ by one $\tau$-transition. 
		
		By \refLemma{identicalView}, we get that 
		\begin{align*}
			\view_p(u) = \view_p(\square(\pastt{q}))
		\end{align*}
		We can conclude \begin{align*}
		A_{\statep{p}{u}} &= f^\varrho_p(\view_p(\square(\pastt{q})))\\
		&= f^\varrho_p(\view_p(u))
		\end{align*}
		The commitment set encoded in $q$ ($A_{\statep{p}{u}}$), i.e, the set place $q'$ decided to add, agrees with the decision of $p$ on $u$. 
		This is a contradiction to $a \in A_{\statep{p}{u}}$ and our assumption $a \not\in f^\varrho_p(\view_p(u))$.\qedhere
	\end{itemize}
\end{proof}

\begin{lemma}
	If $M \relation u$ and $\fireTranTo{M}{ \tau }{M'}$ then $M' \relation u$.
	\label{lem:AtoPl2}
\end{lemma}
\begin{proof}
	Obvious consequence from definition of $\relation$.
\end{proof}

\noindent Using the two previous lemmas, we already show that our assumption \myItem{(A)} in the strategy construction is justified: 
\begin{corollary}
	For any place $q \in \pNet^{\sigma_\varrho}$ that belongs to process $p$, $u = \square(\pastt{q})$ is a $\varrho$-compatible play and $\statep{p}{u} = \lambda(q)$.
\end{corollary}
\begin{proof}
	$q$ is part of some reachable marking $M$ and by \refLemma{AtoPl1} and \refLemma{AtoPl2} there is a play $u$ with $M \relation u$. 
	By \refLemma{identicalView}, it holds that $\square(\pastt{q}) = \view_p(u)$ and, as $\view_p(u)$ is a $\varrho$-compatible play, $\square(\pastt{q})$ is as well. 
	$\statep{p}{u} = \lambda(q)$ follows from \refLemma{AtoPsameState}.
\end{proof}

\begin{lemma}
	If $M \relation u$ and $u' = u \, a \in \Plays(\cGame, \varrho)$ then there exists $M' \in \reach(\pNet^{\sigma_\varrho})$ with $\fireTranTo{M}{ \tau^* \, (a, \_, \_) }{M'}$ and $M' \relation u'$.
	\label{lem:AtoPl3}
\end{lemma}
\begin{proof}
	Since $M \relation u$, we know that $\zeta(\lambda[M]) = \state{u}$ (by \refLemma{AtoPsameState}). 
	We first move every token in $M$ that resides on a system place to an environment place, i.e., move it to a chosen commitment set. 
	Since $\sigma_\varrho$ is by construction deterministic and always commits there is exactly one $\tau$-transition possible from every system place. So, $\fireTranTo{M}{ \tau^* }{M''}$ and  every token is on an environment place in $M''$. 
	It holds that $\zeta(\lambda[M'']) = \zeta(\lambda[M]) = \state{u}$. We furthermore know that for every $q \in M''$, $\lambda(q) = (s, A_s)$ holds, i.e., every place represents a commitment set. 
	For every local state $s \in \state{u}$, there hence is a commitment set $A_s$ s.t.\ there is a token on a place labelled $(s, A_s)$ \myItem{(1)}. \\
	Let $B = \{\statep{p}{u}\}_{p \in \dom(a)}$. Since $a$ can occur from $\state{u}$ (as $u \, a \in \Plays(\cGame, \varrho)$) we know that $B \in \domain(\delta_a)$.
	We now claim that there is a transition corresponding to $a$ possible from $M''$. As each transition explicitly encodes the configuration in $\domain(\delta_a)$ and commitment sets, we need the global state $B$ and the current commitment sets $A_s$ from \myItem{(1)} to ``design'' the transition. 
	We distinguish whether $a$ is controllable or uncontrollable:
	\begin{itemize}
		\item If $a \in \Sigma^{\env}$: 
		We consider the transition $t = (a, B, \{ A_s\}_{s \in B} )$ where $B$ and the $A_s$ are the state and sets from above. Such a transition exists as for uncontrollable actions transitions are added independent of the commitment sets.
		%Since there is a token on every place $(s, A_s)$ for $s \in B$				
		Because of \myItem{(1)}, we know that $t$ is enabled from $M''$, i.e., $\fireTranTo{M''}{ t }{M'}$ for some $M'$. $M' \relation u'$ follows from \refLemma{progress}.

		\item If $a \in \Sigma^{\sys}$: 
		As $u' = u \, a \in \Plays(\cGame, \varrho)$, we know that for every process $p \in \dom(a)$, it holds that $a \in f^\varrho_p(\view_p(u))$ \myItem{(2)}.
		We consider the transition $t = (a, B, \{ A_s\}_{s \in B} )$ where $B$ and the $A_s$ are the state and sets from above. Since $a$ is controllable such a transition must not necessarily exists. We show the existence by proving that $a \in A_s$ for every $s \in B$.
		
		Assume for contradiction that $a \not\in A_{s'}$ for $s' \in B$ (and $p$ is the process with $s' \in S_{p}$).
		We derive the construction by showing that the commitment set $A_{s'}$ agrees with the set of transitions allowed by $p$ after $u$ which, by \myItem{(3)}, contains $a$.
		
		Let $q = \markpro{M''}{p}$. It holds that $\lambda(q) = (s', A_{s'})$, which exists by \myItem{(1)}.
		We can study the construction of $\sigma_\varrho$ to see how the commitment set $A_{s'}$ was chosen. Let $q'$ be the predecessor of $q$, i.e., the place from which $q$ was added as a commitment set. 
		By construction, it holds that 
		\begin{align*}
		A_{s'} &= f_{p}(\view_{p}(\square(\pastt{q'})))\\
		&= f_{p}(\view_{p}(\square(\pastt{q})))
		\end{align*}
		where the last equality holds since the pasts of $q$ and $q'$ differ only by a $\tau$-transition. 
		From  \refLemma{identicalView}, we get
		\begin{align*}
		\view_{p}(u) = \view_p(\square(\pastt{q}))
		\end{align*}
		We conclude
		\begin{align*}
		A_{s'} &= f_{p}(\view_{p}(\square(\pastt{q})))\\
		&= f_{p}(\view_p(u))
		\end{align*}
		The commitment set $A_{s'}$ encoded in $q$ was added from $q'$ and agrees with the decision of $p$ made on play $u$. 
		This is a contradiction to $a \in f_{p}(\view_{p}(u))$  \myItem{(2)} and our assumption that $a \not\in A_{s'}$.	\\
		Since $t$ exists and because of \myItem{(1)}, we know that $t$ is enabled from $M''$, i.e., $\fireTranTo{M''}{ t }{M'}$ for some $M'$. Then, $M' \relation u'$ follows from \refLemma{progress}.\qedhere
	\end{itemize}
\end{proof}

\begin{corollary}
	$\varrho$ and $\sigma_\varrho$ are bisimilar.
\end{corollary}
\begin{proof}
	By definition, $\init^{\sigma_\varrho} \relation \epsilon$ holds. Since there are (unobservable) $\tau$-actions in $\cGame$ the statements follows from \refLemma{AtoPl1}, \refLemma{AtoPl2}, and \refLemma{AtoPl3}.
\end{proof}

\paragraph*{Deadlock-Avoidance} 

\noindent So far, we have established bisimilarity under the assumption that no deadlock detecting transition $t_\DL^M$ is enabled.  We now show that it this is a valid assumption. 
As discussed before, a transition $t_{\DL}^M$ is enabled if the strategy maneuvered into an artificial deadlock. $\sigma_\varrho$ simulates $\varrho$ and copies each decision of the controller. Since the commitment sets are chosen according to the controller an artificial-deadlock is never reached, as it would correspond to a deadlock of $\varrho$. 

\begin{lemma}
	If $\varrho$ is deadlock-avoiding then there are no $t_{\DL}^M$-transitions enabled in any reachable marking of $\pNet^{\sigma_\varrho}$.\label{lem:AtoPdl1}
\end{lemma}
\begin{proof} 
	Suppose there is a reachable marking $M$ and a $t_{\DL}^M$-transition enabled from $M$.
	By construction, $t_{\DL}^M$ only exists if $M$ is final. 
	Using bisimulation of $\sigma_\varrho$ and $\varrho$, we get a play $u \in \Plays(\cGame, \varrho)$ with $M \relation u$. 
	By bisimilarity, $u$ is maximal (since $M$ is final). 
	By construction of $t_{\DL}^M$, there is an action enabled from $\zeta(M)$ in the underlying automaton.
	By \refLemma{AtoPsameState}, it holds that $\zeta(M) = \state{u}$ and $\state{u}$ is not final.
	There hence is a play $u$ that is maximal w.r.t.\ $\varrho$ but ends in a non-final state $\state{u}$.
	This is a contradiction as deadlock-avoidance requires a controller to only terminate in final states. 
\end{proof}

\paragraph*{Winning Equivalence} 
Finally, we can show that $\sigma_\varrho$ is winning. Having already proved the bisimilarity this is rather easy. 

\begin{lemma}
	If $\varrho$ is winning then $\sigma_\varrho$ is winning.
\end{lemma}
\begin{proof}
	As noticed in \refLemma{PtC_dlI}, $\sigma_\varrho$ is deadlock-avoiding. \\
	Suppose there is a reachable marking $M$ in $\pNet^{\sigma_\varrho}$ that contains a bad place $q$. 
	By construction of $\pGame_\cGame$, it either holds that $\lambda(q) \in \bigcup_{p \in  \processes} S_p$ or $\lambda(q) = \bad_p$, i.e., the bad state must either be inherited from $\cGame$ or part of the deadlock-detection mechanism. 
	\begin{itemize}
		\item If $\lambda(q) \in \bigcup_{p \in  \processes} S_p$, i.e., $q$ is a place resulting from a bad state in $\cGame$:
		By bisimulation, there is a $u \in \Plays(\cGame, \varrho)$ with $M \relation u$. 
		It holds that $\zeta(\lambda[M]) = \state{u}$ (\refLemma{AtoPsameState}). By construction of $\pGame_\cGame$, $\state{u}$ must contain a bad place. A contradiction of the fact that $\varrho$ is winning. 
		
		\item If $\lambda(q) = \bad_p$, i.e., $q$ is a bad place added to detect artificial deadlocks. 
		Since $\varrho$ is by definition deadlock-avoiding, \refLemma{AtoPdl1} gives us that no $t_{\DL}^M$-transition is enabled in any reachable marking in $\pNet^{\sigma_\varrho}$. Hence, $\lambda(q) = \bad_p$ is not possible.\qedhere
	\end{itemize}
\end{proof}

\begin{proposition}
	\label{prop:int1}
	If $\varrho$ is a winning controller for $\cGame$ then $\sigma_\varrho$ is a winning, deterministic strategy for $\pGame$ and bisimilar to $\varrho$.
\end{proposition}

%\noindent Using \refLemma{AtoPdl1} we can provide an even stronger statement by claiming that for every deadlock-avoiding (not necessary winning) controller $\varrho$ for $\cGame$ there is a (not necessary winning) strategy $\sigma_\varrho$ for $\pGame_\cGame$ that is bisimilar to $\varrho$.

\subsection{Translating Strategies to Controllers}

In this section, we provide the formal translation of strategies for $\pGame_\cGame$ to controllers for $\cGame$.
Given a winning, deterministic strategy $\sigma$ for $\pGame_\cGame$ that always chooses a commitment set, we construct a winning controller $\varrho_\sigma = \{f^{\varrho_\sigma}_p\}_{p \in \processes}$ for $\cGame$. 
We refer to the fact that $\sigma$ always commits by $\star$.
The description of $\varrho_\sigma$ is depicted in \refFig{CtoP_translation2}.

\begin{figure}[t!]
	\begin{tcolorbox}[colback=white, colframe=myYellow, arc=3mm, boxrule=1mm]
		Given $p \in \processes$ and $u \in \Plays_p(\cGame)$,
		%The strategy simulates ~$u$ in the strategy incrementally.
		fix any linearization $u_0, \cdots, u_{n-1}$ of $u$.
		
		Define $M_0 = \init^\sigma$.
		
		For $i$ from $0$ to $n-1$, define $\hat{M}_i$ as the marking with $\fireTranTo{M_i}{\tau^*}{\hat{M}_i}$ and there is no $\tau$-transition enabled in $\hat{M}_i$.
		Since all $\tau$ are local and $\sigma$ is deterministic, $\hat{M}_i$ is unique.
		
		Check if there is an observable transition $t_i = (u_i, \_, \_ )$ enabled from $\hat{M}_i$.
		There is at most one such transition $t_i$. 
		\begin{enumerate}
			\item[\myItem{a)}] If such $t_i$ exists: Let $M_{i+1}$ be the resulting marking, i.e., $\fireTranTo{\hat{M}_i}{t}{M_{i+1}}$.
			Continue with $i+1$.
			\item[\myItem{b)}] If no such $t_i$ exists: Define $f^{\varrho_\sigma}_p(u) = \emptyset$ \myItem{(break)}.
		\end{enumerate}
		We iterate like this until the marking $\mathds{M}_u = \hat{M}_n$ is reached. 
		We observe that $\lambda({\markpro{(\mathds{M}_u)}{p}}) = (\_, \mathds{E})$ (because of $\star$). Define $f^{\varrho_\sigma}_p(u) = \mathds{E}$ .
	\end{tcolorbox}
	
	\caption{Definition of controller $\varrho_\sigma = \{f^{\varrho_\sigma}_p\}_{p \in \processes}$ constructed from  from strategy $\sigma$ for $\pGame_\cGame$. The figure depicts the description of each of the local controller $f^{\varrho_\sigma}_p$.}
	\label{fig:CtoP_translation2}
\end{figure}

The controller $\varrho_\sigma$ does what we argued informally. 
We depict the decision of a local controller $f_p^{\varrho_\sigma}$ for process $p$. 
Given some play $u$, the controller tries to simulate the actions in $u$ in the branching process of $\pNet^{\sigma}$. 
Since $\pGame_\cGame$ comprises additional local $\tau$-transitions the simulation needs to add them as well. After having played an action from $u$, $\varrho_\sigma$ hence simulates as many $\tau$-transitions as possible, i.e., moves every token to a place that corresponds to a chosen commitment set. 
It is easy to see that this is well-defined since all linearizations of a play result in the same marking: If two actions in trace are independent they correspond to distinct parts in the unfolding. The concrete order in which they are fired is hence irrelevant. 
The simulation of the actions in $u$ can fail, i.e., case \myItem{b)} can be reached. While we later show that if $u$ is a controller-compatible play the simulation always succeeds, we need to include \myItem{b)} to obtain a total function $f_p^{\varrho_\sigma}$. 
In case of a successful simulation, a marking $\mathds{M}_u$ is reached. It is easy to see that in this marking all tokens are on environment places, i.e., have chosen a commitment set. $\markpro{(\mathds{M}_u)}{p}$ is the place in this marking that corresponds to $p$. 
Process $p$ now copies the decision by allowing exactly the actions that are encoded in the commitment set of place $\markpro{(\mathds{M}_u)}{p}$, i.e., the decision of the player that corresponds to $p$. For later reference, we call this set of enabled actions $\mathds{E}$.

\paragraph*{Strategy-Equivalence} 

Having defined $\varrho_\sigma$, we can prove that it is bisimilar to $\sigma$.
We use the same relation $\relation$ and restrict it to the reachable markings in $\pNet^{\sigma}$ and the plays compatible with $\varrho_\sigma$. As before, the existing results (\refLemma{AtoPsameState} and \refLemma{identicalView}) extend to the restricted version.

\begin{lemma}
	If $M \relation u$, $p \in \processes$, and $\markpro{M}{p}$ is an environment place, i.e., corresponds to a chosen commitment set then: Simulating $\view_p(u)$ as in the definition of $\varrho_\sigma$ succeeds and yields a marking $\mathds{M}_{\view_p(u)}$ where $\markpro{M}{p} = \markpro{{(\mathds{M}_{\view_p(u)})}}{p}$
	\label{lem:ident_sim}
\end{lemma}
\begin{proof}
	From \refLemma{identicalView}, we get that $\square(\pastt{\markpro{M}{p}}) =  \view_p(u)$. 
	The actions in $\view_p(u)$ hence agree with the past of $q$ and since $\sigma$ is deterministic the simulation is deterministic as well. Thus, the simulation succeeds and yields a marking $\mathds{M}_{\view_p(u)}$.  
	$\markpro{M}{p} = \markpro{{(\mathds{M}_{\view_p(u)})}}{p}$ follows since the simulation fires exactly the transitions in the past from $\markpro{M}{p}$.
\end{proof}

\noindent \refLemma{ident_sim} is a trivial consequence from \refLemma{identicalView} that allows us to prove bisimilarity. 
It tells us that simulating the local view of a process results in a marking $\mathds{M}_{\view_p(u)}$ and the decisive point in this marking is shared with $M$. 
This allows us to conclude a connection between our definition of $\relation$ and our construction of $\varrho_\sigma$. 
In $\varrho_\sigma$, every process simulates its local view and, according to \refLemma{ident_sim}, copies the decision in a related marking. 

We can use \refLemma{ident_sim} to show bisimilarity.
We can reason in both direction: 
\begin{itemize}
	\item If $\fireTranTo{M}{ (a, B, \{A_s\}_{s \tuplein B}) }{M'}$, we can do a case analysis depending on whether $a$ is uncontrollable or not. If it is uncontrollable we immediately get that $u \, a \in \Plays(\cGame, \varrho_\sigma)$ since the underlying state reached on $u$ agrees with $\lambda[M]$ (\refLemma{AtoPsameState}).
	If $a$ is controllable we can deduce that $a \in A_s$ for all $s \in B$, i.e., every involved token has chosen a commitment set where $a$ is included. By \refLemma{ident_sim}, $\varrho_\sigma$ now simulates the local view of a process and thereby reaches a place in $M$. Since all tokens involved in $a$ have chosen a commitment set where $a$ is included, all processes involved in $a$ will allow $a$. So $u \, a \in \Plays(\cGame, \varrho_\sigma)$.
	
	\item If $u \, a \in \Plays(\cGame, \varrho_\sigma)$, we first move every token in $M$ to a commitment set which is always possible by $\star$. The new marking is $M'$. If $a$ is uncontrollable a transition corresponding to $a$ is possible from this commitment set combination. 
	If $a$ is controllable every involved process has allowed $a$. By construction, the processes decided what to allow by simulating their local view, which, according to \refLemma{ident_sim}, results in a place of $M'$. Since every involved process allows $a$ every involved place must have chosen a commitment set including $a$. We derive that a transition corresponding to $a$ is possible from $M'$. 
\end{itemize}

\noindent Since we assume that $\sigma$ is winning there can never be any $t_{\DL}^M$-transition enabled. We can hence neglect them for our bisimulation proofs. 

\begin{lemma}
	\label{lem:AtoPl4} 
	If $M \relation u$ and $\fireTranTo{M}{ (a, \_, \_) }{M'}$ for some $M' \in \reach(\pNet^{\sigma})$ then $u' = u \, a \in \Plays(\cGame, \varrho_\sigma)$ and $M' \relation u'$.
\end{lemma}
\begin{proof}
	Since $M \relation u$, we know that $\zeta(\lambda[M]) = \state{u}$ \myItem{(1)} (by \refLemma{AtoPsameState}). \\
	Because $(a, B, \{A_s\}_{s \tuplein B})$ is enabled in $M$, \myItem{(1)}, and our construction of transitions, we know that $a$ is enabled from $\state{u}$, i.e., $u \, a \in \Plays(\cGame)$. 
	We distinguish two cases:
	\begin{itemize}
		\item If $a \in \Sigma^{\env}$: Then, $u' = u \, a \in \Plays(\cGame, \varrho_\sigma)$ follows from the definition of control games.
		$M' \relation u'$ follows from \refLemma{progress}.
		\item If $a \in \Sigma^{\sys}$: Assume for contradiction that $u' = u \, a \not\in \Plays(\cGame, \varrho_\sigma)$.
		Then, there is a $p \in \dom(a)$ with $a \not\in f^{\varrho_\sigma}_p(\view_p(u))$.
		We derive the contradiction by showing that the set of allowed transitions by $p$ is exactly one of the commitment sets in $M$ which by assumption includes $a$.\\
		$\markpro{M}{p} \in M$ is the place that corresponds to process $p$.
		As this place is involved in $(a, B, \{A_s\}_{s \tuplein B})$ and $(a, B, \{A_s\}_{s \tuplein B})$ is enabled in $M$, we conclude that $\lambda(\markpro{M}{p})$ is an environment place, i.e., a chosen commitment set. 
		Because of \myItem{(1)}, we obtain that $\lambda(\markpro{M}{p}) = (\statep{p}{u}, A_{\statep{p}{u}})$.
		%So $\lambda(\markpro{M}{p}) = (s, A)$ for some local state $s$ (in fact $s = \statep{p}{u}$). 
		By construction of $\pGame_\cGame$ and since $(a, B, \{A_s\}_{s \tuplein B})$ is enabled, we get that $a \in A_{\statep{p}{u}}$.\\
		Let $\mathds{E} = f^{\varrho_\sigma}_p(\view_p(u))$ be this decision of process $p$ on $u$. 
		We can now study how $\varrho_\sigma$ came to this decision. 
		It does so by simulating $\view_p(u)$ in the branching process of~$\sigma$ and reaching a marking $\mathds{M}_{\view_p(u)}$ (by \refLemma{ident_sim} the simulation is successfully).
		By construction, $p$ then chooses $\mathds{E}$ as the set with $\lambda(\markpro{{(\mathds{M}_{\view_p(u)})}}{p}) = (\statep{p}{u}, \mathds{E})$. 
		That is, $\varrho_\sigma$ copies the decision of the corresponding place in $\mathds{M}_{\view_p(u)}$.\\
		From \refLemma{ident_sim} we now get that
		\begin{align*}\markpro{M}{p} = \markpro{{(\mathds{M}_{\view_p(u)})}}{p}\end{align*}
		This allows us to conclude that
		\begin{align*}
		(\statep{p}{u}, \mathds{E}) &= \lambda(\markpro{{(\mathds{M}_{\view_p(u)})}}{p})\\
		&= \lambda(\markpro{M}{p})\\
		&=(\statep{p}{u}, A_{\statep{p}{u}})
		\end{align*}
		We get that $f^{\varrho_\sigma}_p(\view_p(u)) = \mathds{E} = A_{\statep{p}{u}}$. The decision of what to enable ($f^{\varrho_\sigma}_p(\view_p(u))$) hence agrees with the commitment set of place $\markpro{M}{p}$ which is $A_{\statep{p}{u}}$.
		This is a contradiction to $a \in A_{\statep{p}{u}}$ and our assumption $a \not\in f^{\varrho_\sigma}_p(\view_p(u)) = \mathds{E}$. \\
		So, $u' = u \, a \in \Plays(\cGame, \varrho_\sigma)$ holds.
		$M' \relation u'$ follows from \refLemma{progress}.\qedhere
	\end{itemize}
\end{proof}

\begin{lemma}
	If $M \relation u$ and $\fireTranTo{M}{\tau}{M'}$ for some $M' \in \reach(\pNet^\sigma)$ then $M' \relation u$.
	\label{lem:AtoPl5}
\end{lemma}
\begin{proof}
	Obvious consequence from the definition of $\relation$.
\end{proof}

\begin{lemma}
	If $M \relation u$ and $u' = u \, a \in \Plays(\cGame, \varrho_\sigma)$ then there exists $M' \in \reach(\pNet^\sigma)$ with $\fireTranTo{M}{\tau^* (a, \_, \_)} {M'}$ and $M' \relation u'$.
	\label{lem:AtoPl6}
\end{lemma}
\begin{proof}
	Since $M \relation u$, we know that $\zeta(\lambda[M]) = \state{u}$ (by \refLemma{AtoPsameState}).
	We first move every token that resides on a system place to an environment one, i.e., to a place corresponding to a commitment set. Since $\sigma$ satisfies $\star$, i.e., always commits this is always possible. 
	So, $\fireTranTo{M}{ \kappa }{M''}$ for some $M''$ holds and there are no enabled $\tau$-transitions in $M''$. 
	It holds that $\zeta(\lambda[M'']) = \zeta(\lambda[M]) = \state{u}$.
	For every $q \in M''$, it holds that $\lambda(q) = (s, A_s)$, i.e., all tokens have chosen a commitment set. 
	%For every process $p$ it holds that $\lambda(\markpro{M''}{p}) = (\statep{p}{u}, A_{\statep{p}{u}})$ \textbf{(2)}.
	For every local state $s \in \state{u}$, there is a set $A_s$ such that there is a token on a place labelled $(s, A_s)$ \myItem{(1)}.\\
	Since $u \, a \in \Plays(\cGame, \varrho_\sigma)$ ,we know that $a$ can occur from $\state{u}$.
	We define $B = \{\statep{p}{u}\}_{p \in \dom(a)}$. Since $a$ can occur from $\state{u}$ (as $u \, a \in \Plays(\cGame, \varrho)$) we know that $B \in \domain(\delta_a)$.
	We now claim that there is a transition corresponding to $a$ possible from $M''$. As each transition explicitly encodes the configuration in $\domain(\delta_a)$ and commitment sets, we need the global state $B$ and the current commitment sets $A_s$ from \myItem{(1)} to ``design'' the transition. 
	We distinguish whether $a$ is controllable or uncontrollable. 
	\begin{itemize}
		\item If $a \in \Sigma^\env$:
		Consider transition $t = (a, B, \{A_s\}_{s \tuplein B})$ where $B$ is the global state from above and $A_s$ are the sets such that there is a token on $(s, A_s)$ \myItem{(1)}.
		By construction of $\pGame_\cGame,$ such a transition $t$ exists.
		We conclude that $t$ is enabled in $M''$ and, as $t$ involves only environment places, it is allowed by $\sigma$. So, there is a $M'$ with $\fireTranTo{M''}{t} {M'}$.
		$M' \relation u'$ is follows from \refLemma{progress}.
		
		\item If $a \in \Sigma^\sys$:
		We know that for every $p \in \dom(a)$, $a \in f^{\varrho_\sigma}_p(\view_p(u))$ \myItem{(2)}.
		We again consider the transition $t = (a, B, \{A_s\}_{s \tuplein B})$ where $B$ is the global state from above and $A_s$ are the sets such that there is a token on $(s, A_s)$ \myItem{(1)}.
		By construction of $\pGame_\cGame$, such a transition only exists if $a \in A_s$ for all $s \in B$.
		
		Assume for contradiction that $a \not\in A_{s'}$ for some $s' \in B$. 
		%There hence is a process $p$ with $a \not A_{\statep{p}{u}}$.
		Let $p$ be the process with $s' \in S_{p}$ (it holds that $\statep{p}{u} = s'$). 
		We derive the contradiction by showing that the set $A_{s'}$ is the set of actions allowed by $p$ on $u$ and $a$ must therefore, by \myItem{(2)}, be included.

		For $\markpro{M''}{{p}} \in M$, it holds that $\lambda(\markpro{M''}{{p}}) = (s', A_{s'})$.
		Let $\mathds{E} = f^{\varrho_\sigma}_{p}(\view_{p}(u))$ be the decision made by $p$. 
		We can study how $\varrho_\sigma$ came to this decision. It does so by simulating $\view_p(u)$ in the branching process of $\sigma$ and reaching a marking $\mathds{M}_{\view_{p}(u)}$. 
		$\mathds{E}$ is then, by construction, the set with $\lambda(\markpro{(\mathds{M}_{\view_{p}(u)})}{{p}}) = (\statep{p}{u}, \mathds{E})$.
		
		From \refLemma{ident_sim}, we get that 
		\begin{align*}\markpro{(\mathds{M}_{\view_{p}(u)})}{{p}} = \markpro{M''}{{p}}\end{align*}
		So, we can derive that
		\begin{align*}
		(s', A_{s'}) &= \lambda(\markpro{M''}{{p}})\\
		&=\lambda(\markpro{(\mathds{M}_{\view_{p}(u)})}{{p}})\\
		&= (\statep{p}{u}, \mathds{E})
		\end{align*}
		
		Therefore, $f^{\varrho_\sigma}_{p}(\view_{p}(u)) = \mathds{E} = A_{s'}$ holds, i.e., the commitment set $A_{s'}$ agrees with the decision of $p$ made on $u$. 
		This is a contradiction to $a \in f^{\varrho_\sigma}_{p}(\view_{p}(u))$ \myItem{(2)} and our assumption that $a \not\in A_{s'}$. 
		
		We conclude that $t$ exists and is enabled in $M''$. So, there is a $M'$ with $\fireTranTo{M''}{t} {M'}$.
		$M' \relation u'$ follows from \refLemma{progress}.\qedhere
	\end{itemize}
\end{proof}

\begin{corollary}
	$\sigma$ and $\varrho_\sigma$ are bisimilar
\end{corollary}
\begin{proof}
	By definition, it holds that $\init^\sigma \relation \epsilon$. Since there are no unobservable $\tau$-actions in~$\cGame$ the statements follows from \refLemma{AtoPl4}, \refLemma{AtoPl5}, and \refLemma{AtoPl6}.
\end{proof}

\paragraph*{Deadlock-Avoidance} 
We show that $\varrho_\sigma$ is deadlock-avoiding. 
By construction, $\varrho_\sigma$ allows exactly the actions that $\sigma$ has included in the commitment sets. To avoid all $t_{\DL}^M$-transitions, $\sigma$ has to choose commitment sets such that there is a transition possible if there is an action possible from the corresponding state in $\cGame$ (cf.\ Appendix~\ref{sec:appAD}). By copying commitment sets, there is no deadlock reachable.

\begin{lemma}\label{lem:AtoPdl2}
	If $\sigma$ is deadlock-avoiding and avoids $t_{\DL}^M$-transitions, $\varrho_\sigma$ is deadlock-avoiding.
\end{lemma}
\begin{proof}
	Assume for contradiction that $\varrho_\sigma$ is not deadlock-avoiding. Then, there exists a play $u \in \Plays(\cGame, \varrho_\sigma)$ that is maximal w.r.t.\ the controller that could be extended in the underlying automaton. 
	By bisimulation, there exist a reachable marking $M$ in $\pNet^\sigma$ with $M \relation u$. 
	Let $M'$ be the marking that results from $M$ by playing as many $\tau$-transitions as possible, i.e., where every token has chosen a commitment set. 
	Because $\sigma$ satisfies $\star$, every token is on an environment place in $M'$  (i.e., has chosen a commitment set).
	$M' \relation u$ holds.\\
	Since $u$ is maximal, $M'$ is final (by bisimulation). 
	Because $\sigma$ is deadlock-avoiding, $\lambda[M']$ is final as well.
	Since $u$ can be extended in the underlying automaton there exist an action $a$ that is enabled in $\state{u}$, i.e., $\state{u}$ is not final. By \refLemma{AtoPsameState}, it holds that $\zeta(\lambda[M']) = \state{u}$. 
	By construction of the deadlock detection mechanism, there hence is a transition $t_{\DL}^{\lambda[M']}$ that is enabled in $\lambda[M']$ and cannotbe averted by strategy. 
	This is a contradiction to the assumption.
\end{proof}

\paragraph*{Winning Equivalence}

Once we established the bisimilar behavior we can show that $\varrho_\sigma$ is indeed winning.

\begin{lemma}
	If $\sigma$ is winning, $\varrho_\sigma$ is winning.
	\label{lem:CtoP_winning}
\end{lemma}
\begin{proof}
	Since $\sigma$ is winning it is by definition deadlock-avoiding. It, furthermore, avoids all $t_{\DL}^M$-transitions so, by \refLemma{AtoPdl2}, $\varrho_\sigma$ is deadlock-avoiding.\\
	Suppose $u \in \Plays(\cGame, \varrho_\sigma)$ is a play that reaches a state that contains a bad state. 
	By bisimulation, there is reachable marking $M$ in $\pNet^\sigma$ with $M \relation u$.
	Let $M'$ be the marking where the last $\tau$-transition of every place is reversed, i.e., $M'$ is almost identical to $M$ but every token is on a system place. It holds that $M' \relation u$. (We only need to do this reasoning because bad places in $\pGame_\cGame$ are restricted to system places.)
	Using \refLemma{AtoPsameState}, we conclude that $\zeta(\lambda[M']) = \lambda[M'] = \state{u}$. So, $M'$ contains a bad place by construction of $\pGame_\cGame$. A contradiction to the fact that $\sigma$ is winning.
\end{proof}

\begin{proposition}
	\label{prop:int2}
	If $\sigma$ is a winning deterministic strategy for $\pGame_\cGame$ then $\varrho_\sigma$ is a winning controller for $\cGame$ and bisimilar to $\sigma$.
\end{proposition}

\refProp{int1} and \refProp{int2} give us both parts of our initial obligation of \refTheo{theor4}.

\section{Enforcing Commitment}
\label{sec:appEC}

As an example of why always committing and avoiding deadlocks is fundamentally different, we consider the control game in \refFig{CtoP_exampleChallange} (a) and the translated Petri game in (b) (ignoring all grayed out parts). Even though the control game has no winning controller, the Petri game has a winning strategy: The token in $A$ refuses to commit and the token in $D$ plays transition $c$ forever. We note that if the token in $A$ chooses a commitment set, even if it is the empty one, the uncontrollable $a$-transitions can occur, causing a loss. 

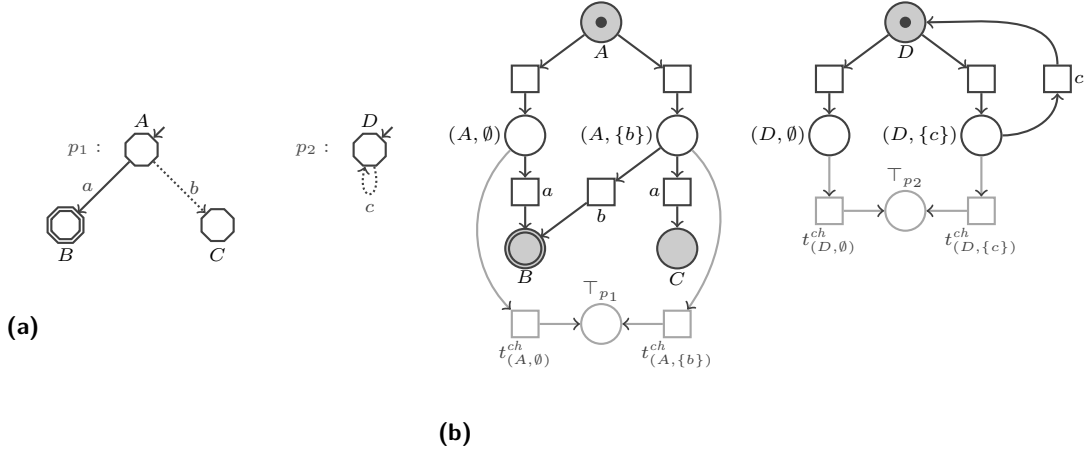
\begin{figure}[t!]
	\begin{subfigure}[c]{0.4\textwidth}
		\begin{center}
			\begin{tikzpicture}[scale=1.0,every label/.append style={font=\scriptsize}, label distance=-0.5mm]
			\node[aastate,label=north:$A$] at (0,0) (s1) {};
			\node[aastate,double,label=south:$B$] at (-1,-1) (s2) {};
			\node[aastate,label=south:$C$] at (1,-1) (s3) {};
			
			\node[aastate,label=north:$D$] at (3,0) (s4) {};
			
			\draw[arrow] (s1) --node[left=0mm]{\scriptsize$a$} (s2);
			\draw[arrow,densely dotted] (s1) --node[right=0mm]{\scriptsize$b$} (s3);
			
			\path[]
			(s4) edge [loop below,arrow,densely dotted] node {\scriptsize$c$} (s4);
			
			\draw[arrow] (s1)+(0.3,0.3) -- (s1);
			\draw[arrow] (s4)+(0.3,0.3) -- (s4);

			\node[] at (-0.75,0) (){\scriptsize \color{myDGray}$p_1:$};
			
			\node[] at (2.25,0) (){\scriptsize \color{myDGray}$p_2:$};
			\end{tikzpicture}
		\end{center}
		\subcaption{}
	\end{subfigure}
	\begin{subfigure}[c]{0.6\textwidth}
		\begin{center}
			\begin{tikzpicture}[scale=1.0,every label/.append style={font=\scriptsize}, label distance=-1mm]
			\node[sysplace,label=south:$A$] at (0,0) (p1) {};
			\node[envplace,label=west:{$(A, \emptyset)$}] at (-1,-1.5) (p2) {};
			\node[envplace,label=west:{$(A, \{b\})$}] at (1,-1.5) (p3) {};
			
			\node[sysplace,specialSys,label=south:$B$] at (-1,-3) (p4) {};
			\node[sysplace,label=south:$C$] at (1,-3) (p5) {};
			
			\node[transition] at (-1, -0.75) (t1){};
			\node[transition] at (1, -0.75) (t2){};
			
			\node[transition,label=east:{$a$}] at (-1, -2.25) (t3){};
			\node[transition,label=west:{$a$}] at (1, -2.25) (t4){};
			\node[transition,label=south:{$b$}] at (0, -2.25) (t5){};
			
			\node[token] at (0,0) (){};
			
			\draw[arrow] (p1) -- (t1);
			\draw[arrow] (p1) -- (t2);
			\draw[arrow] (t1) -- (p2);
			\draw[arrow] (t2) -- (p3);
			
			\draw[arrow] (p2) -- (t3);
			\draw[arrow] (p3) -- (t4);
			\draw[arrow] (p3) -- (t5);
			
			\draw[arrow] (t3) -- (p4);
			\draw[arrow] (t5) -- (p4);
			\draw[arrow] (t4) -- (p5);

			\node[envplace,draw=black!35,label=north:\color{myTGray}{$\top_{p_1}$}] at (0, -4) (ps){};
			\node[transition,draw=orange,draw=black!35,label=south:\color{myTGray}{$t^{\mathit{ch}}_{(A, \emptyset)}$}] at (-1, -4) (ts1){};
			\node[transition,draw=orange,draw=black!35,label=south:\color{myTGray}{$t^{\mathit{ch}}_{(A, \{b\})}$}] at (1, -4) (ts2){};
			
			\draw[arrow,draw=black!35] (p2) to[out=225,in=135] (ts1);
			\draw[arrow,draw=black!35] (p3) to[out=315,in=55] (ts2);
			\draw[arrow,draw=black!35] (ts1) -- (ps);
			\draw[arrow,draw=black!35] (ts2) -- (ps);
			
			%%%%%%
			
			\node[sysplace,label=south:{$D$}] at (4,0) (p1) {};
			\node[envplace,label=west:{$(D, \emptyset)$}] at (3,-1.5) (p2) {};
			\node[envplace,label=west:{$(D, \{c\})$}] at (5,-1.5) (p3) {};
			
			\node[transition] at (3, -0.75) (t1){};
			\node[transition] at (5, -0.75) (t2){};
			
			\node[transition,label=east:{$c$}] at (6, -0.75) (t3){};
			
			\node[token] at (4,0) (){};
			
			\draw[arrow] (p1) -- (t1);
			\draw[arrow] (p1) -- (t2);
			\draw[arrow] (t1) -- (p2);
			\draw[arrow] (t2) -- (p3);
			
			\draw[arrow] (p3) to[out=0,in=270] (t3);
			\draw[arrow] (t3) to[out=90,in=0] (p1);
			
			%%%
			\node[envplace,draw=black!35,label=north:\color{myTGray}{$\top_{p_2}$}] at (4, -2.5) (ps){};
			\node[transition,draw=black!35,label=south:\color{myTGray}{$t^{\mathit{ch}}_{(D, \emptyset)}$}] at (3, -2.5) (ts1){};
			\node[transition,draw=black!35,label=south:\color{myTGray}{$t^{\mathit{ch}}_{(D, \{c\})}$}] at (5, -2.5) (ts2){};
			
			\draw[arrow,draw=black!35] (p2) -- (ts1);
			\draw[arrow,draw=black!35] (p3) -- (ts2);
			\draw[arrow,draw=black!35] (ts1) -- (ps);
			\draw[arrow,draw=black!35] (ts2) -- (ps);
			\end{tikzpicture}	
		\end{center}
		\subcaption{}
	\end{subfigure}
	
	\caption{Control game $\cGame$ with safety as winning condition (a) and our translated Petri game $\pGame_\cGame$ (b). Uninteresting commitment sets for places $B$ and $C$ are omitted. The greyed parts model the addition of the challenge transitions in $\pGame_\cGame^\mathit{ch}$. $\pGame_\cGame$ (without the grayed out parts) has a winning strategy. $\pGame_\cGame^\mathit{ch}$ does not admit a winning strategy. }
	\label{fig:CtoP_exampleChallange}
\end{figure}

%\begin{proposition}
%The game with challenger $\pGame_\cGame^\mathit{ch}$ has a winning strategy iff $\pGame_\cGame$ has a winning strategy where every system place always chooses a commitment set.
%\end{proposition}

We can use the structure of $\pGame_\cGame$ to reduce local deadlock-avoidance to global one.
The idea is to terminate certain players. 
For every process, we add an additional place as a ``safe haven'', i.e., a place that is neither losing nor has any outgoing transitions. We allow every token that is on an environment place, i.e., a place representing a chosen commitment set, to move to this new place.
Every token that behaves as intended, i.e., always commits, can hence be moved to the safe haven and is therefore effectively removed from the game. 
All players that are locally deadlocked, i.e., refuse to commit, could previously do so since some player continued playing. As soon as all other players terminate, the locally deadlocked players do, however, cause a global deadlock since there no longer is a progressing player that can justify its refusal. If the second player in \refFig{CtoP_exampleChallange} corresponding to $p_2$ would be removed from the game the first player creates a global deadlock.
Every strategy where a token is locally deadlocked hence results in a (globally) deadlocked strategy and is therefore by assumption not winning. 
We successfully reduced local deadlock-avoidance to global deadlock-avoidance. 
Our reduction relies on the fact that Petri games are conceptually scheduled by an adversary scheduler. While the transition leading to a safe haven is always possible, it might not be executed. 
Following this high level explanation, we proceed by outlining the precise construction.

\paragraph*{Construction}

Formally, we modify $\pGame_\cGame$ into a new game $\pGame_\cGame^\mathit{ch}$.
For every process $p$, we add an additional place $\top_p$ to $\pGame_\cGame$. This place serves as the ``safe haven''. 
We allow every token to move to this place whenever it has chosen a commitment set. We hence define a new set of transitions 
$$\transitions^\mathit{ch} = \{ t^{ch}_{(s, A)} \mid s \in \bigcup_{p \in \processes} S_p \; \land \; A \subseteq \enabled{s} \cap \Sigma^{\sys}  \}$$
and add them to the game. 
There is exactly one such $t^{ch}_{(s, A)}$-transition for every environment place, i.e., every place with commitment set  $(s, A)$. 
We extend the flow such that these transitions fire from precisely the commitment set encoded in the transition by defining
$$\pre{\pGame_\cGame^{ch}}{t^{ch}_{(s, A)}} = \{ (s, A) \}$$
Every $t^{ch}_{(s, A)}$-transition moves the token of the involved process to the safe haven $\top_p$: 
$$\post{\pGame_\cGame^{ch}}{t^{ch}_{(s, A)}} = \{\top_p \mid \text{where } p \text{ is the process with } s \in S_p \}$$
Note that the precondition of all $t^{ch}_{(s, A)}$ comprises only environment places and can hence not be restricted by a strategy. 
In \refFig{CtoP_exampleChallange}, the added transitions and places are depicted in gray. From every commitment set, a token can always move to the $\top$-place.
In the modified game, the system no longer has a winning strategy, since the token in $D$ can be stopped at any point causing the token in $A$ to create a deadlock. 

\subparagraph{Correctness}

We can first observe that in any winning strategy $\sigma_{ch}$ for $\pGame_\cGame^\mathit{ch}$ every system place always commits. This follows directly from the construction: Suppose there is a contradicting situation, i.e., a marking $M$ in $\pNet^{\sigma_{ch}}$ where a system place $q \in M$ refuses to commit by $\post{\pNet^{\sigma_{ch}}}{q} = \emptyset$. We now consider one possible sequence starting in $M$: Every system place that can commit chooses a commitment set and afterwards terminates using a $t^{ch}$-transition. This results in a final marking $M'$ with $q \in M'$. However, $M'$ is a deadlock as in the underlying Petri net since $q$ could still progress to a commitment set place. 

Following this, we can argue that $\pGame_\cGame$ has a winning strategy that always commits if and only if  $\pGame_\cGame^\mathit{ch}$ has a winning strategy. \\
It is easy to see that any winning strategy $\sigma$ for $\pGame_\cGame$ that always commits results in a winning strategy $\sigma_{ch}$ for $\pGame_\cGame^\mathit{ch}$: The branching process of $\sigma_{ch}$ is just extended by all places and transitions introduced by our construction, i.e., from every commitment set place, an outgoing $t^{ch}$-transition is added. 
Since only $\top$-places are added there is no bad place reachable in $\sigma_{ch}$.
Since $\sigma$ always commits, firing one of the $t^\mathit{ch}$-transitions does not result in a deadlock, since every token can always move to a commitment set and afterwards either progress further or use a $t^\mathit{ch}$-transition to terminate.\\
A winning strategy $\sigma_{ch}$ for $\pGame_\cGame^\mathit{ch}$ results in a winning strategy $\sigma$ for $\pGame_\cGame$ that always commits: The branching process of $\sigma$ is obtained by removing all places and transitions that were added in the construction of $\pGame_\cGame^\mathit{ch}$.
The idea is that, while in $\sigma_{ch}$ the player can terminate early using a $t^{ch}$-transitions, there is also the possibility of it just playing as if there is no challenge transition. 
Since there is no bad place reachable in $\sigma_{ch}$ there are no bad places reachable in $\sigma$ either. 
Now, assume for contradiction that in $\sigma$ some place refused to commit in some marking $M$. Since $\sigma$ is obtained by removing parts of $\sigma_{ch}$ we get that $M$ is also a marking in $\sigma_{ch}$ but by the previous consideration this is not possible. 
Because every system player always commits, $\sigma$ is also deadlock-avoiding. 

\begin{proposition}
	$\pGame_\cGame^\mathit{ch}$ has a winning strategy iff $\pGame_\cGame$ has a winning strategy where every system place always chooses a commitment set.
\end{proposition}

\noindent We can use this result to justify the assumptions made in our correctness proofs since we can always modify $\pGame_\cGame$ to enforce commitment of all system players. 
We remark that $\pGame_\cGame^\mathit{ch}$ is not strategy-equivalent to $\cGame$ as a deadlock challenge can end a game even though the controller can continue to play. 
For every winning strategy for $\pGame_\cGame^\mathit{ch}$ however, there is an ``identical'' winning strategy for $\pGame_\cGame$ that itself is bisimilar to a controller for $\cGame$.
Even though  $\pGame_\cGame^\mathit{ch}$ and $\cGame$ are not strategy-equivalent, they are winning-equivalent. 

\section{Lower Bounds}
\label{sec:appLB2}

\begin{figure}[!t]
	\centering
	
	\begin{tikzpicture}[scale=1.0, every label/.append style={font=\tiny}, label distance=-1mm]
	\node[aastate] at (0,0) (s1){};
	\node[aastate] at (0,-2) (s2){};
	
	\draw[arrow] (s1)+(0.4, 0.4) to (s1);
	
	\draw[arrow] (s1) to[out=180,in=180] node[left=-1mm] {\tiny$x$} (s2);
	\draw[arrow, densely dotted] (s1) to[out=270,in=90] node[left=-1mm] {\tiny$a_1$} (s2);
	\draw[arrow, densely dotted] (s1) to[out=315,in=45] node[left=-1mm] {\tiny$a_2$} (s2);
	\draw[arrow, densely dotted] (s1) to[out=335,in=25] node[left=-1mm] {} (s2);
	%\draw[arrow, dotted] (s1) to[out=300,in=60] (s2);
	\draw[arrow, densely dotted] (s1) to[out=0,in=0] node[right=-1mm] {\tiny$a_n$} (s2);
	\end{tikzpicture}
	
	\caption{Control game family $\{\cGame_n\}_{n \in \mathbb{N}}$ where every strategy-equivalent Petri game (with an equal number of players) must be of exponential size. Action $x$ is uncontrollable. All other actions $a_1, \cdots, a_n$ are controllable.\vspace{-0.2cm}} 
	\label{fig:AtoC_LB}
\end{figure}
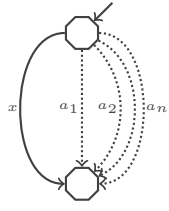

In this section, we give a family of control games s.t.\ every strategy-equivalent Petri game must be of exponential size (in the size of $\abs{\Sigma}$). 
In our translation, we had to duplicate actions into multiple transitions to overcome the restrictive communication scheme of Petri games. 
In our lower bound, we offer a Petri game to do the same, i.e., allow the same transition to occur from multiple distinct situations. 
In particular, our proof does not depend on the fact that any transition can only occur from a fixed precondition.

Consider the control game family $\{\cGame_n\}_{n \in \mathbb{N}}$ depicted in \refFig{AtoC_LB}. 
We fix $n$ and define $\cGame = \cGame_n$. 
The initial state has several outgoing controllable actions ($a_1, \cdots, a_n$) and one outgoing uncontrollable action ($x$). 
Let $\pGame$ be a Petri game that is strategy-equivalent to $\cGame$ and also contains only one player.
For our lower bound, we need the additional assumption that there are no infinite sequences of consecutive $\tau$-transitions possible in a winning strategy\footnote{There is in fact a Petri game that is strategy-equivalent to \refFig{AtoC_LB} of polynomial size (it permits possibly infinite $\tau$-transition sequences).}.
While winning strategies in reachability games never permit such infinite sequences, we need to assume it for safety games.
%Since $\tau$-transitions conceptually correspond to internal computations of a model (not related to observable behavior), infinite consecutive $\tau$-sequences would correspond to infinite internal computations. 
%As safety games require progress, infinite local computations are undesirable. 

\begin{lemma}
	\label{lem:CtoP_lb1}
	For every $\emptyset \neq B \subseteq \{a_1, \cdots, a_n\}$, there is a place $q_B$ such that $$\post{\pGame}{q_B} = B \cup \{x\}$$
\end{lemma}
\begin{proof}
	Choose $\varrho$ as the (winning) controller that allows exactly the controllable actions in~$B$ and $\sigma_\varrho$ as the bisimilar (winning) strategy for $\pGame$. Now, let $M = \{q\}$ be a marking that is reachable (in the strategy) by firing as many $\tau$-transitions as possible from the initial marking. 
	This marking may not be unique but, by assumption, it exists.\\
	Since $\varrho$ and $\sigma_\varrho$ are bisimilar and there are no $\tau$-transitions leaving $q$,  we can conclude that $\post{\sigma_\varrho}{q} = B \cup \{x\}$ and therefore $B \cup \{x\} \subseteq \post{\pGame}{\lambda(q)}$.\\
	We claim that $\lambda(q)$ is an environment place. Assume for contradiction that it is not, i.e., it is a system place. 
	We now modify $\sigma_\varrho$ by removing every $x$-transition leaving place $q$. Call this modified strategy $\sigma'$.
	Note that, since $q$ is by assumption a system place, the resulting branching fulfills justified refusal, i.e., is indeed a strategy.  
	It is, furthermore, easy to see that $\sigma'$ is still winning since, whenever a token is in place $q$, all other transitions in $B$ ($B \neq \emptyset$) are still possible and the behavior on them agrees with the behavior of the (winning) $\sigma_\varrho$. \\
	By assumption, there is a bisimilar winning controller $\varrho_{\sigma'}$ to $\sigma'$.
	This is an immediate contradiction: In $\sigma'$, the place $q$ is still reachable (using only $\tau$-transitions) so it holds that $\{q\} \relation \epsilon$. 
	We know that $x \in \Plays(\cGame, \varrho_{\sigma'})$ as $x$ is uncontrollable but we get that $x \not\in \post{\sigma'}{q}$ by construction of $\sigma'$. A contradiction to the bisimilarity of $\sigma'$ and $\varrho_{\sigma'}$.
	We hence know that $q$ is an environment place. \\
	As $\post{\sigma_\varrho}{q} = B \cup \{x\}$, $q$ is an environment place and there is only one player it follows that $\post{\pGame}{\lambda(q)} = B \cup \{x\}$. Now, $q_B = \lambda(q)$ has the desired properties. 
\end{proof}

This allows us to prove \refTheo{theor5}:

\begin{theorem}
	There is a family of control games $\{\cGame_n\}_{n \in \mathbb{N}}$ with $\abs{\Sigma_n} = n$ such that every strategy-equivalent Petri Game (with an equal number of players) must have at least $\Omega(d^n)$ places for $d > 1$.
\end{theorem}
\begin{proof}
	Follows from \refLemma{CtoP_lb1} with $d = 2$.
\end{proof}

\section{New Decidable Classes}
\label{sec:acyclicSliceNPcomplete}
\begin{lemma}
	Deciding whether a Petri net has an acyclic slice-distribution is \textsc{NP}-complete.
\end{lemma}
\begin{proof}
	It is easy to see that the problem is in NP, since for a given distribution it can be efficiently checked if it is valid and acyclic.  \\
	For hardness, we reduce from 3-SAT. Recall that 3-SAT is the problem of deciding whether a given propositional CNF formula, where each clause are exactly 3 literals, is satisfiable. This problem is known to be NP-hard. 
	Fix such a formula $\phi$ with propositional variables $x_1, \cdots, x_n$ and $\phi = (L_{1}^1 \lor L_{2}^1 \lor L_{3}^1) \land \cdots \land (L_{1}^m \lor L_{2}^m \lor L_{3}^m)$. 
	
	Our reduction relies on the fact that we can force two places to belong to the same slice. Consider the following structure:
	
	
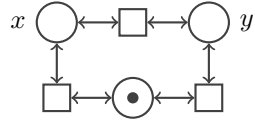
\begin{center}
		\begin{tikzpicture}[scale=1.0]
		\path 	(0,0) node[envplace,label=west:$x$](p1) {}
		(2,0) node[envplace,label=east:$y$](p2) {}
		(1,-1) node[envplace](p3) {}
		(1,0) node[transition](t1) {}
		(0,-1) node[transition](t2) {}
		(2,-1) node[transition](t3) {};

		\path	(1,-1) node[token](tt1) {};

		\draw[darrow] (p1) -- (t1);
		\draw[darrow] (p2) -- (t1);
		\draw[darrow] (p1) -- (t2);
		\draw[darrow] (p3) -- (t2);
		\draw[darrow] (p2) -- (t3);
		\draw[darrow] (p3) -- (t3);
		\end{tikzpicture}
	\end{center}

	It is easy to see that for this sub-net the distribution is acyclic only if nodes $x$ and $y$ belong to the same slice. Otherwise there will be at least one triangle (cycle) in the communication graph.
	By using this gadget, we can now define a Petri Net ($\pNet_\phi$) that is forced to create slices s.t.\ they exactly form a satisfying assignment.
	
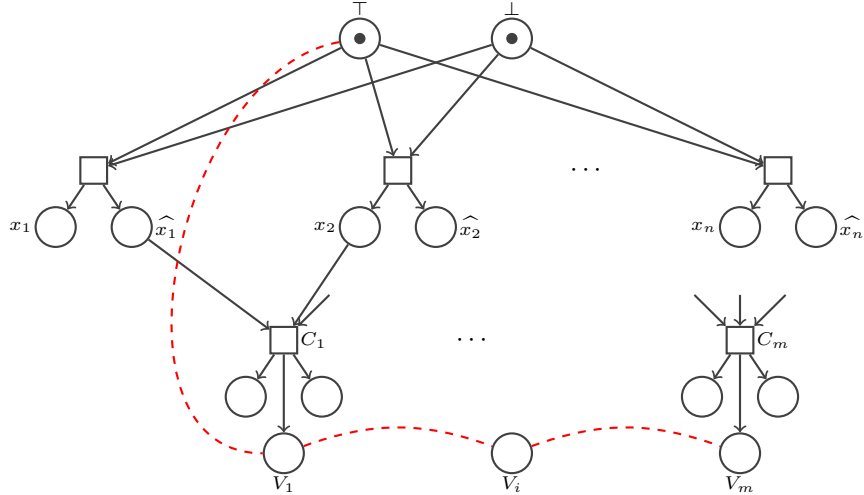
\begin{center}
		\begin{tikzpicture}[scale=1.0, every label/.append style={font=\scriptsize}, label distance=-1mm]
		\path 	(4,1) node[envplace,label=north:$\top$](top1) {}
		(6,1) node[envplace,label=north:$\bot$](bot1) {}
		
		(0,-1.5) node[envplace,label=west:$x_1$](x1) {}
		(1,-1.5) node[envplace,label=east:$\widehat{x_1}$](nx1) {}
		(0.5,-0.75) node[transition](tx1) {}
		
		(4,-1.5) node[envplace,label=west:$x_2$](x2) {}
		(5,-1.5) node[envplace,label=east:$\widehat{x_2}$](nx2) {}
		(4.5,-0.75) node[transition](tx2) {}
		
		(9,-1.5) node[envplace,label=west:$x_n$](xn) {}
		(10,-1.5) node[envplace,,label=east:$\widehat{x_n}$](nxn) {}
		(9.5,-0.75) node[transition](txn) {}
		
		(7,-0.75) node[](dot2) {$\cdots$}
		
		(3, -3) node[transition,label=east:$C_1$](c1) {}
		(9, -3) node[transition,label=east:$C_m$](cm) {}
		
		(2.5, -3.75) node[envplace](c1d1) {}	
		(3, -4.5) node[envplace,label=south:$V_1$](c1d2) {}
		(3.5, -3.75) node[envplace](c1d3) {}

		(8.5, -3.75) node[envplace](cmd1) {}	
		(9, -4.5) node[envplace,label=south:$V_m$](cmd2) {}
		(9.5, -3.75) node[envplace](cmd3) {}

		(5.5, -3) node[](cmmd2) {$\cdots$};	
		
		\path	(4,1) node[token]() {}
		(6,1) node[token]() {};
		
		\draw[arrow] (top1) -- (tx1);
		\draw[arrow] (bot1) -- (tx1);
		\draw[arrow] (top1) -- (tx2);
		\draw[arrow] (bot1) -- (tx2);
		\draw[arrow] (top1) -- (txn);
		\draw[arrow] (bot1) -- (txn);
		
		\draw[arrow] (tx1) -- (x1);
		\draw[arrow] (tx1) -- (nx1);
		\draw[arrow] (tx2) -- (x2);
		\draw[arrow] (tx2) -- (nx2);
		\draw[arrow] (txn) -- (xn);
		\draw[arrow] (txn) -- (nxn);

		\draw[arrow] (nx1) -- (c1);
		\draw[arrow] (x2) -- (c1);
		\draw[arrow] (c1)+(0.6,0.6) -- (c1);

		\draw[arrow] (cm)+(-0.6, 0.6) -- (cm);
		\draw[arrow] (cm)+(0,0.6) -- (cm);
		\draw[arrow] (cm)+(0.6,0.6) -- (cm);
		
		\draw[arrow] (c1) -- (c1d1);
		\draw[arrow] (c1) -- (c1d2);
		\draw[arrow] (c1) -- (c1d3);

		\draw[arrow] (cm) -- (cmd1);
		\draw[arrow] (cm) -- (cmd2);
		\draw[arrow] (cm) -- (cmd3);
		
		\node[envplace,label=south:$V_i$] at (6, -4.5) (helper){} ;
		
		\draw[-, thick, red, dashed] (top1) to[out=190, in=180] (c1d2);
		
		\draw[-, thick, red, dashed] (c1d2) to[out=20, in=160] (helper);
		\draw[-, thick, red, dashed] (helper) to[out=20, in=160] (cmd2);
		
		\end{tikzpicture}
	\end{center}

	For every variable, we create two new places $x_i$ and $\widehat{x_i}$ and connect them with a transition to both $\top$ and $\bot$. To get a valid slice, it is therefore enforced to put $x_i$ either in a slice containing $\top$ or in one containing $\bot$. For each clause, we then create a new transition $C_i$ whose precondition is exactly set such that it contains the three literals of the clause. For example, if $C_i = x_5 \lor \widehat{x_3} \lor x_8$ (all in the formula) then $\pre{}{C_i} = \{x_5,\widehat{x_3}, x_8\}$ (all in the net). For each of these transitions $C_i$, there are three outgoing places: two unimportant ones (only to stay concurrency-preserving) and one dedicated one, $V_i$. We will later see that this will enforce to put $V_i$ in the same slice with one of the literals that is in a slice containing $\top$. \\
	The red lines in the construction above indicate that these places should be in the same slice if the distribution is acyclic. We want all $V_i$s and the $\top$-place to be in the same slice. This can be done using the previously constructed gadget.
	
	The obtained Petri net has $2n + 4m + 1$ places and $n + 4m - 3$ transitions and can be computed efficiently.
	We can now prove that $\phi$ is satisfiable if and only if $\pNet_\phi$ has an acyclic distribution:\\
	$\Rightarrow$: Let $\phi$ be satisfiable and $h : \{x_1, \cdots, x_n\} \to \{0, 1\}$ be a satisfying assignment. We then select $\places^\top = \{\top\} \cup \bigcup\limits_{i = 1, \cdots, n} \left\{ \begin{aligned}
	x_i \quad &\textrm{if } h(x_i) = 1\\
	\widehat{x_i} \quad &\textrm{if } h(x_i) = 0
	\end{aligned} \right\} \cup \bigcup\limits_{j = 1, \cdots, m} V_j$ and $\places^\bot$ as the remaining places. It is easy to see that this is a valid slice when $h$ is satisfying: In this case for each clause, one of the three literal places is in $\places^\top$ and so we can put $V_i$ in $\places^\top$ as well. Since all $V_i$ are in the same slice this is also an acyclic distribution. \\
	$\Leftarrow$: Now, suppose $\pNet_\phi$ has an acyclic distribution. It is easy to see that $x_i$ and $\widehat{x_i}$ cannot be in the same slice. For a valid acyclic distribution, all $V_i$-places must be in the same slice as $\top$. Using these facts, it is easy to see that all literal nodes that are in the same slice with $\top$ form a satisfying assignment (formally $h : \{x_1, \cdots, x_n\} \to \{0, 1\}$ with $h(x_i) = \begin{cases}
	1 \quad \textrm{if } x_i \in \places^\top\\
	0 \quad \textrm{if } x_i \in \places^\bot
	\end{cases}$). 
	There cannot be complement literals in this slice and from each clause at least one literal must be in the $\top$ slice (otherwise $V_i$ cannot be in the $\top$ slice).
	Hence $\phi$ is satisfiable. 
\end{proof}

\end{document}